\documentclass{lmcs}
\pdfoutput=1
\usepackage[utf8]{inputenc}

\usepackage{lastpage}
\lmcsdoi{18}{3}{18}
\lmcsheading{}{\pageref{LastPage}}{}{}%
{Feb.~01,~2019}{Aug.~09,~2022}{}

\usepackage[all]{xy}

\SelectTips{eu}{}
\newdir{ >}{{}*!/-3pt/\dir{>}}
\def\mono{\ar @{ >->}}
\def\epi{\ar @{->>}}
\def\monoInline#1#2#3{\xymatrixcolsep{1.3pc}\xymatrix@1{#1:#2\mono[0,1]&#3}}

\usepackage{amsbsy,mathrsfs,latexsym,amssymb,mathbbol,mathabx,style}
\usepackage{bussproofs}
\usepackage{stmaryrd}

\theoremstyle{definition}

\usepackage{xcolor}

\newcommand{\img}{\mathsf{img}}

\begin{document}
\title{Moss' logic for ordered coalgebras}
\author{Marta B\'{\i}lkov\'{a}\lmcsorcid{0000-0002-3490-2083}}[a]
\address{The Czech Academy of Sciences, Institute of Computer Science, and Faculty of Arts, Charles University, Prague, Czech Republic}
\email{bilkova@cs.cas.cz}
\author{Mat\v{e}j Dost\'{a}l\lmcsorcid{0000-0002-4373-0471}}[b]
\address{Faculty of Electrical Engineering, Czech Technical University in Prague, Czech Republic}
\email{dostamat@fel.cvut.cz}
\thanks{Marta B\'{\i}lkov\'{a} acknowledges the support by the grant No.~P202/11/P304 of the Czech Science Foundation, and (for the revised version of the paper) by the grant From Shared Evidence to Group Attitudes of Czech Science Foundation and DFG, no.\ 16-07954J, and RVO: 67985807.
M.~Dost\'{a}l acknowledges the support by the Grant Agency of the Czech Republic under the grant 19-0092S
	and the grant 22-02964S
}

\begin{abstract}
We present a finitary version of Moss' coalgebraic logic for $T$-coalgebras, where $T$ is a locally monotone endofunctor of the category of posets and monotone maps. The logic uses a single cover modality whose arity is given by the least finitary subfunctor of the dual of the coalgebra functor $T^\partial_\omega$, and the semantics of the modality is given by relation lifting. For the semantics to work, $T$ is required to preserve exact squares. 
For the finitary setting to work, $T^\partial_\omega$ is required to preserve finite intersections. We develop a notion of a base for subobjects of $T_\omega X$. This in particular allows us to talk about the finite poset of subformulas for a given formula. 
The notion of a base is introduced generally for a category equipped with a suitable factorisation system.

We prove that the resulting logic has the Hennessy-Milner property for the notion of similarity based on the notion of relation lifting. We define a sequent proof system for the logic, and prove its completeness.
\end{abstract}

\maketitle

\EnableBpAbbreviations

\section{Introduction}
\label{sec:intro}

\changenotsign

Shortly after the theory of coalgebras emerged as a useful conceptual tool for a uniform study of various kinds of dynamic systems, there has been an interest in finding logics expressive enough to describe the behaviour of the coalgebras up to bisimilarity.

For the case of coalgebras for endofunctors $T: \Set \to \Set$ preserving weak pullbacks, this goal has been achieved by Larry Moss in his pioneering paper~\cite{moss:cl}, where he introduces a coalgebraic logic expressive for bisimulation, uniformly in the choice of the coalgebra functor $T$. The language of the logic contains a single cover modality $\nabla_T$, whose arity is given by the coalgebra functor $T$, and whose semantics is given by relation lifting. 
The original Moss' logic is infinitary: not only it allows for infinite conjunctions and disjunctions, but also modal formulas $\nabla\alpha$ for $\alpha \in T\Lang$ can be infinitary, depending on the functor $T$. 

After that, attention turned to the question how to obtain expressive coalgebraic languages based on more standard modalities~\cite{Jac01,Kurz01}, the current mainstream approach being based on modalities described by predicate liftings~\cite{Patt03,Sch08}.  

Moss' logic itself has soon become a lively field of study. Using similar ideas as Moss, Baltag in~\cite{baltag00} describes an infinitary coalgebraic logic capturing both the notion of bisimulation and simulation for $\Set$ coalgebras. Neither Moss nor Baltag provide a complete axiomatization for their respective infinitary logics, as completeness might not be available in general, depending again on the coalgebra functor.

A finitary version of Moss' logic has consequently been explored by various authors, resulting in an almost full picture, including (among others)
its axiomatization and completeness proof parametric in the coalgebra functor~\cite{kkv12}, a structural Gentzen-style proof
theory~\cite{bpv13}, applications to automata
theory~\cite{kupk:clos05,kupk:coal08,kissvenema}, and applications
to fixpoint logics~\cite{santocanale+venema,venemafp}. The two formalisms, languages based on predicate liftings and Moss' finitary languages, have been compared in detail in~\cite{KL12}.


Much of the theory that has been built around Moss' idea uses heavily that endofunctors of $\Set$ are well understood.
This is not quite the case for the endofunctors on the categories of preorders and posets. One encounters them naturally when interested in similarity alone, rather then bisimilarity~\cite{kkvpos:12,levy}. For this reason (among others), poset extensions and liftings of set functors were investigated in~\cite{balan+kurz,balan+kurz+velebil}. Also a coalgebraic logic for poset coalgebras, based on monotone predicate liftings, has been considered and proved to be expressive in~\cite{kkvpos:12}.

It therefore seems natural to investigate the possibilities of extending the above mentioned techniques and results on finitary Moss' logic for set coalgebras beyond the category of sets. This paper is a step in this direction, developing Moss' logic in the the category $\Pos$ of posets and monotone maps,
building mainly on results obtained in~\cite{bkpv11}. Namely, we use the
existence of a functorial relation lifting for functors preserving exact squares
in the category $\Pos$. This provides us, as we
show below, with the technical background for the development of Moss' logic for poset coalgebras. 

Since the results about relation lifting
in~\cite{bkpv11} were further generalised in~\cite{bkpv12} to the enriched
case of $\V$-categories, where $\V$ is a commutative quantale, it is also possible to define Moss' logic in this level of generality. However, the commitment to the enriched setting in particular means that everything, including the language, forms a $\V$-category and if we are to be truly general, we deal with a highly unusual syntax. We feel that restricting to the case of preorderes ($\V$-categories where $\V = \bTwo$) or posets is the right level of enrichment for presentation of Moss' coalgebraic logic in full detail at the moment.

The commitment to the enriched setting of the category $\Pos$ has similar consequences worth mentioning out front. First, every object we work with is going to be a poset, including a poset of atomic propostions or the poset of formulas. Arities of connectives and of the modality will be given by finitary poset functors. This alone opens possibilities to a rather unusual definition of syntax. Second, the usual notion of finite objects will be replaced by that of finitely generated objects. This affects notions of formulas, sequents and proofs.

\paragraph{Contribution of this paper.} We present a {\em finitary\/}
Moss' coalgebraic language for coalgebras for a locally monotone functor $T:\Pos\to\Pos$ that preserves exact squares.
The logic is based on the language of logic of distributive
lattices equipped with a single cover modality $\nabla$ (and for the sake of proof theory later also with its dual modality $\Delta$). The arity of both cover modalities is the least finitary subfunctor $T^\partial_\omega$ of 
the dual of the coalgebra functor $T$, and the semantics of the cover modalities essentially uses the notion of relation lifting, proven to exist for our choice of coalgebra functors in~\cite{bkpv11}.

For the finitary setting to work, we need to adopt an appropriate notion of a \emph{base}. 
In the ``classical'' case of a finitary coalgebraic functor $T:\Set\to\Set$, the base is
a natural transformation from $T$ to the finitary powerset functor. In the context of Moss' coalgebraic language, it in particular allows us to consider subformulas of objects in $T\Lang$. 
Bases of endofunctors, in a more general categorial setting, have recently been studied in the context of reachability in automata (see e.g.\ the approach taken in~\cite{BKR19,reachability19}).
We have independently developed an appropriate notion of base of a {\em locally monotone\/}
finitary functor $T_\omega:\Pos\to\Pos$. Some additional assumptions on the category
$\Pos$ and the functor $T_\omega$ seem inevitable. First of all, one has to choose
a suitable {\em factorisation system\/} $(\EE,\MM)$ on the category $\Pos$
and the functor $T$ should behave well with regard to the factorisation system. 
We prove that if we equip $\Pos$ with the factorisation system $(\EE,\MM)$ of monotone surjections and order embeddings,
then a base can be computed for every finitary functor $T_\omega$
that preserves order embeddings and their finite intersections.
In particular, bases can be used to produce, for each subobject of the poset $T_\omega X$,
a {\em finite\/} poset of its ``generators'', e.g., a poset of subformulas of a formula of
arity $T_\omega$, or a poset of successors of a state in a coalgebra for $T_\omega$.

The resulting finitary Moss' logic has the {\em Hennessy-Milner property\/} --- it is
expressive for a notion of simulation and similarity based on the relation lifting. This notion of similarity coincides with the notion of similarity given, e.g., in~\cite{worr,hjsim:04,balan+kurz,levy}.
The result matches the similar result for Moss' logic for coalgebras on the category of sets and bisimilarity.
It can also be seen as a counterpart to the result proved in~\cite{kkvpos:12} for positive coalgebraic logics
in the category of posets, stating that the logic of all monotone predicate liftings is expressive
for any endofunctor of posets that satisfies our conditions.

We present the resulting logic in a form of a cut-free two-sided sequent proof calculus, and we prove its completeness. To be able to define such a proof system, we essentially use a dual cover modality delta and its semantical relation to the nabla modality. This part of the paper very closely follows previous work of the first author~\cite{bpv13} on proof theory of classical Moss' logic in $\Set$.

\paragraph{Comparison to related work.}
The work we present in this paper is firmly rooted in previous work on various Moss' logics. From the original work~\cite{moss:cl} of Moss it takes the idea of a modality whose arity is given by the coalgebra functor and whose semantics by the relation lifting. There are differences worth mentioning: Moss' language does not use propositional variables, and it is inherently infinitary. Moss does not provide any proof calculus for the logic. In this sense, our work is much closer to the finitary version of Moss' logic studied in detail in~\cite{kkv12}, and can be understood as the poset-based version of this logic.~\cite{kkv12} addressed an important open problem at the time and provided a sound and complete derivation system for the finitary Moss' logic in terms of a Hilbert-style calculus parametric in the coalgebra functor, and proven completeness in a one-step manner.  We differ from this approach and opt to present the logic in terms of a cut-free sequent calculus. In this we closely follow previous work by the first author on Gentzen style proof theory for Moss' logic of Kripke frames~\cite{bpv08}, and consequently proof theory for general Moss' logic parametric in the coalgebra functor~\cite{bpv13}. One difference is that~\cite{bpv13} assumes the coalgebra functor to be finitary, while we, following~\cite{kkv12}, do not do so and use finitary functors on the syntactic side of matters only. In Section~\ref{sec:conclusion} we briefly describe how our approach relates to positive fragments of logics covered in~\cite{bpv13}.

One possible motivation to move from $\Set$ to $\Pos$ is to study similarity. This can be done in $\Set$~\cite{baltag00} in case one is interested in studying bisimilarity and similarity together (first one is represented by equality on the final coalgebra, while the other by a preorder on the final coalgebra). However, as argued by Levy~\cite{levy}, if we are exclusively interested in similarity, we would want the universe of final coalgebra to be a poset: if two nodes are mutually similar, they should be equal.  We employ the notion of similarity given by relation lifting of monotone relation between posets studied in~\cite{bkpv11}. This coincides with the notion of similarity given elsewhere in literature, in particular in~\cite{worr,hjsim:04,balan+kurz,levy}.

Comparing with the logic for similarity developed in Baltag's paper~\cite{baltag00}, we can, not surprisingly, find striking resemblances. 
The functorial relation lifting (called strong relator in~\cite{baltag00}) is weakened to capture similarity. For example, while the strong relator extending the powerset funstor yields the lifting pattern of the Egli-Milner lifting,
the relator $(\PP,\subseteq)$ extending the powerset functor yields the same lifting pattern as the $\LLL$-relation lifting by the lowerset functor developed in~\cite{bkpv11} and used in this paper, namely, the first half of Egli-Milner lifting. Similarly, the relator $(\PP,\supseteq)$ yields the same lifting pattern as the $\UU$-relation lifting by the upperset functor, namely, the second half of Egli-Milner lifting. Similarly, the cover operator in~\cite{baltag00} corresponds the the modality $\nabla_\PP$ of Moss, and can be somewhat related the the modality $\nabla_{\PP^c}$ for the convex powerset functor of this paper,
the box operator in~\cite{baltag00} can be related to the modality $\nabla_{\LLL}$ of this paper, and the diamond operator can be related to the modality $\nabla_\UU$ of this paper. The similarities are not surprising because the $\Set$ powerset functor $\PP$ and the $\Pos$ functors $\LLL,\UU, \PP^c$ can meaningfully be related~\cite{balan+kurz}: while the finitary convex powerset functor is the canonical extension (the posetification) of the powerset functor, the functors  $\LLL,\UU$ can be obtained as quotients of certain liftings of the powerset functor to the category of preorders and monotone maps. 
However, the obvious differences between~\cite{baltag00} and this paper are that the logic of this paper is finitary, genuinely developed in $\Pos$, and equipped with a sound and complete derivation system.

As for the techniques applied in this paper, we rely on: (1) the results on relation lifting of monotone relation between posets developed in~\cite{bkpv11}, and do not provide new results in this respect, (2) an appropriate notion of base of a locally monotone
finitary functor $T_\omega:\Pos\to\Pos$ developed in this paper in a slightly more general categorial setting. Bases of endofunctors  have recently been studied in the context of reachability in automata (see e.g.\ the approach taken in~\cite{BKR19,reachability19}), and we relate to their work in detail throughout the Section~\ref{sec:EM-basis}.

\subsection*{Organisation of the paper}
\begin{itemize}
\item We start by listing examples of poset endofunctors and examples of ordered coalgebras in Section~\ref{sec:functors}.
\item In Section~\ref{sec:EM-basis} we study the notion of a base in a general category equipped with a factorisation system. We describe the conditions under which an endofunctor of a category $\K$ admits a base, and we show that a large class of $\Pos$-endofunctors admits a base.
\item Section~\ref{sec:lifting} introduces monotone relations and the notion of relation lifting of monotone relations. It states the main result of~\cite{bkpv11}, that is, the characterisation of locally monotone $\Pos$-endofunctors that admit a relation lifting.
\item We define the syntax and semantics of Moss' logic for ordered coalgebras in Section~\ref{sec:logic}. The syntax and semantics of this logic is parametric in the type of coalgebras involved. We prove that the logic has Hennessy-Milner property for the notion of a simulation of coalgebras defined by relation lifting.
\item In Section~\ref{sec:proof-system} we develop a sequent calculus for Moss' logic that is again parametric in the type of coalgebras involved. We give a proof of soundness and completeness of the calculus.
\item Section~\ref{sec:conclusion} concludes with possible topics for future study, and clarifies how the current setting can capture positive fragments of finitary Moss' logic in $\Set$ studied in~\cite{bpv13}.
\end{itemize}

\subsection*{Acknowledgements}
We would like to thank Ji\v{r}\'{i} Velebil for his substantial help. His insights were crucial in the development of the paper. We are also thankful for the helpful comments of the three anonymous referees.

\section{Functors and coalgebras in $\Pos$}
\label{sec:functors}

In this section we fix the basic notation and introduce the running examples that are used throughout the paper.
Firstly we introduce an important class of endofunctors of posets, namely the locally monotone Kripke polynomial functors. In the second part of this section we introduce the notion of a coalgebra and show several examples of coalgebras for Kripke polynomial functors.

\subsection{Basic notions and Kripke polynomial functors}

We denote by $\Pos$ the category of all posets and all monotone maps, i.e.\ maps $f$ satisfying $x\leq y$ implies $f(x) \leq f(y)$. A monotone map $f$ is an \emph{order embedding} if moreover $f(x) \leq f(y)$ implies $x\leq y$.
For a poset $(X,\leq)$, we denote by $X^\op$ the opposite poset $(X,\geq )$.
For every pair $X$ and $Y$ of posets, the hom-set $\Pos(X,Y)$ of monotone
maps from $X$ to $Y$ carries a natural partial order: given two monotone maps $f$ and $g$ in $\Pos(X,Y)$, we define that $f \leq g$ holds if and only if
$f(x) \leq g(x)$ holds for every $x \in X$ (i.e.\ we introduce a pointwise order). The category $\Pos$ therefore can be seen as
{\em enriched\/} in posets. In cases where we need to emphasise this extra structure
of $\Pos$, we will speak of $\Pos$ as of a 2-category.

A functor $T: \Pos \to \Pos$ is then called a {\em locally monotone functor\/} (or a {\em 2-functor\/}), if it preserves the additional structure present in $\Pos$:
that is, if for any pair $f$, $g$ of comparable morphisms in $\Pos$ the inequality $f \leq g$ implies that the inequality $T(f) \leq T(g)$ holds.

The Kripke polynomial endofunctors of $\Pos$ are those defined by
the following grammar:
\begin{equation}
\label{eq:kripke-polynomial}
T::=
E
\mid
\Id
\mid
T+T
\mid
T\times T
\mid
T^E
\mid
T^\partial
\mid
\LLL T.
\end{equation}
We give an explanation of the building blocks of the grammar.
\begin{enumerate}
\item Let $E$ be an arbitrary poset. Slightly abusing the notation, we denote by $E$ the constant-at-$E$ functor. The functor $\Id$ is the identity functor.
\item Given two Kripke polynomial functors, their product and coproduct (in the category of poset endofunctors) is again a Kripke polynomial functor.
\item Given a poset $E$, we define that $T^E(X)=(TX)^E$ is the poset
of all monotone maps from $E$ to $TX$ with the pointwise ordering. The obvious action on morphisms makes $T^E$ into a functor.
\item Given a functor $T$, the functor $T^\partial$ is the {\em dual\/} of $T$, defined by putting
\[
T^\partial X=(T X^\op)^\op
\]
and again extending to the obvious action on morphisms. (Recall that $X^\op$ denotes the \emph{opposite} poset with the same underlying set satisfying $x \leq y$ in $X^\op$ precisely when $y \leq x$ in $X$.)
\item We denote by $\LLL$ the lowerset functor, with
$\LLL X=[X^\op,\bTwo]$ being the poset of all lowersets
on the poset $X$, and the order is given by inclusion. (By $\bTwo$ we denote the two-element chain.)
The lowerset functor $\LLL$ acts on morphisms
as the direct image followed by closure: given $f: X \to Y$, a lowerset
$l \in \LLL X$ is mapped by
$\LLL f$ to ${\downarrow} f[l]$, i.e., to the lowerset
generated by the image $f[l]$ of $l$.
\end{enumerate}
All Kripke polynomial endofunctors of $\Pos$ are locally monotone, as seen by a routine induction.

Observe that the dual of $\LLL$ is the upperset
functor $\UU$: for any poset $X$ we have that $\LLL^\partial X=[X,\bTwo]^\op$,
and the latter is the set of all uppersets on $X$,
ordered by reversed inclusion.
Therefore, the upperset functor $\UU$ is Kripke polynomial as well.
For a monotone map $f: X \to Y$, the map $\UU f: \UU X \to \UU Y$ sends an upperset $u \in \UU X$ to the upperset
${\uparrow} f[u]$
generated by the image $f[u]$ of $u$.

Apart from Kripke polynomial functors, we will list an additional functor we will briefly need at the very end of the paper: the convex powerset functor $\PP^c$. For a poset $X$, $\PP^cX$ is the poset of convex subsets of $X$ (i.e.\ subsets $w$ such that if $x,y\in w$ and $x\leq z\leq y$, then $z\in w$) partially ordered by the Egli-Milner preorder: $w \leq^{EM} w'$ iff
\[
\forall x\in w \ \exists x'\in w' (x\leq x') \mbox{ and } \forall x'\in w' \ \exists x\in w (x\leq x').
\]
Its action on morphisms $\PP^cf$ is the direct image map. Observe that this functor is self-dual: $(\PP^c)^\partial = \PP^c$.
This functor comes from the powerset functor $\PP: \Pre \to \Pre$ in the category of preorders and monotone maps, via the $Q \dashv I$ adjunction (of the quotient and inclusion functors) as $Q\PP I$~\cite{bkpv11}. 

We will systematically use \emph{finitary functors} on the syntactic side of the logics, i.e.\ functors that are completely determined by their action on finite posets. Formally, finitary functors are those preserving filtered colimits (see~\cite{kelly:structures}), but we will use the following elementary definition in this paper:

For a functor $T$, let $T_\omega$ be its least finitary subfunctor, with  $\monoInline{\nu^T}{T_\omega}{T}$ being the natural inclusion. Categorically, $T_\omega$ is the finitary coreflection of the functor $T$.
It can be defined in an elementary manner as follows: let $m$ range through finite subobjects 
$\monoInline{m}{Z}{X}$ of $X$ (i.e.\ with $Z$ being a finite poset) and put 
\[
T_\omega X := \bigcup\limits_{\monoInline{m}{Z}{X}} \img(Tm). 
\]
For a map $f: X \to Y$, $T_\omega f$ is the restriction of $Tf$ to $T_\omega Y$. 
\[
\xymatrixcolsep{1pc}
\xymatrix{
T_\omega X
\ar@{>->}[0,2]^{\nu_{X}^{T}}
\ar@{>->}[1,0]_{T_\omega f}
&
&
TX
\ar@{>->}[1,0]^{Tf}
\\
T_\omega Y
\ar@{>->}[0,2]^{\nu_{Y}^{T}}
&
&
T Y
}
\]
The functors $T$ and $T_\omega$ agree on finite posets (i.e., the corresponding natural inclusion is the identity on finite posets). We say that $T$ is \emph{finitary} if $T$ and $T_\omega$ coincide.

We will often explicitly use the finitary lowerset functor $\LLL_\omega$. For a poset $X$, the poset $\LLL_\omega X$ consists of
{\em finitely generated\/} lowersets on $X$ ordered by inclusion. Analogously, the finitary upperset functor $\UU_\omega$ assigns to a poset $X$ the poset of its finitely generated uppersets, ordered by reversed inclusion. Both functors are finitary subfunctors of the usual lowerset and upperset functors. We will use the notation $g(l)$ for the minimal discrete finite poset (i.e.\ finite set) of generators of the lowerset $l \in \LLL_\omega X$, and in the same way $g(u)$ will denote the minimal finite set of generators for the upperset $u \in \UU_\omega X$. When a lowerset is generated by a single element, we call it a \emph{principal} lowerset and denote such principal loweset generated by $x$ by $x{\downarrow}$. Similarly, principal upperset generated by $x$ is denoted by $x{\uparrow}$.

Similarly, the finitary convex powerset functor $\PP^c_\omega$ assigns, to a poset $(X,\leq)$, the poset of the finitely generated convex subsets of $X$ ordered by the Egli-Milner partial order $\leq^{EM}$. The functor $\PP^c_\omega$ is the natural counterpart --- the posetification --- of the finitary powerset functor $\PP_\omega: \Set \to \Set$ (see~\cite{balan+kurz}).

\subsection{Coalgebras for a functor}
Given a (not necessarily locally monotone) functor $T:\Pos\to\Pos$, we define coalgebras and their homomorphisms in
the usual manner.

Explicitly, a {\em coalgebra\/} for $T$ is a monotone map $c: X\to TX$,
and a monotone map
$h:X\to Y$ is a homomorphism from $c:X\to TX$ to $d:Y\to TY$ if the following
square
\[
\xymatrix{
X
\ar[0,1]^-{c}
\ar[1,0]_{h}
&
TX
\ar[1,0]^{Th}
\\
Y
\ar[0,1]_-{d}
&
TY
}
\]
commutes.

There are various interesting structures from computer science and logic that can be modelled as coalgebras for a suitable Kripke polynomial functor. We show some of the examples below.

\begin{exa}
\label{ex:coalgebras}
Examples of ordered coalgebraic structures.
\begin{enumerate}
\item
Consider the functor $T = A \times \Id$ for some fixed poset $A$. Coalgebras for the functor $T$ are the monotone maps $c = \langle {\mathit{out}}, {\mathit{next}}\rangle: X \to A \times X$ that can be seen as a particularly trivial kind of automata with the set of states $X$ and an (ordered) output alphabet $A$. Hence the monotone map $\mathit{out}:X\to A$ produces an output in $A$ for every state in $X$ and the monotone map $\mathit{next}:X\to X$ produces the next state of the automaton.

\item
A \emph{complete deterministic ordered automaton} (as considered in Chapter~2 of~\cite{pin04}) 
is a B\"{u}chi automaton $(X,A,E,I,F)$ where $X$ is equipped with a partial order $\leq$, $A$ is the alphabet, $E$ the set of $A$-labelled (deterministic) transitions on $X$ of the form $x\cdot a = y$, and $F$ is an order-ideal of final states. 
For each $x,y,y'\in X$ and $a\in A$, $x\leq y$ and $y\cdot a$ is defined, implies $x\cdot a$ is defined and $x\cdot a\leq y\cdot a$.

In a B\"{u}chi automaton, the words in the semigroup $A^+$ can bee seen to generate a morphism to the semigroup of binary relations on $X$, mapping a word $w$ to the accessibility relation labelled by (the letters) of $w$. We can think of a state $x$ as accepting $w$ if there is a $w$-labelled path from $x$ to a final state. In this sense, the partial order on the states of automata from this example corresponds to the reverse inclusion of the accepted languages: if $x\leq y$ and $y$ accepts $w$, then $x$ accepts $w$. The order on states in a precise technical sense generates a congruence  order on the semigroup $A^+$.

A complete deterministic ordered automaton can be seen as a coalgebra of the form $c:X\to X^A\times\bTwo$ where $X$ is a partial order of states, $A$ the (discrete) input alphabet, the left-hand component of the monotone map $c$ is a monotone transition map $x\mapsto (a\mapsto x\cdot a)$ and its right-hand component is a monotone map from $X$ to the two-element chain $\bTwo$, determining the accepting states. (In the poset setting, there would be no problem in considering the more general case of ordered automata with an \emph{ordered} input alphabet; we however stick to the definition already used in literature for this example.)

Ordered automata have been considered as a counterpart of ordered semigroups in~\cite{pin04}, to deal with some natural classes of recognizable sets of words closed under finite unions and intersection, but not under complement. 

\item
Consider the functor $T = \LLL_\omega^A \times \bTwo$
(with $A$ being discrete). A coalgebra
$c: X \to (\LLL_\omega X)^A \times \bTwo$
is a nondeterministic automaton consisting of an upperset
$F \subseteq X$ of accepting states, and a set $C$ of transitions
of the form $(x,a,y)$ with $x,y \in X$ and $a \in A$, satisfying
two confluence properties:
\begin{enumerate}
\item
Whenever $(x,a,y) \in C$ and $x \leq x'$ in $X$, it follows
that $(x',a,y) \in C$.
\item
Whenever $(x,a,y) \in C$ and $y \leq y'$ in $X$, it follows
that $(x,a,y') \in C$.
\end{enumerate}
Let moreover $d: Z \to (\LLL_\omega Z)^A \times \bTwo$ be
an automaton with an upperset $G \subseteq Z$ of accepting states
and a set $D$ of transitions. A monotone map $h: X \to Z$
is a homomorphism between $c$ and $d$ if the following three
conditions are satisfied:
\begin{enumerate}
\item
$h[F] = G$ and $h[X \setminus F] \cap G = \emptyset$
($h$ preserves the set of accepting states).
\item
For every transition $(x,a,y) \in C$ there is a transition
$(h(x),a,h(y)) \in D$.
\item
For every transition $(h(x),a,z) \in D$ there is a
transition $(x,a,y) \in C$ with $z \leq h(y)$.
\end{enumerate}
We will call coalgebras for the functor $\LLL_\omega^A \times \bTwo$ \emph{lowerset automata}.

\item \emph{Frames for distributive substructural logics}~\cite{rest:intr00} are frames
consisting of a poset $X$ of states, together with a ternary relation on $X$,
satisfying the following monotonicity condition:
\[
R(x,y,z)
\mbox{ and }
x'\leq x
\mbox{ and }
y'\leq y
\mbox{ and }
z\leq z'
\mbox{ implies }
R(x',y',z').
\]
In the notation of monotone relations of
Section~\ref{sec:lifting} below, $R$ is a monotone relation
of the form
$
\xymatrix@1{
R:
X
\ar[0,1]|-{\object @{/}}
&
X\times X
}
$.

Given a set $\At$ of atomic formulas and a monotone valuation
$X \times \At \to \bTwo$ of the atomic formulas,
we can interpret the fusion and
implication connectives as follows:
\begin{eqnarray*}
x\Vdash a_0\otimes a_1
& \mbox{iff} &
\ (\exists x_0)(\exists x_1)
(R(x_0,x_1,x)\mbox{ and } x_0\Vdash a_0\mbox{ and }x_1\Vdash a_1)
\\
x\Vdash a\rightarrow b
& \mbox{iff}
&\ (\forall y)(\forall z)\
((R(x,y,z)\mbox{ and } y\Vdash a) \mbox{ implies } z\Vdash b)
\end{eqnarray*}
It has been shown in~\cite{bhv12} that such frames can be treated as coalgebras in a natural way so that the coalgebraic morphisms coincide with the frame morphisms.
In particular, $R$ generates a natural coalgebraic structure $c_\otimes:X\to\LLL(X\times X)$ for a locally monotone functor $\LLL(\Id\times\Id)$, when we understand $R$ as interpreting the fusion connective. We call $R$ finitary when the generated coalgebraic structure is in fact of the form $c_\otimes:X\to\LLL_\omega(X\times X)$, i.e.\ when we can use the finitary lowerset functor.
\item
Consider a poset $X$ equipped with a monotone relation
$
\xymatrix@1{
R:
X
\ar[0,1]|-{\object @{/}}
&
X
}
$.
This can be seen as a frame for modal logic with adjoint modalities.
More in detail, such frames with a monotone valuation enable
one in principle to interpret conjunction, disjunction, and two monotone
modalities: a forward-looking $\Box$ and
a backward-looking $\blacklozenge$. We put
\begin{eqnarray*}
x\Vdash \Box a
& \mbox{iff} &
\ (\forall y)\ (R(x,y) \mbox{ implies } y\Vdash a)
\\
x\Vdash \blacklozenge a
& \mbox{iff} &
\ (\exists y)\ (R(y,x)\mbox{ and } y\Vdash a)
\end{eqnarray*}
The modalities $\Box$ and $\blacklozenge$ are adjoint in a sense that will be explained in~Remark~\ref{rem:semantic-preorder}.
The relation $R$ generates two coalgebras
\[
c: X \to \UU X \ \ \mbox{and}\ \ d: X \to \LLL X
\]
defined by $c(x)=\{y\ |\ R(x,y)\}$ and $d(x)=\{y\ |\ R(y,x)\}$. We say that the frame $(X,R)$ is \emph{finitary} if $R$ in fact generates coalgebras
\[
c: X \to \UU_\omega X \ \ \mbox{and}\ \ d: X \to \LLL_\omega X
\]
for the finitary functors $\UU_\omega$ and $\LLL_\omega$.

While the $\UU$-coalgebras can be seen as models of the box fragment of positive modal logic (i.e.\ the logic one obtains by extending the logic of distributve lattices with a box operator), the $\LLL$-coalgebras can be seen as models of the logic of distributive lattices extended with a backward looking diamond operator (which is the common fragment of substructural epistemic logics considered in~\cite{bmp16}).

\item Consider a $\PP^c$-coalgebra $c: X \to \PP^c$, assigning to a state $x$ a convex subset $c(x)$ so that whenever $x\leq y$, we have $c(x) \leq^{EM} c(y)$. Equipped with  a monotone valuation
$X \times \At \to \bTwo$ assigning uppersets to the atomic formulas, it enables us to interpret (conjunction and disjunction and) the box and diamond modalities:
\begin{eqnarray*}
x\Vdash \Box a
& \mbox{iff} &
\ (\forall y)\ (R(x,y) \mbox{ implies } y\Vdash a)
\\
x\Vdash \Diamond a
& \mbox{iff} &
\ (\exists y)\ (R(x,y)\mbox{ and } y\Vdash a),
\end{eqnarray*}
creating the following operations on $\UU X$: 
\begin{align*}
\Box_R(u) &:= \{x \ |\ c(x)\subseteq u \}\\
\Diamond_R(u) &:=  \{x \ |\ c(x)\cap u \neq \emptyset \}
\end{align*}
When seen as a relational assignment on the underlying set given by $xRy := y\in c(x)$, the above monotonicity condition does not make $R$ monotone per se, but translates to the following two conditions on $R$, familiar from the theory of partially ordered frames $(X,\leq)$ with valuations taking values in $\UU X$ (see e.g.~\cite{palmigiano04}):
\[
\leq \cdot R \subseteq R \cdot \leq \ \mbox{ and }\ \geq \cdot R \subseteq R \cdot \geq.
\]
The first condition corresponds to $\UU X$ being closed under $\Box_R$, the second condition corresponds to $\UU X$ being closed under $\Diamond_R$, thus ensuring the persistence of the semantics. 

It is not hard to work out that this is an alternative semantics to the \emph{positive modal logic}, i.e.\ the positive fragment of modal logic K. 
\item \emph{Semantics for description logic EL.} The small description logic EL is a lite description logic whose language allows for conjunction, existential restrictions, and the top-concept (see e.g.~\cite{Baader03}). While its semantics itself is not based on a poset, the syntax employs concept and role subsumptions in the form of inclusions. 

In a simple case, and when formulated in the poset setting, the syntax works as follows: we fix a poset $N^c$ of concept names ordered by concept subsumptions of the form $A\sqsubseteq B$, a poset $N^r$ of role names ordered by role subsumptions of the form $r\sqsubseteq s$, and generate the following grammar of concepts:
\[
C:=\ A\ |\ \top \ |\ C\sqcap D\ |\ \exists r.C
\]
The semantics is based on a discrete poset (i.e.\ a set) $\Delta$ together with: (ii) interpretations $C^I$ of concept names by lowersets of $\Delta$ (i.e.\ subsets of $\Delta$) respecting their order: $C\sqsubseteq D$ implies $C^I\subseteq D^I$, (ii) interpretations $r^I$ of role names by binary relations on $\Delta$, again respecting their order: $r\sqsubseteq s$ implies $r^I\subseteq s^I$. We can see each structure of this kind as a coalgebra
\[
c: \Delta \to (\LLL \Delta)^{N^r} \times \bTwo^{N^c}
\]
which, for each $d\in\Delta$, assigns: for each role name $r$ the lowerset (i.e.\ subset) $\{e\ |\ (d,e)\in r^I\}$, and for each concept name $C$ the value $1$ if $d\in C^I$, and the value $0$ otherwise. Both assignments are monotone. 
%
\end{enumerate}
\end{exa}

We shall use some of the above coalgebras as running
examples to illuminate concretely the general theory.

\section{Base in a category $\K$ with a factorisation system}
\label{sec:EM-basis}

When studying finitary coalgebraic logics for finitary endofunctors
of $\Set$, the notion of a {\em base\/} is of central
importance. On the level of models, the base allows us
to produce a finite set of successors of a state of the model.
On the level of the logical syntax, the base yields a finite set
of subformulas of a given formula.

Technically, given a functor $T:\Set\to\Set$, a base is a natural
transformation $\base^T_X: TX \to \PP_\omega X$, assigning to
each element $\alpha$ of $TX$ a finite subset of $X$,
called the {\em base\/} of $\alpha$. The idea is to define
$\base^T_X(\alpha)$ as the smallest finite $Z\subseteq X$
such that $\alpha\in TZ$ holds. For this idea to work,
$T$ has to satisfy additional conditions:  it has been shown
in~\cite{gumm05} that when $T$ preserves intersections and
weakly preserves inverse images, then the base exists and
it is natural. For various properties of bases for endofunctors
of $\Set$ we refer to, e.g.,~\cite{kkv12}.

In this section we explain how to define and compute a
{\em base for an endofunctor\/} $T:\K\to\K$ of an arbitrary
category $\K$ that is equipped with a factorisation
system $(\EE,\MM)$ for morphisms. The factorisation system
is used in order to be able to imitate the powerset:
the set of all subsets is
going to be replaced by the poset of $\MM$-subobjects.
\begin{rem}[Comparison to related work]
In~\cite{reachability19,BKR19}, the notion of base has been used to relate reachability within a coalgebra to a monotone operator on the (complete) lattice of subobjects of the carrier of
the coalgebra. 
\cite{BKR19} assumes the category $\K$ to be complete (which may limit applications), while~\cite{reachability19} (and us) do not do so.
Both~\cite{reachability19,BKR19} however define base for any $f: X \to TY$, while we do so only for $\MM$-subobjects $\monoInline{\alpha}{Z}{TX}$.

Both~\cite{reachability19,BKR19} developed the notion independently, almost at the same time: the relation is described in Remark 5.17 in~\cite{reachability19}. We have developed the notion independently of the two (a preliminary version of the current paper has been available on arXiv since 2019), and only got aware of their related work recently. Credit for our contribution should go mainly to J.~Velebil, who closely cooperated on a previous version of this paper.
\end{rem}
By applying our definition of a base to the epi-mono
factorisation system on the category of sets and mappings,
we then obtain the above notion of a base, see
Example~\ref{ex:base-in-sets}. By applying
our definition of a base to the category of all posets
and all monotone maps, equipped with a suitable
factorisation system, we will be able to speak of bases
in the ordered case as well. See Examples~\ref{ex:(E,M)-cats}
and~\ref{ex:bounded-cats} below.

\subsection{Factorisation systems}
For more details on the general theory of factorisation
systems we refer to the book~\cite{ahs}.

\begin{defi}
Suppose $\K$ is a category and let $\EE$ and $\MM$ be classes
of morphisms in $\K$. We say that $(\EE,\MM)$ is a
{\em factorisation system on $\K$\/}
(and that $\K$ is an {\em $(\EE,\MM)$-category\/}),
provided the following conditions are satisfied:
\begin{enumerate}
\item
The classes $\EE$ and $\MM$ are closed under composition with isomorphisms.
\item
Every morphism in $\K$
can be factorised as a morphism in $\EE$ followed by a morphism in $\MM$.
\item
$\K$ has the $(\EE,\MM)$-diagonalisation property, i.e., for every
commutative square
\[
\xymatrix{
A
\ar[0,1]^-{e}
\ar[1,0]_{u}
&
B
\ar[1,0]^{v}
\\
X
\ar[0,1]_-{m}
&
Y
}
\]
where $e$ is in $\EE$ and $m$ is in $\MM$, there exists a unique
diagonal $d:B\to X$, making both triangles
in the diagram
\[
\xymatrix{
A
\ar[0,1]^-{e}
\ar[1,0]_{u}
&
B
\ar[1,0]^{v}
\ar[1,-1]_{d}
\\
X
\ar[0,1]_-{m}
&
Y
}
\]
commutative.
\end{enumerate}
A factorisation system $(\EE,\MM)$ is called {\em proper\/}
if members of $\EE$ are (some) epimorphisms and members of $\MM$
are (some) monomorphisms.
\end{defi}

\begin{rem}
Any factorisation system $(\EE,\MM)$ on any category $\K$ satisfies the following (see Proposition~14.6 of~\cite{ahs}):
\begin{enumerate}
    \item The classes $\EE$ and $\MM$ are closed under composition.
    \item The intersection $\EE \cap \MM$ is the class of all isomorphisms in $\K$.
\end{enumerate}
\end{rem}

\begin{exa}
\label{ex:(E,M)-cats}
The following are examples of $(\EE,\MM)$-categories $\K$:
\begin{enumerate}
\item
$\K$ is the category $\Set$ of all sets and mappings, $\EE$ is the class
of all epis (=surjections) and $\MM$ is the class of all monos
(=injections).
\item
$\K$ is the category $\Pos$ of all posets and monotone mappings,
$\EE$ is the class of all monotone surjections and
$\MM$ is the class of all monotone maps that reflect the order.
\item
If $\K$ is the category $\Pre$ of all preorders and monotone mappings,
one can choose from at least one of the following three prominent
factorisation systems:
\begin{enumerate}
\item
$\EE$ is the class of all monotone maps that are bijective
on the level of elements and $\MM$ is the class of all monotone
maps that reflect the order.
\item
$\EE$ is the class of all monotone surjections and
$\MM$ is the class of all
monotone injections that reflect the order.
\item
$\MM$ is the class of monotone injections. The class
$\MM$ is the class of all monomorphisms, hence the
corresponding $\EE$ necessarily coincides with the
class of all strong epimorphisms, see~\cite{ahs}.
\end{enumerate}
\end{enumerate}
All of the above cases, except for (3a), are
proper factorisation systems.
\begin{enumerate}
\setcounter{enumi}{3}
\item
If $\K$ is an $(\EE,\MM)$-category, then $\K^\op$
is an $(\MM,\EE)$-category. If the factorisation
system $(\EE,\MM)$ on $\K$ is proper, so is
the factorisation system $(\MM,\EE)$ on $\K^\op$.
\end{enumerate}
\end{exa}

For a factorisation system $(\EE,\MM)$ on $\K$, the morphisms in $\MM$ will serve as ``subobjects''. It is however more useful to define subobjects as {\em equivalence classes}. More precisely: we fix an object $X$ and we denote
a morphism in $\MM$ into $X$ by $\monoInline{m}{Z}{X}$.
We define a preorder
on such morphisms by putting $m \subseteq m'$ if there
is a factorisation
\[
\xymatrixcolsep{1pc}
\xymatrix{
Z
\mono[0,2]
\mono[1,1]_{m}
&
&
Z'
\mono[1,-1]^{m'}
\\
&
X
&
}
\]
The preorder defines an equivalence $m\sim m'$
iff $m \subseteq m'$ and $m' \subseteq m$. We denote the
equivalence classes by $[m]$.
If the category $\K$ has a final object $\bOne$, we denote the situation
\[
\xymatrixcolsep{1pc}
\xymatrix{
\bOne
\mono[0,2]
\mono[1,1]_{x}
&
&
Z
\mono[1,-1]^{m}
\\
&
X
&
}
\]
by writing that $x \in m$ holds.

\begin{defi}
The equivalence class $[m]$ of $\monoInline{m}{Z}{X}$
is called
an {\em $\MM$-subobject\/} of $X$.

The class of all subobjects of $X$ is denoted by
\[
\Sub_\MM(X)
\]
and we consider it ordered by $[m] \subseteq [m']$ if
$m \subseteq m'$.
\end{defi}

\begin{defi}
We say that an $(\EE,\MM)$-category
$\K$ is {\em $\MM$-wellpowered},
if every $\Sub_\MM(X)$ has a set of elements.
And $\K$ is {\em $\EE$-cowellpowered}, if
$\K^\op$ (equipped with the factorisation system
$(\MM,\EE)$, see Example~\ref{ex:(E,M)-cats})
is $\EE$-wellpowered.
\end{defi}

\begin{rem}
\label{rem:sublattice}
\mbox{}\hfill
\begin{enumerate}
\item
We will abuse the notation and write $m$ instead of $[m]$, and $m \subseteq m'$ instead of $[m] \subseteq [m']$.
\item
An $(\EE,\MM)$-category need not be $\MM$-wellpowered,
even for a proper factorisation system $(\EE,\MM)$.
Indeed, suppose $\K$ is the class of all ordinals with the reversed
order (so that the ordinal $0$ becomes the top element) and consider $\K$
as a category. Put $\EE$ to be the class of identity morphisms,
$\MM$ to be the class of all morphisms. The factorisation
system $(\EE,\MM)$ on $\K$ is proper, since every morphism in
$\K$ is both a monomorphism and an epimorphism.
Then it follows that $\Sub_\MM(0)$ is a proper class.
\item
If $\Sub_\MM(X)$ has a set of elements,
then it is a complete lattice whenever $\K$ has enough limits.
In fact, it suffices to establish
the existence of infima in $\Sub_\MM(X)$. The infima can be computed
as limits of diagrams of the form
\[
\xymatrixrowsep{.1pc}
\xymatrix{
\vdots
&
&
\\
Z_i
\mono[1,1]^{m_i}
&
&
\\
\vdots
&
X
&
(i,j\in I)
\\
Z_j
\mono[-1,1]_{m_j}
&
&
\\
\vdots
&
&
}
\]
that are called {\em wide pullbacks of $\MM$-morphisms\/}
(or, {\em intersections of $\MM$-subobjects\/}).

We denote by $\monoInline{p_i}{Z}{Z_i}$ the limit
cone. Every $p_i$ belongs to $\MM$ by the general
properties of (proper) factorisation systems,
see~\cite{ahs}, Proposition~14.15. The composite
$\monoInline{m_i \cdot p_i}{Z}{X}$
represents the
$\MM$-subobject that is the intersection of $m_i$'s.
We denote the intersection by
\[
\xymatrix{
\bigcap_{i\in I} m_i:Z\mono[0,1] & X
}
\]
The existence of suprema in $\Sub_\MM(X)$ then
follows from the existence of infima by the usual
argument.
\item
When, moreover, {\em coproducts\/} exist in $\K$, suprema
in $\Sub_\MM(X)$ can be computed by an explicit formula.
Namely, the supremum $\bigcup_{i\in I} m_i$
(called the {\em union\/} of $m_i$'s) can be computed
as the $\MM$-part
\[
\xymatrix@1{\bigcup_{i\in I} m_i : C\mono[0,1] & X}
\]
of the $(\EE,\MM)$-factorisation of the unique cotupling morphism
$m:\coprod_{i\in I} Z_i \to X$ as follows:
\[
\xymatrixrowsep{.5pc}
\xymatrix{
\vdots
&
&
&
\\
Z_i
\ar[1,1]^-(.6){\inj_i}
\ar@{ >->} `r[rrrrd]^-{m_i} [rrrrd]
&
&
&
\\
\vdots
&
\coprod\limits_{i\in I}Z_i
\epi[0,1]
&
C
\mono[0,2]^-{\bigcup_{i\in I}m_i}
&
&
X
\\
Z_j
\ar[-1,1]_(.6){\inj_j}
\ar@{ >->} `r[rrrru]_-{m_j} [rrrru]
&
&
&
&
\\
\vdots
&
&
&
}
\]
\item When dealing with a pair $\monoInline{m}{Z}{X}$, $\monoInline{n}{Z}{X}$ of subobjects of $X$, we denote their intersection and union as $m \cap n$ and $m \cup n$, respectively.
\end{enumerate}
\end{rem}

\subsection{General bases}
\label{subsec:general-bases}
\begin{asm}
\label{ass:1}
For the rest of this section we assume the following
three conditions:
\begin{enumerate}
\item
$\K$ is an $(\EE,\MM)$-category
for a proper factorisation system $(\EE,\MM)$.
\item
$\K$
is $\MM$-wellpowered (i.e., we assume that every
$\Sub_\MM(X)$ has a set of elements).
\item
$T:\K\to\K$ is a functor which preserves $\MM$-morphisms.
\end{enumerate}
\end{asm}
The reader might compare our assumption~\ref{ass:1} (and what we have said in the item (6) of Remark~\ref{rem:sublattice} above) with Assumption~5.1 of~\cite{reachability19}, where $\K$ is assumed to be complete, wellpowered, with arbitrary (small) coproducts, equipped with a factorisation system $(\EE,\MM)$ where $\MM$ is a class of monomorphisms.

Also, compare with Assumption of Proposition 12 in~\cite{BKR19}, where $\K$ is assumed to be complete and well-powered.
\begin{defi}
\label{def:base}
An $\MM$-subobject $\monoInline{\base_X(\alpha)}{\wt{Z}}{X}$
is called a {\em base\/} of an
$\MM$-subobject $\monoInline{\alpha}{Z}{TX}$
provided that
the following holds for every $\monoInline{m}{Z'}{X}$

\begin{center}
$\base_X(\alpha) \subseteq m$ holds in $\Sub_\MM(X)$

 \quad
iff \quad

$\alpha \subseteq Tm$ holds in $\Sub_\MM(TX)$.
\end{center}
\end{defi}
Again, the reader might compare this with Definition~5.7 of~\cite{reachability19}, where $T$ is assumed to preserve monos (like we do here), but base is defined for every arrow $f: X\to TY$, while we do so only for subobjects. For the class of all monomorphisms such notion was first introduced by Alwin Block~\cite{Block12} and called a base.

Compare also with~\cite[Proposition 12]{BKR19}, where for any endofunctor $T$ which preserves (wide) intersections, base of every arrow $f: X\to TY$ is shown to exist.

\begin{rem}
\label{rem:3.7.bis}
Since $\Sub_\MM(X)$ is a poset, the base $\base_X(\alpha)$
is determined uniquely, whenever it exists. The base satisfies the \emph{unit} property
\[
\alpha \subseteq T\base_X(\alpha)
\]
for every $\alpha$ in $\Sub_\MM(TX)$,
and satisfies the \emph{counit} property
\[
\base_X(Tm) \subseteq m
\]
for every $m$ in $\Sub_\MM(X)$.
These properties are immediate consequences of the definition of base.
\end{rem}

We shall characterise functors admitting a base in Proposition~\ref{prop:preserves-intersections=>base}. Before we state the result, let us show that some functors do not admit a base.

\begin{exa}
Consider the category $\Pos$ equipped with the factorisation system of Example~\ref{ex:(E,M)-cats} (2), and the lowerset functor $\LLL:\Pos\to\Pos$.
We claim that
\[
\monoInline{\alpha}{\bOne}{\LLL \mathbb{Z}},
\quad
*\mapsto \mathbb{Z},
\]
where $\mathbb{Z}$ is the poset of integers with
the usual order, does not admit a base.

Indeed, for any $\monoInline{m}{Z'}{\mathbb{Z}}$,
the inequality $\alpha \subseteq \LLL m$ states that
(the image of) $Z'$ is a cofinal subset of $\mathbb{Z}$.
Had $\monoInline{\base(\alpha)}{\wt{\bOne}}{\mathbb{Z}}$ existed,
its image would yield the least cofinal
subset of $\mathbb{Z}$ --- a contradiction.
\end{exa}

Recall that meet-preserving maps between complete lattices have left adjoints. The proof of the following characterisation is a simple application of this fact (cf. Proposition~5.9 of Wissmann et al~\cite{reachability19}).

\begin{prop}
\label{prop:preserves-intersections=>base}
Suppose $\K$ has intersections of $\MM$-subobjects.
For $T:\K\to\K$, the following are equivalent:
\begin{enumerate}
\item
For every $X$ and every $\monoInline{\alpha}{Z}{TX}$, the base
$\monoInline{\base_X(\alpha)}{\wt{Z}}{X}$ exists.
\item
$T$ preserves intersections of $\MM$-subobjects.
\end{enumerate}
Moreover, under any of the above conditions,
$\base_X(\alpha)$ can be computed as the intersection
\[
\bigcap_{\alpha \subseteq Tm} m
\]
of all $\MM$-subobjects $\monoInline{m}{Z'}{X}$ such
that $\alpha \subseteq Tm$ holds in $\Sub_\MM(TX)$.
\end{prop}
\begin{proof}
Observe that the assignment $f\mapsto Tf$ induces a monotone
map
\[
T_X:\Sub_\MM(X)\to\Sub_\MM(TX)
\]
since $T$ is assumed to preserve $\MM$-morphisms.
The assignment $\alpha\mapsto\base_X(\alpha)$ is then simply
the value of a left adjoint
\[
\base_X:\Sub_\MM(TX)\to\Sub_\MM(X)
\]
to $T_X$ at a point $\alpha$. Moreover, $\Sub_\MM(X)$ and $\Sub_\MM(TX)$
are complete lattices, since $\K$ is assumed to have
intersections of $\MM$-subobjects.

Then~(1) is equivalent to~(2), since~(1) asserts that $\base_X\dashv T_X$
holds for any $X$ and~(2) asserts that $T_X$ preserves infima for any $X$.

The final assertion is true since the value of the left adjoint
$\base_X$ can be computed by the formula
\[
\base_X(\alpha)
=
\bigcap_{\alpha\subseteq Tm} m. \qedhere
\]
\end{proof}

In practical applications it may be the case that one is interested
in the existence of $\base_X(\alpha)$ only for particular $\alpha$'s.
For example, in~\cite{venemafp} the base is computed only
for finite subsets of $TX$ for functors $T:\Set\to\Set$.

We imitate this approach here. Namely, besides our standing
assumptions, we assume further that a full
subcategory $\K_\lambda$ of $\K$ is given, where $\lambda$
is a regular cardinal. We want to think of the
objects of $\K_\lambda$ as being ``smaller than $\lambda$''.
Our goal is to construct the value of a left adjoint
$\monoInline{\base_X(\alpha)}{\wt{Z}}{X}$ only for
those $\MM$-subobjects $\monoInline{\alpha}{Z}{TX}$
with $Z$ in $\K_\lambda$, and such that $\wt{Z}$ is in $\K_\lambda$
as well.

For every $X$, denote by
\[
I_X:\Sub_{\MM,\lambda}(X)\to\Sub_{\MM}(X)
\]
the inclusion of the subposet
spanned by $\MM$-subobjects having a domain in $\K_\lambda$.

One can refine the existence of the base by applying
the ideas of Freyd's Adjoint Functor
Theorem~\cite{freyd}. In the formulation
we use the terminology {\em $\lambda$-small\/}
to abbreviate ``of cardinality less than $\lambda$''.
For example, a $\lambda$-small set is one of less than $\lambda$
elements, a $\lambda$-small intersection is one indexed by a
$\lambda$-small set, etc.

\begin{prop}
\label{prop:lambda-small}
Suppose that for every $X$, the poset $\Sub_\MM(X)$ has $\lambda$-small infima, and suppose that $\Sub_{\MM,\lambda}(X)$ is a subset in $\Sub_\MM(X)$ that is closed under $\lambda$-small infima.

Suppose further that the following two conditions are satisfied:
\begin{enumerate}
\item
$T$ preserves $\lambda$-small intersections.
\item
For any $\monoInline{\alpha}{Z}{TX}$ with $Z$ in $\K_\lambda$
there exists a $\lambda$-small set
\[
F = \{f_i \mid i \in I, \; \alpha \subseteq Tf_i \} \subseteq \Sub_{\MM,\lambda}(X)
\]
such that for every $m$ in $\Sub_{\MM}(X)$ with $\alpha \subseteq Tm$ there exists $f \in F$ with $Tf \subseteq Tm$.
\end{enumerate}
Then $\base_X(\alpha)$ exists for every
$\monoInline{\alpha}{Z}{TX}$
with $Z$
in $\K_\lambda$, it is an element of $\Sub_{\MM,\lambda}(X)$,
and it is computed by the formula
\[
\base_X(\alpha)
=
\bigcap F.
\]
\end{prop}
\begin{proof}
By hypothesis, $\Sub_{\MM,\lambda}(X)$ has $\lambda$-small infima, and the map $T_{X}$ preserves them by condition~(1).
Condition~(2) now ensures that
\[
\bigcap_{\alpha \subseteq Tm} m =
\bigcap F
\]
holds. The latter infimum is $\lambda$-small and $T_{X}$
is assumed to preserve such infima. The result now follows.
\end{proof}

We show how Proposition~\ref{prop:lambda-small} is applied
in the classical case when $\K=\Set$. Recall that a
category $\D$ is {\em $\lambda$-filtered\/}
if every $\lambda$-small diagram in $\D$ has a cocone in $\D$.
A colimit is {\em $\lambda$-filtered\/} if its scheme
is a $\lambda$-filtered category.

\begin{exa}
\label{ex:base-in-sets}
Let $\K$ be the category $\Set$ of all sets and mappings
with the factorisation system of surjections and injections.
Let $T:\Set\to\Set$
be a finitary functor that preserves injections and
finite intersections.
Then, for every $\monoInline{\alpha}{Z}{TX}$
with the set $Z$ finite,
$\monoInline{\base_X(\alpha)}{\wt{Z}}{X}$
exists and the set $\wt{Z}$ is finite.

It suffices to verify conditions~(1) and~(2) of
Proposition~\ref{prop:lambda-small}
(where $\lambda=\aleph_0$ and $\Set_{\aleph_0}$ are
the finite sets). Condition~(1)
is well-known to hold for $\Set$\footnote{In fact, one has to be slightly careful here: every endofunctor of $\Set$ preserves all finite (thus even \emph{empty}) intersections provided it is \emph{sound}, see~\cite{adamek+gumm+trnkova} for an explanation or~\cite{trnkova} for the original result. Any non-sound endofunctor of $\Set$ can easily be ``repaired'' to be sound~\cite{adamek+gumm+trnkova}.}. Condition~(2) is verified as
follows.
Given $\monoInline{\alpha}{Z}{TX}$ with $Z$
a finite set, express $X$ as a filtered colimit
$\colim_{i\in I} X_i$ of subobjects
$\monoInline{m_i}{X_i}{X}$ with all $X_i$ finite.
Since $T$ is finitary and is assumed
to preserve injections, $TX$ is a filtered colimit
$\colim_{i\in I} TX_i$ of injections.
Since $Z$ is finite, there exists
$i_0$ such that the triangle
\[
\xymatrix{
Z
\mono[0,2]^-{\alpha}
\mono[1,2]
&
&
TX
\\
&
&
TX_{i_0}
\mono[-1,0]_{Tm_{i_0}}
}
\]
commutes. Denote by $F$ the set
\[
\{ f\mid f \subseteq m_{i_0} \mbox{ and } \alpha \subseteq Tf\}
\]
of elements of $\Sub_\MM(X)$.

Observe first that $F$ is nonempty and finite. Secondly,
if $\alpha \subseteq Tm$, then
\[
\alpha \subseteq Tm_{i_0} \cap Tm = T(m_{i_0} \cap m)
\]
since $T$ preserves finite intersections.
Since the subobject $m_{i_0}\cap m$ is in $F$,
condition~(2) has been verified.
\end{exa}

More generally, Proposition~\ref{prop:lambda-small} and
the technique of Example~\ref{ex:base-in-sets}
can be applied to the case when $\K$
is a {\em locally $\lambda$-bounded category\/}
for some regular cardinal $\lambda$ and a proper factorisation system
$(\EE,\MM)$, see~\cite{kelly:book}.
In particular, we will use the results for $\K=\Pos$
equipped with the factorisation system of monotone
surjections and order-embeddings, see Example~\ref{ex:bounded-cats}.

\begin{defi}
A category $\K$ is {\em locally $\lambda$-bounded\/}
w.r.t.\ a proper factorisation system $(\EE,\MM)$,
if the following two conditions hold:
\begin{enumerate}
\item
There is an essentially small full subcategory $\K_\lambda$
of $\K$ such that:
\begin{enumerate}
\item
Every object $Z$ in $\K_\lambda$ is {\em $\lambda$-bounded},
i.e., the functor $\K(Z,{-})$ preserves colimits of $\lambda$-filtered
diagrams of $\MM$-subobjects.
\item
A morphism $\monoInline{m}{X}{Y}$ is an isomorphism iff the map $\K(Z,m):\K(Z,X)\to\K(Z,Y)$ is a bijection for every $\lambda$-bounded $Z$. (I.e., $\K_\lambda$ is an \emph{$(\MM,\EE)$-generator} of $\K$.)
\end{enumerate}
\item
The category $\K$ is cocomplete and has all
{\em $\EE$-cointersections}.
This means that the opposite category $\K^\op$
(that is equipped with $(\MM,\EE)$
as a factorisation system) has all intersections of
$\EE$-subobjects.
\end{enumerate}
A category $\K$ is called {\em $\lambda$-ranked}, if it is locally
$\lambda$-bounded and $\EE$-cowellpowered.
\end{defi}

\begin{exa}
\label{ex:bounded-cats}
Examples of locally $\lambda$-bounded categories w.r.t.\ a factorisation
system $(\EE,\MM)$ include the following:
\begin{enumerate}
\item
All locally $\lambda$-presentable categories when we take
$\EE$ to be strong epimorphisms and $\MM$ to be monomorphisms,
see~\cite{kelly:book}.

In particular, the categories $\Pos$ and $\Pre$ are locally
$\aleph_0$-bounded, if we take $\MM$ to consist of all
monotone injections (not necessarily reflecting order).
A preorder or a poset $X$ is $\aleph_0$-bounded iff it is finite.
\item
The category $\Pos$ of posets and monotone maps is locally finitely presentable, and thus it is locally
$\aleph_0$-bounded
for $\EE$ consisting of surjective monotone maps
and $\MM$ of monotone maps reflecting the order, see e.g.\ Proposition~1.61 in~\cite{adamek+rosicky}.
A poset is $\aleph_0$-bounded iff it is finite.
\item
The category $\Pre$ of preorders and monotone maps
can be equipped with a variety of factorisation systems,
see Example~\ref{ex:(E,M)-cats}. The category $\Pre$ is locally
$\aleph_0$-bounded w.r.r.\ the factorisation
system where
$\EE$ consists of all monotone surjections and $\MM$
consists of all monotone injections that reflect the order,
see Theorem~5.6 of~\cite{kelly+lack}.
A preorder $X$ is $\aleph_0$-bounded iff it is finite.
\end{enumerate}
\end{exa}

\begin{defi}
Let $\K$ be a $\lambda$-bounded category w.r.t.\
a factorisation system $(\EE,\MM)$. We say that
$T:\K\to\K$ {\em admits $\lambda$-bounded bases},
provided that $\monoInline{\base^T_X}{\wt{Z}}{X}$ exists for any
$\monoInline{\alpha}{Z}{TX}$ with $Z$ in $\K_\lambda$
and, moreover, $\wt{Z}$ is in $\K_\lambda$.
\end{defi}

\begin{prop}
\label{prop:3.10.bis}
Suppose that $\K$ is locally $\lambda$-bounded w.r.t.\
a factorisation system $(\EE,\MM)$. Then
$T:\K\to\K$ admits $\lambda$-bounded bases,
whenever the following four conditions are satisfied:
\begin{enumerate}
\item
Every subposet $\Sub_{\MM,\lambda}(X)$ in $\Sub_\MM(X)$ is closed under $\lambda$-small infima.
\item
The principal lowersets
in every poset $\Sub_{\MM,\lambda}(X)$ are $\lambda$-small.
\item
$T$ preserves colimits
of $\lambda$-filtered diagrams of $\MM$-subobjects.
\item
$T$ preserves $\lambda$-small intersections of $\MM$-subobjects.
\end{enumerate}
\end{prop}
\begin{proof}
It is proved in~\cite{kelly+lack} that
every locally $\lambda$-bounded category is necessarily complete. Hence every $\Sub_\MM(X)$ is a complete lattice. Every object $X$ of $\K$ can be represented as a filtered colimit of its $\lambda$-subobjects; consider the canonical filtered diagram $D: \D \to \K$ of $X$. Its colimit $\colim D$ yields a unique comparison morphism $\colim D \to X$ by the couniversal property of the colimit. This comparison is an isomorphism since $\K_\lambda$ is an $(\MM,\EE)$-generator.
We can thus use Proposition~\ref{prop:lambda-small} and the same argument as in Example~\ref{ex:base-in-sets} to conclude the result.
\end{proof}

\subsection{Bases in posets}

We will apply the general theory of
Subsection~\ref{subsec:general-bases}
to the particular case of the category $\Pos$.
In more detail, we consider $\Pos$ as an $(\EE,\MM)$-category
for $\EE$ = monotone surjections and $\MM$ = order-embeddings (i.e.\ monotone maps $f$ satisfying $x\leq y$ iff $f(x)\leq f(y)$).
Then the following conditions hold:
\begin{enumerate}
\item
$(\EE,\MM)$ is a proper factorisation system, see
Example~\ref{ex:(E,M)-cats}.
\item
$\Pos$ is $\MM$-wellpowered.
\item
$\Pos$ is locally $\aleph_0$-bounded, see
Example~\ref{ex:bounded-cats}. The category
$\Pos_{\aleph_0}$ consists of finite posets.
\end{enumerate}
Hence, by Proposition~\ref{prop:3.10.bis}, a functor $T:\Pos\to\Pos$ admits $\aleph_0$-bounded (=finite) bases, whenever $T$ preserves order-embeddings
and their finite intersections. Next we survey which functors preserve those properties.
From now on, all bases will be considered finite. 
We start with a negative example:
\begin{exa}
Not all locally monotone endofunctors $T: \Pos \to \Pos$ preserve order-embeddings. See e.g.\ Example~6.1 of~\cite{bkpv12}.

Moreover, a locally monotone endofunctor of $\Pos$ that preserves
order-embeddings does not necessarily preserve finite intersections. The reasoning is the same as in the case of $\Set$-endofunctors, as empty intersections are not necessarily preserved. Suppose $T:\Pos\to\Pos$ assigns the two-element chain $\bTwo$ to every nonempty poset and it assigns the one-element poset $\bOne$ to the empty poset. On morphisms $T$ sends the unique morphism $!_X:\emptyset\to X$ to $\ulcorner 1\urcorner:\bOne\to\bTwo$, all other morphisms are mapped to the identity morphism. The intersection
\[
\xymatrix{
\emptyset
\ar[0,1]^-{!_\bOne}
\ar[1,0]_{!_\bOne}
&
\bOne
\ar[1,0]^{\ulcorner 1\urcorner}
\\
\bOne
\ar[0,1]_-{\ulcorner 0\urcorner}
&
\bTwo
}
\]
is clearly not preserved by $T$.
\end{exa}

\paragraph{Preservation of order embeddings.} We now list examples of endofunctors of the category
$\Pos$ that admit finite bases. We start with the preservation of order-embeddings.

\begin{exa}
All Kripke-polynomial endofunctors~\eqref{eq:kripke-polynomial}
of the category $\Pos$ preserve order-embeddings.
This is essentially proved in~\cite{bkpv12}, Examples~5.3 and~6.3,
for the case of Kripke polynomial endofunctors of $\Pre$.
One only needs to notice that order-embeddings in $\Pos$ are
precisely the order-reflecting monotone injective
maps of the underlying preorders.
\end{exa}
In fact, the class of endofunctors of $\Pos$
that preserve order-embeddings is quite large.
It includes all functors that preserve certain lax diagrams
called  \emph{exact squares}, see~\cite{guitart} or~\cite{bkpv11}.

\begin{defi}
\label{def:BCC}
A lax square
\[
\xymatrix{
P
\ar[0,1]^-{p_1}
\ar[1,0]_{p_0}
&
B
\ar[1,0]^{g}
\\
A
\ar[0,1]_{f}
\ar@{}[-1,1]|{\nearrow}
&
C
}
\]
in $\Pos$ is called exact, if $fa \leq_C gb$ entails
the existence of $w$ in $P$ such that both $a \leq_A p_0 w$
and $p_1 w \leq_B b$ hold.

A functor $T$ which preserves exact squares will be also called
a functor satisfying the {\em Beck-Chevalley Condition}, or
BCC for short.
\end{defi}

\begin{exa}
\label{ex:3.15.bis}
All endofunctors of $\Pos$ satisfying BCC preserve
order-embeddings. In particular, all Kripke-polynomial endofunctors of $\Pos$ satisfy BCC, as does the convex powerset functor $\PP^c$. See~\cite{bkpv11}.

For an example of an endofunctor of $\Pos$ not satisfying BCC, consider the \emph{connected components functor} $C: \Pos \to \Pos$ that takes a poset $P$ to the discrete poset consisting of connected components of $P$. The functor $C$ does not preserve (e.g.) the order-embedding from the discrete poset on $\{a,b\}$ to the poset on $\{a,b,c\}$ with the ordering $a < c$, $b < c$.
(This is Example~6.7 of~\cite{bkpv11}.)
\end{exa}

\paragraph{Preservation of finite intersections.}
The following series of results deals with
preservation of finite intersections.
As the next example shows, not all Kripke polynomial
functors preserve finite intersections.

\begin{exa}
\label{ex:3.16}
The lowerset functor $\LLL$ does not preserve intersections
of order-embeddings.
Consider the poset $\ZZ$ of integers and its subposets
$\monoInline{m}{E}{\ZZ}$ and
$\monoInline{n}{O}{\ZZ}$ of even and odd integers,
respectively.
The intersection of $m$ and $n$ is \emph{empty}, as it is shown
in the diagram on the left below:
\[
\xymatrix{
\emptyset \mono[0,1] \mono[1,0] & O \mono[1,0]^{n} \\
E \mono[0,1]_{m} & \ZZ
}
\qquad
\xymatrix{
\LLL \emptyset \mono[0,1] \mono[1,0] & \LLL O \mono[1,0]^{\LLL n} \\
\LLL E \mono[0,1]_{\LLL m} & \LLL \ZZ
}
\]
Then the one-element poset $\LLL\emptyset$
in the diagram on the right above
is not the intersection of $\LLL m$ and $\LLL n$.
Namely, the intersection of $\LLL m$ and $\LLL n$
contains at least two elements, since
$\LLL m (\emptyset) = \emptyset = \LLL n (\emptyset)$ and
$\LLL m (E) = \ZZ = \LLL n (O)$.

The functor $\LLL$ does not, in general, preserve \emph{non-empty} intersections of order-embeddings as well. Consider again the above example, and take as $m: E \to \ZZ$ the poset of even integers and $n: O \to \ZZ$ of odd integers and zero. Then the domain of $m \cap n$ is the one-element set $\{ 0 \}$ and $\LLL \{ 0 \}$ has two elements, while the domain of $\LLL m \cap \LLL n$ contains $\emptyset$, ${\downarrow} 0$, and $\ZZ$. 
\end{exa}

\begin{prop}
All {\em finitary\/} Kripke-polynomial endofunctors of
$\Pos$ preserve finite
intersections of order-embeddings.
\end{prop}
\begin{proof}
The proof is trivial for all the formation
steps~\eqref{eq:kripke-polynomial}
except for $T^E$ and $\LLL_\omega T$.
For the induction hypothesis,
suppose $T$ preserves finite intersections of
order-embeddings.
\begin{enumerate}
\item
Since $T^E X=(TX)^E$, it suffices to observe
that the diagram
\[
\xymatrix{
(TL)^E
\mono[0,1]^-{(Ti)^E}
\mono[1,0]_{(Tj)^E}
&
(TB)^E
\mono[1,0]^{(Tn)^E}
\\
(TA)^E
\mono[0,1]_{(Tm)^E}
&
(TX)^E
}
\]
is a pullback whenever the following diagram is:
\[
\xymatrix{
TL
\mono[0,1]^-{Ti}
\mono[1,0]_{Tj}
&
TB
\mono[1,0]^{Tn}
\\
TA
\mono[0,1]_{Tm}
&
TX.
}
\]
This is true because the functor $(-)^E:\Pos\to\Pos$ preserves all limits: it is a right adjoint by cartesian closedness of $\Pos$ (see Example~27.3 of~\cite{ahs}).
\item
It suffices to prove that
$\LLL_\omega$ preserves finite intersections of
order-embeddings.
\[
\xymatrix{
L \mono[0,1] \mono[1,0] & B \mono[1,0]^{n} \\
A \mono[0,1]_{m} & X
}
\qquad
\xymatrix{
\LLL_\omega L \mono[0,1] \mono[1,0] & \LLL_\omega B
\mono[1,0]^{\LLL_\omega n} \\
\LLL_\omega A \mono[0,1]_{\LLL_\omega m} & \LLL_\omega X
}
\qquad
\xymatrix{
\lambda \ar@{|->}[0,1] \ar@{|->}[1,0]
& \beta \ar@{|->}[1,0]^{\LLL_\omega n} \\
\alpha \ar@{|->}[0,1]_{\LLL_\omega m} & \chi
}
\]

Given an intersection in the diagram on the left above,
we show that the diagram in the center above is
again an intersection of the morphisms
$\LLL_\omega m$ and $\LLL_\omega n$. For any two lowersets $\alpha$ in $\LLL_\omega A$ and $\beta$ in $\LLL_\omega B$ satisfying
$\LLL_\omega m (\alpha) = \chi = \LLL_\omega n (\beta)$
it is enough to find a lowerset $\lambda$ in $\LLL_\omega L$
as shown in the diagram above on the right.
All of the depicted lowersets are determined
by their minimal sets of generators.
Recall that the minimal set of generators of the lowerset $\alpha$ is denoted by
$g(\alpha)$. Since $m$ and $n$ are inclusions, we see that
$g(\alpha) = g(\chi) = g(\beta)$. The equality of sets of
generators implies that $g(\chi) \subseteq L$.
Defining $\lambda$ in $\LLL_\omega L$
to be the lowerset generated by $g(\chi)$ yields
the unique witness in $\LLL_\omega L$.

Observe that the preceding argument works even in
the case of $g(\chi)$ being empty. \qedhere
\end{enumerate}
\end{proof}

By the above, all finitary Kripke-polynomial functors
$T$ admit finite bases.

\begin{exa}
We will show how the bases for Kripke-polynomial functors can be computed.
\begin{enumerate}
\item Constant functor $E$: For an arbitrary subobject $\monoInline{\alpha}{Z}{E}$, we have that its base is of the form $\monoInline{\base(\alpha)}{\emptyset}{X}$, with $\base(\alpha)$ being the empty mapping.
\item Identity functor $\Id:$ Given a subobject $\monoInline{\alpha}{Z}{X}$, its base $\monoInline{\base(\alpha)}{\wt{Z}}{X}$ is the mapping $\alpha$ itself.
\item Sum and product of functors: Fix two functors $T_1$ and $T_2$. Given a subobject
\[
\monoInline{\alpha}{Z}{T_1 X + T_2 X},
\]
it is equally well a subobject $\monoInline{\alpha_1 + \alpha_2}{Z_1 + Z_2}{T_1 X + T_2 X}$ for some choice of posets $Z_1$ and $Z_2$. Then $\base(\alpha) = \base(\alpha_1) \cup \base(\alpha_2)$. For the case of the product of functors, given a subobject $\monoInline{\alpha}{Z}{T_1 X \times T_2 X}$, denote by $p_i: T_1 X \times T_2 X \to T_i X$ the $i$-th projection from the product ($i \in \{1,2\}$). Then $\base(\alpha) = \base(p_1 \cdot \alpha) \cup \base(p_2 \cdot \alpha)$.
\item Power of a functor: Given a functor $T^E$ and a subobject $\monoInline{\alpha}{Z}{(TX)^E}$, we denote by $p_e: TX^E \to TX$ the obvious projection (with $e \in E$). Then $\base(\alpha) = \bigcup_{e \in E} \base(p_e \cdot \alpha)$.
\item The dual of a functor: A straightforward computation yields that the base of
\[
\monoInline{\alpha}{Z}{T^\partial X}
\]
is $(\base(\alpha^\op))^\op$.
\item The lowerset functor: Given a lowerset $\monoInline{\alpha}{\bOne}{\LLL_\omega X}$, the base
$\monoInline{\base(\alpha)}{\wt{\bOne}}{X}$ is the (discrete) finite poset of
generators of $\alpha$.
More generally, the base of  $\monoInline{\alpha}{Z}{\LLL_\omega X}$ is $\bigcup_{z \in Z} \base(\alpha \cdot z)$.
\end{enumerate}
\end{exa}

\section{Monotone relations and their lifting}
\label{sec:lifting}

In the current section, we summarise the notation and the necessary facts that concern monotone relations between posets. We will use these facts and the facts about liftings of monotone relations in Section~\ref{sec:logic} to introduce the semantics of a coalgebraic logic over posets. For a more detailed treatment of the theory of relation liftings in the categories $\Pre$ and $\Pos$ we refer
to~\cite{bkpv11}.

The category $\Rel$ of monotone relations over $\Pos$ has the same
objects as the category $\Pos$, and has \emph{monotone relations}
as arrows. A monotone relation $R$ from $X$ to $Y$ will be denoted
by
$
\xymatrix@1{
R:
X
\ar[0,1]|-{\object @{/}}
&
Y
}
$
and it will be identified with a monotone map
$R:Y^\op\times X\to\bTwo$.
We shall often write $R(y,x)$ or $y \mathrel{R} x$ to denote that $R$ relates $x$ to $y$.
Using this notation, a relation $R$ is monotone if it
satisfies the following monotonicity condition:
\[
R(y,x)
\mbox{ and }
y'\leq_Y y
\mbox{ and }
x\leq_X x'
\mbox{ implies }
R(y',x').
\]
The composition in $\Rel$ is computed in the usual manner:
$(S\cdot R)(z,x)$ holds if and only if $S(z,y)$ and $R(y,x)$ hold for some $y$.
The identity morphism on a poset $X$ is the order $\leq_X$ of $X$,
considered as the monotone map ${\leq_X}:X^\op\times X\to\bTwo$.
For any two posets $X$ and $Y$, the hom-set $\Rel(X,Y)$ of all monotone relations from $X$ to $Y$ carries a poset structure given by the inclusion of relations.

\begin{exa}
We shall often use two special kinds of relations: graph relations and membership relations.
\begin{enumerate}
\item
{\em Graph relations.\/}
For a monotone map $f: X\to Y$ we define two graph relations
\[
\xymatrix{
X
\ar[0,1]|-{\object @{/}}^-{f_\diamond}
&
Y
&
&
Y
\ar[0,1]|-{\object @{/}}^-{f^\diamond}
&
X
}
\]
by putting
\[
f_\diamond(y,x)
\textrm{ iff }
y\leq fx,
\quad
f^\diamond(x,y)
\textrm{ iff }
fx\leq y.
\]
The assignment
\[
f \mapsto f_\diamond
\]
can be extended to a locally monotone functor
\[
({-})_\diamond: \Pos \to \Rel.
\]
\item
{\em Membership relations.\/}
For a poset $X$, we use the following two membership relations:
\[
\xymatrix@1{
\sqin_{X}:
\LLL X
\ar[0,1]|-{\object @{/}}
&
X
&
&
\sqni_{X}:
X
\ar[0,1]|-{\object @{/}}
&
\UU X
}
\]
defined by $\sqin_X(x,l)=1$ iff $x$ is in $l$
and $\sqni_X(u,x)=1$ iff $x$ is in $u$.
\end{enumerate}
\end{exa}

As in the case of ordinary relations, there are various operations that we can perform on monotone relations.

\begin{defi}[Operations on relations]
\label{def:operations-on-relations}
Suppose
$
\xymatrix@1{
R:
X
\ar[0,1]|-{\object @{/}}
&
Y
}
$
is a monotone relation.
\begin{enumerate}
\item
The {\em converse\/}
$
\xymatrix@1{
R^{\con}:
Y^\op
\ar[0,1]|-{\object @{/}}
&
X^\op
}
$
of $S$ is defined by putting
$R^\con(x,y)$ iff $R(y,x)$.
\item
The {\em negation\/}
$
\xymatrix@1{
\neg R:
X^\op
\ar[0,1]|-{\object @{/}}
&
Y^\op
}
$
is defined by putting $\neg R(y,x)$ iff it is not the case that $(y,x) \in R$.
\item
Given monotone maps $f:A\to X$, $g:B\to Y$,
the composite
\[
\xymatrix{
A
\ar[0,1]|-{\object @{/}}_-{f_\diamond}
&
X
\ar[0,1]|-{\object @{/}}_-{R}
&
Y
\ar[0,1]|-{\object @{/}}_-{g^\diamond}
&
B
\ar @{<-} `u[lll] `[lll]|-{\object @{/}}_-{R(g-,f-)} [lll]
}
\]
is called a {\em restriction\/} of $R$ along $f$ and $g$.
\end{enumerate}
\end{defi}

\begin{rem}
Observe that $\sqni_{X}$ is the converse
of $\sqin_{X^\op}$. We omit the subscript $X$ whenever it is clear from the context. We also often use the notation
\[
\xymatrix@1{
\not\sqin_{X}:
(\LLL X)^\op
\ar[0,1]|-{\object @{/}}
&
X^\op
&
&
\not\sqni_{X}:
X^\op
\ar[0,1]|-{\object @{/}}
&
(\UU X)^\op
}
\]
instead of the notation $\neg {\sqin}$ and $\neg {\sqni}$ for the negation of the respective membership relations. Whenever it is possible, we use the infix notation for the membership relations.
\end{rem}

Any locally monotone functor $T: \Pos\to\Pos$ that satisfies BCC
(see Definition~\ref{def:BCC})
can be lifted to a locally monotone
functor $\ol{T}:\Rel\to\Rel$. In fact, BCC is equivalent
to the existence of such a lifting. The following theorem
was proved in~\cite{bkpv11} for the category $\Pre$,
but the proof goes through verbatim for the case of $\Pos$.

\begin{thmC}[\cite{bkpv11}]
\label{th:lifting}
For a locally monotone functor $T:\Pos\to\Pos$ the following are equivalent:
\begin{enumerate}
\item
The functor $T$ has a locally monotone functorial relation
lifting $\ol{T}$, i.e.,
there is a 2-functor $\ol{T}:\Rel\to\Rel$
such that the square
\begin{equation}
\label{eq:extension_square}
\vcenter{
\xymatrix{
\Rel
\ar[0,2]^-{\ol{T}}
&
&
\Rel
\\
\Pos
\ar[0,2]_-{T}
\ar[-1,0]^{({-})_\diamond}
&
&
\Pos
\ar[-1,0]_{({-})_\diamond}
}
}
\end{equation}
commutes.
\item
The functor $T$ satisfies the Beck-Chevalley Condition.
(i.e., it preserves exact squares).
\end{enumerate}
\end{thmC}

A relation lifting $\ol{T}$ of a locally
a locally monotone functor $T$ satisfying the BCC
is computed in the following way.
A relation
$
\xymatrix@1{
R:
X
\ar[0,1]|-{\object @{/}}
&
Y
}
$
can be represented by a certain span
\[
\xymatrix@1{
Y
&
E
\ar[0,1]^{p_1}
\ar[0,-1]_{p_0}
&
X,
}
\]
of monotone maps such that the equality
$R= (p_1)_\diamond\cdot (p_0)^\diamond$ holds
(see~\cite{bkpv11} or~\cite{bkpv12} for details).
Then the lifting
$\ol{T}$ is computed by the following composition of graph
relations:
\[
\xymatrix@1{
\ol{T}R:
TX
\ar[0,2]|-{\object @{/}}^(.6){(Tp_1)^\diamond}
&
&
TE
\ar[0,2]|-{\object @{/}}^{(Tp_0)_\diamond}
&
&
TY
}
\]
In elementary terms, we can check if the elements $\alpha \in TX$ and $\beta \in TY$ are related by the lifted relation $\ol{T}R$ as follows:
\begin{equation}
\label{eq:relation-lifting}
\ol{T}R(\beta,\alpha)
\ \mbox{ iff }\
(\exists w \in TR)
(\beta\leq_{TY}Tp_0(w)
\mbox{ and }
Tp_1(w)\leq_{TX}\alpha).
\end{equation}

Relation lifting behaves well with respect to graph relations, converse relations and restrictions of relations. The easy proofs of these properties follow immediately from the
formula~\eqref{eq:relation-lifting}.
We also give explicit explicit instances of relation liftings for Kripke-polynomial functors, and the convex powerset functor.

\begin{exa}
\label{ex:lifting_prop}
Suppose $T:\Pos\to\Pos$ is a locally monotone functor satisfying BCC, and
$
\xymatrix@1{
R:
X
\ar[0,1]|-{\object @{/}}
&
Y
}
$
is a given relation.
\begin{enumerate}
\item
{\em The relation lifting commutes with graph relations}, i.e.,
the equalities
$\ol{T}f_\diamond = (Tf)_\diamond$
and
$\ol{T}f^\diamond = (Tf)^\diamond$
hold.
\item
{\em Relation lifting commutes with converses}, i.e.,
the equality
$\ol{T}R^\con =(\ol{T^\partial}R)^\con$
(or, equivalently, $\ol{T^\partial}R = (\ol{T}R^\con)^\con$)
holds for every monotone relation $R$.
\item
{\em Relation lifting commutes with taking restrictions.\/}
Given a restriction
\[
\xymatrix{
A
\ar[0,1]|-{\object @{/}}_-{f_\diamond}
&
X
\ar[0,1]|-{\object @{/}}_-{R}
&
Y
\ar[0,1]|-{\object @{/}}_-{g^\diamond}
&
B
\ar @{<-} `u[lll] `[lll]|-{\object @{/}}_-{R(g-,f-)} [lll]
}
\]
of $
\xymatrix@1{
R:
X
\ar[0,1]|-{\object @{/}}
&
Y
}
$ along $f:A\to X$ and $g:B\to Y$, we can apply $\ol{T}$ to obtain
\[
\xymatrix{
T A
\ar[0,1]|-{\object @{/}}_-{(Tf)_\diamond}
&
T X
\ar[0,1]|-{\object @{/}}_-{\ol{T}R}
&
T Y
\ar[0,1]|-{\object @{/}}_-{(Tg)^\diamond}
&
T B
\ar @{<-} `u[lll] `[lll]|-{\object @{/}}_-{\ol{T}R(Tg-,Tf-)} [lll]
}
\]
since $\ol{T}$ commutes with taking the graphs by~(1) above.
\item
{\em Relation lifting for Kripke-polynomial functors\/}
can be computed inductively:
\begin{align*}
\ol{const_X}R\ &= {\leq_X}, \\
\ol{T_0\times T_1}R((\alpha_0,\alpha_1)(\beta_0,\beta_1)) &\mbox{ iff }
\ol{T_0}R(\alpha_0,\beta_0) \mbox{ and } \ol{T_1}R(\alpha_1,\beta_1), \\
\ol{T_0 + T_1}R(\alpha,\beta) = 1 &\mbox{ iff }
\begin{cases}
\alpha \in T_0 Y, \; \beta \in T_0 X, \text{ and } \ol{T_0}R(\alpha,\beta),\\
\alpha \in T_1 Y, \; \beta \in T_1 X, \text{ and } \ol{T_1}R(\alpha,\beta).
\end{cases} \\
\ol{\LLL}R(o,l) &\mbox{ iff }
(\forall y) (y \sqin o \mbox{ implies } (\exists x) (x \sqin l \mbox{ and } R(y,x))), \\
\ol{\UU}R(v,u) &\mbox{ iff }
(\forall x) (u \sqni x \mbox{ implies } (\exists y) (v \sqni y \mbox{ and } R(y,x))).
\end{align*}
\item {\em Relation lifting for the convex powerset functor\/} is computed as follows:
\begin{align*}
\ol{\PP^c}R(v,u) &\mbox{ iff } ((\forall y\in v) (\exists y'\geq_Y y)( R(y',x') \mbox{ and } (\forall x')(R(y',x') \mbox{ implies } (\exists x\in u) x'\leq_X x))\\
   &\mbox{ and } (\forall x\in u) (\exists x'\leq_X x)( R(y',x') \mbox{ and } (\forall y')(R(y',x') \mbox{ implies } (\exists y\in v) y'\geq_Y y)))\\
   &\mbox{ iff } (\forall y\in v)(\exists x\in u) R(y,x) \mbox{ and }  (\forall x\in u)(\exists y\in v) R(y,x). 
\end{align*}

\end{enumerate}
\end{exa}

\begin{rem}
Let the least finitary subfunctor $T_\omega$ of a locally monotone functor $T:\Pos\to\Pos$ satisfying BCC be given 
via the natural transformation $\monoInline{\nu^T}{T_\omega}{T}$. We can take any relation
$
\xymatrix@1{
R:
X
\ar[0,1]|-{\object @{/}}
&
Y
}
$
and restrict its lifting
$
\xymatrix@1{
\ol{T}R:
TX
\ar[0,1]|-{\object @{/}}
&
TY
}
$ along $\nu^T$ to obtain the relation
\[
\xymatrix{
T_\omega X
\ar[0,1]|-{\object @{/}}_-{(\nu^T_X)_\diamond}
&
T X
\ar[0,1]|-{\object @{/}}_-{\ol{T}R}
&
T Y
\ar[0,1]|-{\object @{/}}_-{(\nu_Y^T)^\diamond}
&
T_\omega Y.
\ar @{<-} `u[lll] `[lll]|-{\object @{/}}_-{\ol{T}R(\nu_Y^T-,\nu^T_X-) = \ol{T_\omega}R} [lll]
}
\]
Thus defined, the operation $\ol{T_\omega}$ need not be functorial. In general, we only obtain lax functoriality $\ol{T_\omega} R \cdot \ol{T_\omega} S \subseteq \ol{T_\omega} (R\cdot S)$ (this follows from $(\nu^T_Y)^\diamond \cdot (\nu^T_Y)_\diamond$ being the identity $\leq_{TY}$ for every $Y$). In case $R = f_\diamond$, we even obtain $\ol{T_\omega} (f_\diamond) \cdot \ol{T_\omega} S = \ol{T_\omega} (f_\diamond\cdot S)$. The operation $\ol{T_\omega}$ commutes with converses and with graph relations. To sum up, $\ol{T_\omega}$ is a \emph{lax extension} of $T_\omega$.

Being only a lax extension is not a real obstacle --- in $\Set$, the coalgebraic logic based on a cover modality can be meaningfully defined for an arbitrary lax extension, as shown in~\cite{MV15}. 
Apart from the case when $T_\omega = T$ in Proposition~\ref{prop:express} (2) (and $T_\omega$ is therefore  assumed to satisfy BCC) we only use the operation $\ol{T_\omega}$ on the syntax side of matters in formulating the modal rules of the calculus (this part roughly starts with Definition~\ref{def:collections_L_R} in~\ref{ssec:delta}, and continues in Section~\ref{sec:proof-system}).
There we mostly use the fact that $\ol{T_\omega}R$ is the indicated restriction of $\ol{T}R$ and that the relations $\ol{T}R$ and $\ol{T_\omega}R$ coincide whenever $R$ is a relation between finite posets.

It is worth mentioning though that the finitary Kripke polynomial functors satisfy the BCC, and therefore they admit functorial relation lifting. For example, $\ol{\LLL_\omega}$ is the appropriate restriction of $\ol{\LLL}$, and therefore is computed in the same way. The same goes for the finitary convex powerset functor $\PP^c_\omega$.
\end{rem}

The above calculus of relation liftings allows us to explicitly describe the relation liftings for functors introduced in Section~\ref{sec:functors}.

\begin{exa}
Recall the various ordered coalgebraic structures from Example~\ref{ex:coalgebras}.
\mbox{}\hfill
\begin{enumerate}
\item
The lifting of the functor $T = A \times \Id$ is particularly easy.
Given a relation
$
\xymatrix@1{
R:
X
\ar[0,1]|-{\object @{/}}
&
Y
}
$,
let us denote the lifted relation
\[
\xymatrix@1{
\ol{(A \times \Id)} R:
A \times X
\ar[0,1]|-{\object @{/}}
&
A \times Y
}
\]
by $\ol{R}$. Then the relation
$\ol{R}((b,y),(a,x))$ holds if and only if
$b \leq a$ and $R(y,x)$ holds.
\item
The coalgebras for the functor $T = \Id^A \times \bTwo$ model deterministic ordered
automata. Given a monotone relation
$
\xymatrix@1{
R:
X
\ar[0,1]|-{\object @{/}}
&
Y
}
$
the elements $\alpha = (f,i) \in X^A \times \bTwo$ and
$\beta = (g,j) \in Y^A \times \bTwo$ are related by the lifted relation
$
\xymatrix@1{
\ol{T}R:
TX
\ar[0,1]|-{\object @{/}}
&
TY
}
$
if and only if $j \leq i$ and
$R(g(a),f(a))$ for every $a \in A$.
\item
Given the functor $T = \LLL^A \times \bTwo$
(with $A$ discrete) for lowerset automata and a relation
$
\xymatrix@1{
R:
X
\ar[0,1]|-{\object @{/}}
&
Y
}
$, its lifted relation
$
\xymatrix@1{
\ol{T}R:
TX
\ar[0,1]|-{\object @{/}}
&
TY
}
$
is defined as follows: For an element $\alpha = (l, i)$
in $(\LLL X)^A \times \bTwo$ and
$\beta = (o, j)$ in $(\LLL Y)^A \times \bTwo$
we have that $\ol{T}R(\beta,\alpha)$ holds if and only if
\[
j \leq i
\mbox{ and }
(\forall a \in A) \mathrel{\ol{\LLL} R} (o(a),l(a))
\]
holds, with the latter condition meaning that for every $y$ such that  $y \sqin o(a)$ there is an element $x$ with $x \sqin l(a)$ such that $R(y,x)$ holds. 

In the finitary case, and due to the monotonicity of the relation $R$, the condition can be weakened further using the sets of minimal generators of the lowersets $l$ and $o$. For every $y \in g(o(a))$ there has to be some $x \in g(l(a))$ such that $R(y,x)$ holds. The computation of the lifting stays the same even for an ordered set $A$ of ``inputs''.

\item Given the functor $T = \UU$, two coalgebras $c: X \to \UU X$ and $c': Y \to \UU Y$,  and a relation
$
\xymatrix@1{
R:
X
\ar[0,1]|-{\object @{/}}
&
Y
},
$ 
The lifted relation $\ol{\UU}R$ is defined as follows: 
\[
\ol{\UU}R(c'(y), c(x))  \ \mbox{ iff } (\forall x') (c(x)\sqni x' \mbox{ implies } (\exists y') (c'(y) \sqni y' \mbox{ and } R(y',x'))).
\]
Given the functor $T = \LLL$, two coalgebras $c: X \to \LLL X$ and $c': Y \to \LLL Y$,
\[
\ol{\LLL}R (c'(y),c(x)) \mbox{ iff }
(\forall y') (y' \sqin c'(y) \mbox{ implies } (\exists x') (x' \sqin c(x) \mbox{ and } R(y,x))).
\]
Given the functor $T = \PP^c$, two coalgebras $c: X \to \PP^c X$ and $c': Y \to \PP^c Y$,
\begin{align*}
\ol{\PP^c}R (c'(y),c(x)) \mbox{ iff }
&(\forall y'\in c'(y))(\exists x'\in c(x)) R(y',x')\\
\mbox{ and }
&(\forall x'\in c(x))(\exists y'\in c'(y)) R(y',x').
\end{align*}
\end{enumerate}
\end{exa}

\section{Moss' logic for ordered coalgebras}
\label{sec:logic}
In this section we introduce a logic for ordered coalgebras parametric in the coalgebra functor $T$. The syntax of the logic will be finitary, the propositional part will be given by the conjunction and disjunction connectives, and the modal part will consist, up to Subsection~\ref{ssec:delta}, of a single modality $\nabla$ of arity $T_{\omega}^\partial$. The semantics of the logic will be given by monotone valuations over $T$-coalgebras, with the semantics of the modality $\nabla$ given by $T^\partial$ relation lifting.
We will show that this language is always adequate for $T$-coalgebras, and adequate and expressive for $T_\omega$-coalgebras. This means that the logic, in case of finitary coalgebras, has the Hennessy-Milner property.

\begin{asm}
\label{ass:5.1}
From this section on, we fix a locally monotone functor $T: \Pos\to\Pos$, its least finitary subfunctor $\monoInline{\nu^T}{T_\omega}{T}$, 
and assume $T$ (and therefore $T^\partial$) preserve exact squares. We assume that $T_{\omega}^\partial$ preserves finite intersections, and therefore admits a finite base.

\end{asm}
In particular, all Kripke polynomial functors of~\ref{eq:kripke-polynomial} in place of the functor $T$ comply with this assumption.

\subsection{The syntax of the coalgebraic language}
\label{ssec:lang}

The syntax of Moss' logic for ordered coalgebras will be based on finitary conjunction
and disjunction connectives, and a single finitary modality. It will be convenient
to regard the arities of both the connectives and the modality
as finitary functors. A similar approach has been already used
in~\cite{kkv12} for the Moss' logic in sets. There are some differencies worth mentioning though: First, when passing from
sets to the enriched setting of posets, to be as general as possible, one should start with a poset of propositional variables. Second, one has to be more careful about the precise shape of the arities, namely:
\begin{enumerate}
\item
In $\Set$, the arity of the finitary conjunction and disjunction
is given by the finitary powerset functor $P_\omega:\Set\to\Set$.
This means that, e.g., the conjunction is a map
$\bigwedge:P_\omega\Lang\to\Lang$ where $\Lang$ denotes the
set of all formulas. In $\Pos$, the natural choice for the
arity of the conjunction is $\UU_\omega$, whereas for the disjunction it
is $\LLL_\omega$.
\item
In $\Set$, the cover modality $\nabla$ has the coalgebraic
functor $T_\omega$ as its arity. In $\Pos$, one needs to use the
dual $T^\partial_\omega$ of $T_\omega$ for type-checking reasons,
as will become clear in Subsection~\ref{ssec:semant} below.
\end{enumerate}
\begin{rem}
Propositional variables are proposition place-holders. As such, the most natural (and the most free) choice is to simply start with a set of them (i.e., a discrete poset). This is indeed the setting of the most of our examples. The poset setting however allows to provide place-holder patterns, which consequently the semantics must respect.
This can be seen and used as prescribing order between formulas (which is normally done afterwords by the logic and induced by semantics of logical connectives, usually in a form of assuming implications --- and we can do so by using theories).
But what if implications are not at hand, or, we do not want to use (possibly infinite) theories? What if we used place-holder patterns instead? Not many examples are out there, but let us provide some simple intuition behind the possibilities this opens:
\begin{itemize}
    \item[-] Consider the $\{\wedge,\vee\}$-fragment of classical propositional logic, but take $\{ p_i\leq q_i | i\in N \}$ for the poset of propositional variables. This would produce a situation where ``p's always entail q's''.
    \item[-] Similarly, atomic propositions can ``code'' scalable properties often used in questionnaires (where mostly 3--7 options in a form of ``degrees'' are given --- e.g.\ heavy smoker, smoker, occasional smoker). For such settings, the poset of propositional variables can be chosen as e.g.\ $\{ p_i\leq q_i\leq r_i | i\in N \}$.
\end{itemize}
Examples where a partial order (namely inclusion) is built in syntax can be found in literature which consider ordered sorts or predicates: Order-Sorted Predicate Logic has been proposed as a formal knowledge representation languages for handling structural knowledge, such as the classification of objects~\cite{Ober89,Cohn87}, and extended by~\cite{Kaneiwa04} which formalizes a logic programming language with not only a sort hierarchy, but also a predicate hierarchy.
\end{rem}
\begin{rem}\label{rem:arities}
The naturality of the choice of the arity of the finitary conjunction to be $\UU_\omega$, i.e.\ a finitely generated upwards closed set of formulas, can be demonstrated for example as follows: Assume we have a poset of atomic formulas at hand, and think of a conjunction as eventually becoming the infimum in the free algebra of formulas. Whenever $\bigwedge \phi\leq b$, we expect $\bigwedge (\phi\cup b) \equiv \bigwedge \phi$ --- the conjunction is immune to adding arguments greater than the conjunction. Similarly for disjunction and its arity being $\LLL_\omega$.

The choice of arities is in accord with how the free distributive lattice over a poset $X$ can be constructed on $\LLL_\omega\UU_\omega X$: first take the free finite meet completion on $\UU_\omega X$, then the free finite join completion (cf.~\cite{PTJ83}).
\end{rem}
Therefore the formulas of the coalgebraic language should
have the following intuitive description in BNF:
\begin{equation}
\label{eq:intuitive-BNF}
a::= p\mid\bigwedge\phi\mid\bigvee\psi\mid\nabla\alpha
\end{equation}
where $p$ is an atom, $\phi={\uparrow}\{ a_1,\dots,a_k\}$
is a finitely generated upperset of formulas,
$\psi={\downarrow}\{ a_1,\dots,a_k\}$ is a finitely generated lowerset
of formulas, and $\alpha$ is a ``$T^\partial$-tuple'' of formulas.
There is a slight technicality however: since we work in posets
we expect to obtain a poset of formulas. The precise definition
of formulas is achieved by the free algebra construction in
the category $\Pos$.

\begin{defi}[Formulas]
\label{def:lang}
Fix a poset $\At$ of propositional atoms. The {\em language\/}
$\Lang$ is given as an algebra for $\UU_\omega+\LLL_\omega+T^\partial_\omega$,
free on $\At$. The components of the algebraic structure
$a:\UU_\omega\Lang+\LLL_\omega\Lang+T^\partial_\omega\Lang\to\Lang$
will be denoted by
\[
\bigwedge: \UU_{\omega}\Lang\to\Lang,
\quad
\bigvee: \LLL_{\omega}\Lang\to\Lang,
\quad \nabla_{T^\partial_\omega}: T^\partial_\omega\Lang\to\Lang
\]
\end{defi}

\begin{rem}
\label{rem:syntax}
\mbox{}\hfill
\begin{enumerate}
\item
The language $\Lang$ is a poset by its construction.
Moreover, an algebra for $F=\UU_\omega+\LLL_\omega+T^\partial$
free on $\At$ can be defined by a colimit of a transfinite chain
\[
w_{i,i+1} : W_i \to W_{i+1}
\]
where $W_0 = \At$, $W_{i+1} = F(W_i)+\At$ and the connecting
morphisms are defined in the obvious way:
$w_{0,1}:\At\to F(\At)+\At$ is the coproduct injection
and $w_{i+1,i+2}=F(w_{i,i+1})+\At$.

The above chain $w_{i,i+1}$ has $\Lang$ as its colimit;
we denote the colimit injections by $w_i : W_i \to \Lang$.
\item
The transfinite construction of $\Lang$ also shows that
the ``intuitive BNF'' of~\eqref{eq:intuitive-BNF} works.
More in detail, one can show that each formula $a$ in $\Lang$
has a unique finite depth.
Indeed, for every $\monoInline{a}{\bOne}{\Lang}$ there is a least $i$ such that $a \in w_i$.
There are two cases for a fixed $a$ in $\Lang$:
\begin{enumerate}
\item
$i = 0$ means that $a$ is an atom.
\item
$i$ is positive, i.e., $i = k+1$.
Then $a \in w_{k+1}$ holds.
The formula $a$ is not an atom by the definition of $i$; hence $a$ is either
a conjunction, a disjunction or a nabla of an object in $W_k$.

The finite poset of direct subformulas of $a$ can be
obtained as the subobject $\base^F_{F(W_k)}(w_{k+1}(a))$.
\end{enumerate}
\item
We defined the arity of a conjunction to be a finitely
generated upperset of formulas, which itself may not
be finite. We often abuse the notation in writing the formulas,
and list only the finite set of generators of the upperset
to keep the description finite. For example, by
$\bigwedge\{a,b\}$ we implicitly mean the conjunction
applied to the upperset ${\uparrow}{\{ a,b\}}$
generated by $a$ and $b$.

Notice that, for finite sets $A$ and $B$ of formulas such that
${\uparrow}A={\uparrow}B$ holds, the equality
$\bigwedge A=\bigwedge B$ is built into the relaxed
notation. Similarly for disjunctions, the equality ${\downarrow} A = {\downarrow} B$ implies the syntactic equality $\bigvee A = \bigvee B$. In both conjunction and disjunction, the commutativity, associativity and idempotence properties are built into the notation as well, due to the choice of their arity. Moreover, whenever $\phi\subseteq \phi'$ in $\UU_\omega \Lang$, $\bigwedge \phi' \leq_\Lang \bigwedge \phi$, and, similarly whenever $\psi\subseteq \psi'$ in $\LLL_\omega \Lang$, $\bigvee \psi \leq_\Lang \bigvee \psi'$.
\end{enumerate}
\end{rem}
For a given formula $a$, the above remark allows us to define
its finite poset of subformulas and its modal depth inductively.

\begin{defi}[Subformulas and modal depth]
\label{def:subflas}
Given a formula $a$ in $\Lang$, the (finite) subobject $\sfor(a)$ of $\Lang$ is called the subobject of \emph{subformulas} of $a$, and is defined inductively as follows.
Simultaneously we define the \emph{modal depth} $d(a)$.
\begin{enumerate}
\item
For $a$ in $\At$, put $\sfor(a) = \monoInline{a}{\bOne}{\Lang}$
 and $d(a) = 0$.
\item
For $a$ of the form $\bigwedge\phi$,
put
\begin{eqnarray*}
\sfor(a)
& = &
\bigcup\limits_{z\in \base(\phi)}\sfor(z)\cup a
\\
d(a)
& = &
\max\{d(b)\mid\phi\sqni b\}
\end{eqnarray*}
\item
For $a$ of the form $\bigvee\psi$,
put
\begin{eqnarray*}
\sfor(a)
& = &
\bigcup\limits_{z\in \base(\psi)}\sfor(z) \cup a
\\
d(a)
& = &
\max\{d(b)\mid b\sqin\psi\}
\end{eqnarray*}
\item
For $a$ of the form $\nabla\alpha$,
put
\begin{eqnarray*}
\sfor(a)
& = &
\bigcup\limits_{w\in \base(\alpha)}\sfor(w) \cup a
\\
d(a)
& = &
\max\{d(w)\mid w \in \base(\alpha)\}+1
\end{eqnarray*}
\end{enumerate}
Above, all the unions are taken in the lattice of subobjects,
see Remark~\ref{rem:sublattice}.
\end{defi}

\begin{exa}
\label{ex:streamsyntax}
Let us describe the syntax of the logic for the functor
$T = A \times\Id$, where as $A$ we take the poset
$\bTwo = \{ 0 < 1 \}$. If we
are given a poset $\At$ of atomic propositions,
the syntax is defined inductively as follows:
\[
a::=
p
\mid
\bigwedge \phi
\mid
\bigvee \psi
\mid
\nabla_{T^\partial_\omega}(n,a).
\]
In the above, $p$ is any atomic proposition,
$\phi={\uparrow}\{ a_1,\dots,a_k\}$ is a finitely
generated upperset and
$\psi={\downarrow}\{ a_1,\dots,a_k\}$
is a finitely generated lowerset of formulas,
and the $n$ in $\nabla_{T^\partial_\omega}(n,a)$ is an
element of $\bTwo$. Observe that $\nabla$ is monotone with respect to the first argument.

Hence the more relaxed description of the syntax
can be given by
\[
a::=
p
\mid
\bigwedge\{ a_1,\dots,a_k\}
\mid
\bigvee\{ a_1,\dots,a_k\}
\mid
\nabla_{T^\partial_\omega}(n,a).
\]
\end{exa}

\subsection{The semantics of the coalgebraic language}
\label{ssec:semant}
The above language $\Lang$ will be interpreted in coalgebras
for the functor $T$. More precisely, given a coalgebra $c: X\to TX$,
the {\em semantics\/} will be given by a monotone relation
$
\xymatrix@1{
\Vdash_0:
\At
\ar[0,1]|-{\object @{/}}
&
X^\op
}
$
called a monotone {\em valuation\/} of propositional atoms. The valuation being monotone, $x\leq y$ and $x\Vdash_0 p$ implies $y\Vdash_0 p$, and $p\leq_\Lang q$ and $x\Vdash_0 p$ implies $x\Vdash_0 q$.

We extend $\Vdash_0$ to obtain the semantics
$
\xymatrix@1{
\Vdash:
\Lang
\ar[0,1]|-{\object @{/}}
&
X^\op
}
$
of an arbitrary formula $a$ by induction on $a$ as follows:
\begin{eqnarray*}
x\Vdash\bigwedge\phi\
& \mbox{iff} &\
(\forall a)(\phi\sqni a\mbox{ implies }x\Vdash a)
\\
x\Vdash\bigvee\psi\
& \mbox{iff} &\
(\exists a)(a\sqin\psi\mbox{ and }x\Vdash a)
\\
x\Vdash\nabla_{T^\partial_\omega}\alpha\
& \mbox{iff} &\
c(x) (\mathrel{\ol{T^\partial}{\Vdash}\cdot (\nu^{T^\partial}_{\Lang})_{\diamond}}) \alpha
\end{eqnarray*}

\begin{rem}
The semantics for conjunction and disjunction is standard. The semantics for the modal formula $\nabla_{T^\partial_\omega}\alpha$ in a state $x$ is given by checking whether the $T^\partial$-lifted $\Vdash$ relation contains $c(x)$ and $\alpha$: for type-checking reasons, this is achieved by composing the lifted $\Vdash$ relation with the graph of the natural inclusion of $T^\partial_\omega$ in $T^\partial$, i.e.\ $(\nu^{T^\partial}_{\Lang})_{\diamond}$:
  \[
\xymatrix@1{
T^\partial_\omega\Lang
\ar[0,2]|-{\object @{/}}^-{(\nu^{T^\partial}_{\Lang})_{\diamond}}
&
&
T^\partial\Lang
\ar[0,2]|-{\object @{/}}^-{\ol{T^\partial}{\Vdash}}
&
&
(TX)^\op.
}
  \]
We can also equivalently write it like this:
  \[
x\Vdash\nabla_{T^\partial_\omega}\alpha\
\ \ \mbox{iff} \ \ 
c(x) \mathrel{\ol{T^\partial}{{\Vdash}}} \nu^{T^\partial}_{\Lang}(\alpha).
  \]
The semantics of nabla is indeed well defined: in the inductive process of defining the relation $\Vdash$ we use the fact that relation lifting commutes with restrictions (see Definition~\ref{def:operations-on-relations} and Example~\ref{ex:lifting_prop}). In particular, to compute the lifted relation it is enough to have $\Vdash$ restricted to the base of $\alpha$, whose semantics has been defined previously. Unravelling this yields the following simplification, using the lifted restricted relation ${\Vdash}(-,\base(\alpha)-)$: Let $\monoInline{\base(\alpha)}{Z}{\Lang}$ and $z\in Z$ be the unique element with $T^\partial\base(\alpha)(z) = \alpha$.
  \[
x\Vdash\nabla_{T^\partial_\omega}\alpha\
\ \  \mbox{iff} \ \
c(x) \mathrel{\ol{T^\partial}{{\Vdash}(-,T^\partial\base(\alpha)-)}} z.
  \]
\end{rem}

\begin{rem}
\label{rem:semantic-preorder}
Recall from Remark~\ref{rem:syntax} that $\Lang$ is a poset.
The above semantics allows us to define a monotone relation
$
\xymatrix@1{
\sqsubseteq:
\Lang
\ar[0,1]|-{\object @{/}}
&
\Lang
}
$
by putting:
\begin{equation}
\label{eq:semantic-preorder}
\mbox{$a\sqsubseteq b\ $ iff \ for all $c, x$ and all valuations:
$x\Vdash a$ implies $x\Vdash b$.}
\end{equation}
The relation $\sqsubseteq$ is reflexive (i.e.,
the relation ${\leq_\Lang}$ is smaller than ${\sqsubseteq}$)
and transitive
(i.e., the composite relation ${\sqsubseteq}\cdot {\sqsubseteq}$
is smaller than ${\sqsubseteq}$). Therefore $\sqsubseteq$
allows us to define an order-quotient
${\mathcal{Q}}=\Lang/{\sqsubseteq}$ of $\Lang$. The elements
of the poset $\mathcal{Q}$ are equivalence classes
$[a]_\equiv$ where
\[
a\equiv b\ \mbox{iff}\ a\sqsubseteq b \mbox{ and } b\sqsubseteq a
\]
and $\leq_{\mathcal{Q}}$ is the least order such that
$a\sqsubseteq b$ entails $[a]_\equiv\leq_{\mathcal{Q}} [b]_\equiv$.

The quotient $\mathcal{Q}$ carries an algebra structure for
$\UU_\omega+\LLL_\omega+T^\partial_\omega$ that is derived from the algebra structure on $\Lang$.
\end{rem}

\begin{exa}
\label{ex:nabla}
\mbox{}\hfill
\begin{enumerate}
\item
{\em Nabla for the functor $T = A\times\Id$}.
Observe that, $T=T_\omega$, and the functor dual to $T_\omega$ is:
\[
T^\partial_\omega X = (A\times X^\op)^\op = A^\op\times X.
\]
If we fix a coalgebra
$c = \langle out, next \rangle: X\to A \times X$, then it holds that
\[
x \Vdash \nabla_{T^\partial_\omega}(a,b) \quad \mbox{iff} \quad
out(x) \geq_A a \;\mbox{and}\; next(x) \Vdash b
,
\]
since the lifting of the semantics relation
$
\xymatrix@1{
\Vdash:
\Lang
\ar[0,1]|-{\object @{/}}
&
X^\op
}
$
is the relation
\[
\xymatrix@1{
(\leq_{A^\op} \times \Vdash) :
A^\op \times\Lang
\ar[0,1]|-{\object @{/}}
&
A^\op \times X^\op
}
.\]
The monotonicity of $\Vdash$ says that
\begin{enumerate}
\item
If $x \leq_X x^\prime$ and $x \Vdash\nabla_{T^\partial_\omega}(a,b)$, then
$x^\prime \Vdash\nabla_{T^\partial_\omega}(a,b)$, and
\item
if $a \geq_A a'$ and $b \leq_\Lang b'$, then
$x \Vdash
\nabla_{T^\partial_\omega}(a,b)$ implies that
$x \Vdash\nabla_{T^\partial_\omega}(a',b')$.
\end{enumerate}
\item
{\em Nabla for the functor $T = \LLL^A \times \bTwo$}.
Observe that the arity of nabla is:
\[
T^\partial_\omega \Lang=(\UU_\omega \Lang)^A \times \bTwo.
\]
By definition, $x \Vdash \nabla_{T^\partial_\omega} \alpha$ holds
if and only if $c(x) (\mathrel{\ol{T^\partial} {\Vdash} \cdot (\nu^{T^\partial}_{\Lang})_{\diamond}}) \alpha$ holds. Suppose
$c(x)$ is a tuple $(l, i)$ and $\alpha$ is a tuple $(o, j)$. Then
this is the case if and only if
\[
i \geq j
\mbox{ and }
(\forall a \in A) (\forall \phi) (o(a) \sqni \phi  \mbox{ implies } (\exists x) (x \sqin l(a) \mbox{ and } x \Vdash \phi)).
\]
We can again use the monotonicity of the semantics relation $\Vdash$ to weaken the above condition:
$c(x) (\mathrel{\ol{T^\partial} {\Vdash} \cdot (\nu^{T^\partial}_{\Lang})_{\diamond}}) \alpha$ holds if and only if
\[
i \geq j
\mbox{ and }
(\forall a \in A) (\forall \phi \in g(o(a)))(\exists x \in g(l(a))) 
\ x \Vdash \phi.
\]
Here $g(l(a))$ and $g(o(a))$ are, again, the generators for the lowerset $l(a)$ in $X$
and \emph{upperset} $o(a)$ in $\Lang$, respectively.
\item
Recall from Example~\ref{ex:coalgebras} that a frame for positive modal logic is a poset $X$ equipped with a monotone relation
$
\xymatrix@1{
R:
X
\ar[0,1]|-{\object @{/}}
&
X
}
$
that gives rise to two coalgebras
\[
c: X \to\UU X \ \ \mbox{and}\ \ d: X \to \LLL X
\]
defined by $c(x)=\{y\ |\ R(x,y)\}$ and $d(x)=\{y\ |\ R(y,x)\}$. The modalities $\Box$ and
$\blacklozenge$ defined by the equivalences
\begin{eqnarray*}
x\Vdash \Box a
& \mbox{iff} &
\ (\forall y)\ (R(x,y) \mbox{ implies } y\Vdash a)
\\
x\Vdash \blacklozenge a
& \mbox{iff} &
\ (\exists y)\ (R(y,x)\mbox{ and } y\Vdash a)
\end{eqnarray*}
are adjoint in the sense that $a\sqsubseteq\Box b$ holds if and only if $\blacklozenge a\sqsubseteq b$
holds, where $\sqsubseteq$ is the semantic preorder as defined in~\eqref{eq:semantic-preorder}.

Now, using the definition of the modalities $\nabla_{\LLL_\omega}$ and
$\nabla_{\UU_\omega}$ and the corresponding liftings
of the semantics $\Vdash$, we see that
\begin{eqnarray*}
x\Vdash\nabla_{\LLL_\omega}\alpha
&\ \mbox{iff} &
c(x) (\mathrel{\ol{\LLL}{\Vdash} \cdot (\nu^{\LLL}_{\Lang})_{\diamond}}) \alpha \\
&\ \mbox{iff}  &
(\forall y)(y \sqin c(x) \mbox{ implies } (\exists a)\ (a \sqin\alpha \mbox{ and } y\Vdash a)),\\
x\Vdash\nabla_{\UU_\omega}\beta
&\ \mbox{iff} &
d(x) (\mathrel{\ol{\UU}{\Vdash} \cdot (\nu^{\UU}_{\Lang})_{\diamond}}) \beta \\
&\ \mbox{iff}  &
(\forall b) (\beta \sqni b \mbox{ implies } (\exists y)\ (d(x) \sqni y \mbox{ and } y \Vdash b) ).
\end{eqnarray*}
Therefore
\begin{align*}
\nabla_{\LLL_\omega}\alpha
&\mbox{ can be expressed as }
\Box\bigvee\alpha, \\
\nabla_{\UU_\omega}\beta
&\mbox{ can be expressed as }
\bigwedge\blacklozenge\beta,
\end{align*}
and also conversely
\begin{align*}
\Box a
&\mbox{ can be expressed as }
\nabla_{\LLL_\omega}\{a\}, \\
\blacklozenge b
& \mbox{ can be expressed as }
\nabla_{\UU_\omega}\{b\}.
\end{align*}
\item
Let $T = \PP^c$, and $T^\partial_\omega = \PP^c_\omega$. Recall from Example~\ref{ex:coalgebras} that $\PP^c$ coalgebras can be used as semantics for positive modal logic. For a coalgebra $c: X \to \PP^c X$ and a monotone valuation, we have
\begin{eqnarray*}
x\Vdash \Box a
& \mbox{iff} &
\ (\forall y)\ (R(x,y) \mbox{ implies } y\Vdash a)
\\
x\Vdash \Diamond a
& \mbox{iff} &
\ (\exists y)\ (R(x,y)\mbox{ and } y\Vdash a).
\end{eqnarray*}
By the definition of the semantics for $\nabla_{\PP^c_\omega}$ and the relation lifting, we see that for $\alpha$ in $\PP^c_\omega \Lang$
\begin{align*}
x\Vdash \nabla_{\PP^c_\omega}\alpha \ &\mbox{ iff }\  c(x) 
\mathrel{
(\ol{\PP^c}\Vdash \cdot (\nu^{\PP^c}_{\Lang})_{\diamond} )
}
\alpha \\
 &\mbox{ iff }
(\forall x'\in c(x))(\exists a\in \alpha) x'\Vdash a \mbox{ and }
(\forall a\in \alpha)(\exists x'\in c(x)) x'\Vdash a.
\end{align*}
As $\alpha$ is a finitely generated convex subset of formulas, we can create a finitely generated upperset of formulas ${\uparrow}\alpha$, and  a finitely generated lowerset of formulas ${\downarrow}\alpha$. Then
\[
\nabla_{\PP^c_\omega}\alpha\  \mbox{ can be expressed as }\  \Box\bigvee{\downarrow}\alpha \wedge \bigwedge \Diamond {\uparrow}\alpha,
\]
and also conversely
\begin{align*}
\Box a
&\mbox{ can be expressed as }
\nabla_{\PP^c_\omega}\{a\}\vee \nabla_{\PP^c_\omega}\emptyset, \\
\Diamond a
& \mbox{ can be expressed as }
\nabla_{\PP^c_\omega}\{a,\top\},
\end{align*}
just the same way as it is in the classical finitary Moss' logic~\cite{kkv12}.

\item Recall the semantics of small description logic EL from Example~\ref{ex:coalgebras}
The syntax consists of a poset $N^c$ of concept names ordered by concept subsumptions of the form $A\sqsubseteq B$, a poset $N^r$ of role names ordered by role subsumptions of the form $r\sqsubseteq s$, and the following grammar of concepts:
\[
C:=\ A\ |\ \top \ |\ C\sqcap D\ |\ \exists r.C
\]
This results in a poset $\Lang$ of concept expressions. The semantics is based on a discrete poset (i.e.\ a set) $\Delta$ together with: (ii) interpretations $C^I$ of concept names by lowersets of $\Delta$ (i.e.\ subsets of $\Delta$) respecting their order: $C\sqsubseteq D$ implies $C^I\subseteq D^I$, (ii) interpretations $r^I$ of role names by binary relations on $\Delta$, again respecting their order: $r\sqsubseteq s$ implies $r^I\subseteq s^I$. We can see each structure of this kind as a coalgebra
\[
c: \Delta \to (\LLL \Delta)^{N^r} \times \bTwo^{N^c}
\]
which, for each $d\in\Delta$, assigns: for each role name $r$ the lowerset (i.e.\ subset) $\{e\ |\ (d,e)\in r^I\}$, and for each concept name $C$ the value $1$ if $d\in C^I$, and the value $0$ otherwise. Both assignments are monotone. 

For the sake of this example, we shall extract the $v: \Delta \to \bTwo^{N^c}$ part of the coalgebra structure and treat it as a monotone valuation of concept names 
$
\xymatrix@1{
\Vdash:
N^c 
\ar[0,1]|-{\object @{/}}
&
\Delta^\op
}
$
instead.
The relation extends inductively to all concept expressions as follows:
\begin{align*}
    d &\Vdash \top \\
    d &\Vdash C\sqcap D \mbox{ iff }\  d\Vdash C \mbox{ and } d\Vdash D\\
    d &\Vdash \exists r.C \ \ \mbox{ iff }\  \exists e (e\in c(d)(r) \mbox{ and } e\Vdash C)
\end{align*}
The monotonicity of $\Vdash$ entails that if $C\sqsubseteq_\Lang D$ and $d\Vdash C$, than also $d\Vdash D$.

Each $\beta\in (\UU_\omega\Lang)^{N^r}$ can be seen as a monotone assignment of uppersets of concepts to role names.
Unravelling the semantics of the modality
$\nabla_{(\UU_\omega)^{N^r}}$ on a coalgebra $c: \Delta \to (\LLL \Delta)^{N^r}$, we see that 
\begin{eqnarray*}
d\Vdash\nabla_{(\UU_\omega)^{N^r}}\beta
&\ \mbox{iff} &
c(d) (\mathrel{\ol{(\UU)^{N^r}}{\Vdash} \cdot (\nu^{\UU^{N^r}}_{\Lang})_{\diamond}}) \beta \\
&\ \mbox{iff}  &
  (\forall r)((\forall C) (\beta(r) \sqni C \mbox{ implies } (\exists e)\ (c(d)(r) \sqni e \mbox{ and } e \Vdash C) )).
\end{eqnarray*}
Therefore
\begin{align*}
\nabla_{(\UU_\omega)^{N^r}}\beta
&\mbox{ can equivalently be expressed as }
\bigwedge_r(\bigwedge\exists r.[\beta(r)]),
\end{align*}
and also conversely
\begin{align*}
\exists r. C
& \mbox{ can equivalently be expressed as }
\nabla_{(\UU_\omega)^{N^r}}\beta,
\end{align*}
where $\beta \in (\UU_\omega\Lang)^{N^r}$ is a monotone assignment of uppersets of concepts to role names such that $\beta(r) = C{\uparrow}$ and $\beta(s) = \top$ for all $s\neq r$.

\item
Recall frames for distributive substructural logics from
Example~\ref{ex:coalgebras}.
We restrict ourselves to ternary relations $R$
that generate coalgebras of the form
$c_\otimes: X\to\LLL(X\times X)$.

The polynomial coalgebra functor
$T=\LLL(\Id\times \Id)$ is locally monotone and
satisfies BCC.\@ Its relation lifting is easy to compute,
using the properties listed in
Example~\ref{ex:lifting_prop}.
The semantics of the nabla modality with the arity
$T_\omega^\partial=\UU_\omega(\Id\times\Id)$,  works as follows:
\[
x\Vdash \nabla_{T^\partial_\omega}\alpha
\mbox{ iff }
(\forall (a_0,a_1)\sqin\alpha)(\exists (x_0,x_1)\sqin c(x))
\ (x_0\Vdash a_0\ \&\ x_1\Vdash a_1).
\]
Therefore
\[
\nabla_{T^\partial_\omega}\alpha
\mbox{ can be expressed as }
\bigwedge\limits_{\alpha\sqni(a_0,a_1)}(a_0\otimes a_1),
\]
and conversely
\[
(a_0\otimes a_1)
\mbox{ can be expressed as }
\nabla_{T^\partial_\omega}\{(a_0,a_1)\}.
\]
\end{enumerate}
\end{exa}

\subsection{Hennessy-Milner property}
\label{ssec:express}

We now turn to proving that the finitary language defined in
Subsection~\ref{ssec:lang} is adequate and, in the finitary case $T$=$T_\omega$, also expressive for the following notion
of \emph{similarity}, which is defined in terms of relation lifting.

\begin{defi}
\label{def:simul}
We fix two pointed models $(c,\Vdash_c,x_0)$ and $(d,\Vdash_d,y_0)$,
i.e., we fix coalgebras $c:X\to TX$ and $d:Y\to TY$, valuations
$\Vdash_c$ and $\Vdash_d$, and $x_0 \in X$ and $y_0 \in Y$.
\begin{enumerate}
\item
A relation
$
\xymatrix@1{
S:
Y
\ar[0,1]|-{\object @{/}}
&
X
}
$
is called a $T$-\emph{simulation}
(from $(d,\Vdash_d,y_0)$ to $(c,\Vdash_c,x_0)$), if
the following three conditions are satisfied:
\begin{enumerate}
\item
$S(x_0,y_0)$ holds.
\item
For any $x\in X$ and $y\in Y$,
$S(x,y)$ implies $\ol{T}S(c(x),d(y))$.
\item
If $S(x,y)$ holds, then $x\Vdash_c p$ implies
$y\Vdash_d p$, for each atom $p$ in $\At$.
\end{enumerate}
And we say that $(d,\Vdash_d,y_0)$ {\em simulates\/}
$(c,\Vdash_c,x_0)$, if there is a simulation
$
\xymatrix@1{
S:
Y
\ar[0,1]|-{\object @{/}}
&
X
}
$
from $(d,\Vdash_d,y_0)$ to $(c,\Vdash_c,x_0)$.
\item
We say that $(d,\Vdash_d,y_0)$ is {\em modally stronger\/}
than $(c,\Vdash_c,x_0)$, if $x_0\Vdash_c a$ implies $y_0\Vdash_d a$,
for each formula $a$.
\end{enumerate}
\end{defi}

\begin{rem}
The notion of simulation given by the monotone relation lifting
coincides with the one given by Worrell in~\cite{worr}
in the enriched setting of $\V$-categories (of which preorders, hence
also posets, are a special case for $\V = \bTwo$). It is also shown
in~\cite{worr} that similarity coincides with the preorder on the
final coalgebra, whenever the final coalgebra exists.
\end{rem}

\begin{exa}\label{ex:streamsimulation}
Fix the empty poset $\At$ of atomic propositions.
Taking two coalgebras $c: X \to A \times X$ and $d: Y \to A \times Y$
with two distinguished states $x_0 \in X$ and $y_0 \in Y$,
the notion of
simulation yields that $(d,\Vdash,y_0)$
simulates $(c,\Vdash,x_0)$ if the infinite
stream $\gamma$ obtained as the behaviour of $x_0$
is pointwise smaller than the infinite stream $\delta$ obtained
as the behaviour of $y_0$.
\end{exa}

\begin{prop}
\label{prop:express}
Assume $T$ complies with the Assumption~\ref{ass:5.1}, and $\Lang$ is defined as in Definition~\ref{def:lang}.
\begin{enumerate}
\item
The language $\Lang$ is adequate for $T$-coalgebras: if $(d,\Vdash_d,y_0)$ simulates $(c,\Vdash_c,x_0)$, then
$(d,\Vdash_d,y_0)$ is modally stronger than $(c,\Vdash_c,x_0)$.
\end{enumerate}
and
\begin{enumerate}
\setcounter{enumi}{1}
\item
The language $\Lang$ is expressive for $T_\omega$-coalgebras:
if $(d,\Vdash_d,y_0)$ is modally stronger than $(c,\Vdash_c,x_0)$,
then $(d,\Vdash_d,y_0)$ simulates $(c,\Vdash_c,x_0)$.
Moreover, the relation
``being modally stronger'' is a $T_\omega$-simulation.
\end{enumerate}
\end{prop}
\begin{proof}
(1)
Let us assume that
$(d,\Vdash_d,y_0)$ simulates $(c,\Vdash_c,x_0)$
via a simulation
$
\xymatrix@1{
S:
Y
\ar[0,1]|-{\object @{/}}
&
X
}
$.
We need to prove that $x_0\Vdash_c a$ implies $y_0\Vdash_y a$
for every formula $a$.
The adequacy is proved by induction on the complexity
of a given formula $a$.

The case when $a$ is an atom is immediate from $S$ being
a simulation, the cases of conjunction and disjunction
are easy. For the induction step for $a=\nabla\alpha$, let
us assume that $x_0\Vdash_c \nabla\alpha$ holds.
The induction hypothesis states that for every $x \in X$, $y \in Y$ and $z \in \base(\alpha)$ the following implication holds:
\[
\mbox{If } x \Vdash_c z \mbox{ and } S(x,y), \mbox{ then } y  \Vdash_d z.
\]
The induction hypothesis can equivalently be described as a lax triangle
\begin{equation}
\label{eq:adequacy-ih}
\xymatrixcolsep{4pc}
\xymatrix{
\wt{Z}
\ar[0,2]|-{\object @{/}}^-{\Vdash_d(-,\base(\alpha)-)}
\ar[1,1]|-{\object @{/}}_-{\Vdash_c(-,\base(\alpha)-)\phantom{MM}}
&
&
Y^\op
\\
&
X^\op
\ar[-1,1]|-{\object @{/}}_-{S^\con}
\ar@{}[-1,0]|{\uparrow}
&
}
\end{equation}
denoting the base of 
$
\monoInline{\alpha}
           {Z}
           {T^\partial_\omega\Lang}
$
by
$
\monoInline{\base(\alpha)}
           {\wt{Z}}
           {\Lang}
$.

We shall prove that for any $x$ in $X$, $y$ in $Y$, and $w \in T^\partial_\omega \base(\alpha)$, the following implication holds:
\begin{equation}
\label{eq:adequacy}
\ol{T}S(c(x),d(y))
\mbox{ and }
c(x)(\mathrel{\ol{T^\partial}{\Vdash_c}\cdot (\nu^{T^\partial}_{\Lang})_{\diamond}}) w
\; \mbox{ implies } \;
d(y)(\mathrel{\ol{T^\partial}{\Vdash_d}\cdot (\nu^{T^\partial}_{\Lang})_{\diamond}}) w
\end{equation}
The relation $\ol{T}S(c(x_0),d(y_0))$ holds since $S$ is a simulation. Moreover, $c(x_0)(\mathrel{\ol{T^\partial}{\Vdash_c} \cdot (\nu^{T^\partial}_{\Lang})_{\diamond}}) \alpha$ holds since $x_0\Vdash_c\nabla\alpha$ holds. Since $\alpha \in T^\partial \base(\alpha)$ (to be precise, $\alpha \in T^\partial_\omega \base(\alpha)$, but the base is a finite poset and therefore $T^\partial$ and $T^\partial_\omega$ coincide on it), we could instantiate the implication to get that $d(y_0)(\mathrel{\ol{T^\partial}{\Vdash_d} \cdot (\nu^{T^\partial}_{\Lang})_{\diamond}}) \alpha$ holds, and this would yield $y_0\Vdash_d\nabla\alpha$ as we wanted.

However, the implication~\eqref{eq:adequacy} can be expressed by the following diagram. The lax triangle is simply the image under $\ol{T^\partial}$ of~\eqref{eq:adequacy-ih}, and we formally pre-compose it with graph of the natural inclusion $(\nu^{T^\partial}_{\Lang})_\diamond$ (which, by $\wt{Z}$ being finite, is the identity):
\[
\xymatrixcolsep{4pc}
\xymatrix{
T_\omega^\partial\wt{Z}
\ar[0,1]|-{\object @{/}}^-{(\nu^{T^\partial}_{\Lang})_\diamond}
&
T^\partial\wt{Z}
\ar[0,2]|-{\object @{/}}^-{\ol{T^\partial}{\Vdash_d}(-,T^\partial\base(\alpha)-)}
\ar[1,1]|-{\object @{/}}_-{\ol{T^\partial}{\Vdash_c}(-,T^\partial\base(\alpha)-)\phantom{MM}}
&
&
(T Y)^\op
\\
&
&
(T X)^\op
\ar@{}[-1,0]|{\uparrow}
\ar[-1,1]|-{\object @{/}}_-{(\ol{T}S)^\con}
}
\]

\medskip\noindent
(2) For this part, we assume $T=T_\omega$ (i.e., the coalgebra functor is finitary). As both the coalgebra functor $T_\omega$ and the arity functor $T^\partial_\omega$ are finitary, we can simplify the semantic clause for $x\Vdash \nabla_{T^\partial_\omega}\alpha$ to
\[
c(x) \mathrel{\ol{T^\partial_\omega}{\Vdash}} \alpha.
\]
Expressivity boils down to proving that ``being modally stronger''
is a $T_\omega$-simulation. For the purpose of the proof, write
\[
S(x,y)
\]
to denote that $(d,\Vdash_d,y)$ is modally stronger than
$(c,\Vdash_c,x)$.

Hence the following implication holds for any formula $a$:
\begin{equation}
\label{eq:expressivity1}
\mbox{If }
S(x,y)
\mbox{ and }
x\Vdash_c a
\mbox{, then }
y\Vdash_d a.
\end{equation}
It is clear that $S$ verifies the properties~(a) and~(c)
of $T_\omega$-simulations from Definition~\ref{def:simul}. It remains
to be proved that $S$ satisfies the property~(b):
\begin{equation}
\label{eq:expressivity2}
\mbox{If }
S(x,y)
\mbox{, then }
\ol{T_\omega}S(c(x),d(y)).
\end{equation}
Assume therefore that $S(x,y)$ holds
and denote the corresponding bases of $c(x)$ and $d(y)$ by
\[
\xymatrixcolsep{1pc}
\xymatrix{
\base(c(x)):U
\mono[0,1]
&
X
&
\mbox{and}
&
\base(d(y)):W
\mono[0,1]
&
Y.
}
\]
We first construct a monotone map $f:X^\op\to\Lang$ such that
for any  $x' \in X$ and $w \in \base(d(y)) $ the following two requirements hold:
\begin{eqnarray}
&&
x'\Vdash_c f(x')\label{eq:expressivity3}
\\
&&
w \Vdash_d f(x')
\mbox{ implies }
S^\con(w,x')\label{eq:expressivity4}
\end{eqnarray}
Fix $x'$ in $X$, and for each $w \in \base(d(y))$ such that
\[
\neg S(x',w)
\]
pick a formula $b_{(x',w)}$ for which
$x'\Vdash_c b_{(x',w)}$ but
$w\not\Vdash_d b_{(x',w)}$, which is possible
by $\neg S(x',w)$ and~\eqref{eq:expressivity1}.
Define
\[
f(x')=\bigwedge\limits_{w \in \base(d(y)): \neg S(x',w)}b_{(x',w)}.
\]
The above conjunction is finitary since $W$ is a finite
poset. In case the conjunction is empty, we set $f(x') = \top$. Moreover, $f$ is a monotone map (with $x'' \geq x'$ the number of conjuncts in the definition of $f(x'')$ decreases), and the properties~\eqref{eq:expressivity3} and~\eqref{eq:expressivity4} hold. In particular, given any $u \in \base(c(x))$, the implication
\[
w \Vdash_d f(u)
\mbox{ implies }
S^\con(w,u)
\]
holds.

The latter implication can be expressed as a 2-cell
\[
\xymatrix{
U^\op
\ar@{} [0,4]|-{\uparrow}
&
&
&
&
W^\op
\ar@{<-} `u[llll] `[llll]|-{\object @{/}}_-{S^\con(\base(d(y)^\op)-,\base(c(x)^\op)-)} [llll]
\ar@{<-} `d[llll] `[llll]|-{\object @{/}}^-{\Vdash_d(\base(d(y))-,f\base(c(x))-)} [llll]
}
\]
in $\Rel$.
We apply the functor $\ol{T^\partial_\omega}$ to it and use the properties of relation lifting to obtain a 2-cell
\begin{equation}
\label{eq:expressivity5}
\vcenter{
\xymatrix{
(T_\omega U)^\op
\ar@{} [0,4]|-{\uparrow}
&
&
&
&
(T_\omega W)^\op
\ar@{<-} `u[llll] `[llll]|-{\object @{/}}_-{(\ol{T_\omega}S)^\con(T_\omega\base(d(y))^\op-,T_\omega\base(c(x))^\op-)} [llll]
\ar@{<-} `d[llll] `[llll]|-{\object @{/}}^-{\ol{T^\partial_\omega}{\Vdash_d}(T_\omega\base(d(y))-,T_\omega^\partial f\cdot T_\omega\base(c(x))-)} [llll]
}
}
\end{equation}
By the unit property of $\base$, we know that
\[
c(x) \in T_\omega \base(c(x))
\mbox{ and }
d(y) \in T_\omega \base (d(y))
\]
holds. Therefore we can use~\eqref{eq:expressivity5} and the definition of the converse of a relation to observe that
$
\ol{T_\omega}S(c(x),d(y))
$
holds whenever
$
d(y)\mathrel{\ol{T_\omega^\partial}{\Vdash_d}}(T_\omega^\partial f)c(x)
$
holds.

The relation $
d(y)\mathrel{\ol{T_\omega^\partial}{\Vdash_d}}(T_\omega^\partial f)c(x)
$ holds iff
$
y\Vdash_d\nabla((T_\omega^\partial f)c(x))
$ holds. Since we assume that $S(x,y)$ holds, it is enough to prove that $x\Vdash_c\nabla((T_\omega^\partial f)c(x))$ holds, and $
y\Vdash_d\nabla((T_\omega^\partial f)c(x))
$ will follow.

By~\eqref{eq:expressivity3} we know that $x'\Vdash_c f(x')$ holds
for every $x'$. Thus we get the inequality  $\id_X\leq{\Vdash_c}(-,f-)$,
and by applying the functor $\ol{T_\omega^\partial}$ to it, we obtain the inequality
$\id_{T_\omega^\partial X}\leq\ol{T_\omega^\partial}{\Vdash_c}(-,T_\omega^\partial f-)$.
Hence $c(x)\mathrel{\ol{T_\omega^\partial}{\Vdash_c}}(T_\omega^\partial f)c(x)$ holds,
and therefore $x\Vdash_c\nabla((T_\omega^\partial f)c(x))$ holds as required.
\end{proof}

\subsection{A dual modality $\Delta$}
\label{ssec:delta}

In the preceding paragraphs we proved that the coalgebraic
language $\Lang$ with a single nabla modality is adequate, and in the finitary case also expressive. Nevertheless, we introduce another
modality $\Delta$ in this subsection. The semantics of this
modality will show that $\Delta$ is, in a sense, dual
to $\nabla$. Why would we want
to extend the language with a new modality, knowing that
the language with nabla as the only modality is already
expressive? It will turn out that the dual modality
is crucial in designing a \emph{cut-free} two-sided sequent
calculus which is sound and complete for the logic of $T$-coalgebras.

The definition of semantics of $\Delta$ below is
a straightforward adaptation of a similar modality
studied in~\cite{kissvenema} in the context of complementation of coalgebraic automata in
the case of $\Set$ and classical Moss'
logic, where $\Delta$ is
the \emph{boolean} dual of $\nabla$.
For the classical Moss' logic in $\Set$, the dual modality $\Delta$, and its mutual definability with
the modality $\nabla$, played a crucial role
in the formulation of a cut-free \emph{two-sided} sequent calculus as shown in~\cite{bpv13}.
Without the boolean negation, a one-sided sequent calculus is
not available in our case already for the propositional
part of the language $\Lang$. For that reason we aim at a two-sided calculus,
and therefore it is convenient to have a modality
dual to nabla in the language.

We extend the language $\Lang$ with a monotone modality
\[
\Delta_{T^\partial_\omega}: T^\partial_\omega\Lang\to\Lang.
\]
Given a coalgebra $c: X\to TX$ and a monotone valuation
$
\xymatrix@1{
\Vdash:
\Lang
\ar[0,1]|-{\object @{/}}
&
X^\op
}
$,
the semantics of $\Delta_{T^\partial}$
is defined \emph{negatively} by the relation lifting of the \emph{negated} relation
$
\xymatrix@1{
\nVdash:
\Lang^\op
\ar[0,1]|-{\object @{/}}
&
X
}
$
as follows:
\[
x\nVdash\Delta\alpha\
\mbox{ iff }\
c(x) (\mathrel{\ol{T}{\nVdash}}\cdot (\nu^{T^\partial}_{\Lang})_{\diamond} )\alpha.
\]

Before we discuss the mutual relation of the two modalities $\nabla$ and $\Delta$,
we illustrate the semantics of $\Delta$ on the running examples:

\begin{exa}
\mbox{}\hfill
\begin{enumerate}
\item
{\em Delta for the functor $T = A\times\Id$}.
Let us fix a coalgebra
$c = \langle out, next \rangle: X\to A \times X$. The semantics of the modality $\Delta$ is then given as follows:
\[
x \Vdash \Delta(a,b) \quad \mbox{iff} \quad
out(x) \nleq_A a \;\mbox{or}\; next(x) \Vdash b
.\]
\item {\em Delta for the functor $T = (\Id)^A \times \bTwo$}.
Fix a coalgebra $c: X \to X^A \times \bTwo$. If for some state $x \in X$ we denote by $c(x) = (f,i)$ the successors and output of $x$, the semantics of
the modality $\Delta$ in the state $x$ is given as follows:
\[x \Vdash \Delta(\Phi,j) \quad \mbox{iff} \quad j \nleq i \;\mbox{or}\; (\exists a) f(a) \Vdash \Phi(a)\]
\item
{\em Delta for the functor $T = (\LLL)^A \times \bTwo$}.
For a coalgebra $c: X \to (\LLL X)^A \times \bTwo$ and for its state $x$ with $c(x) = (l,i)$, we have that $x \Vdash \Delta (p,j)$ holds if and
only if
\[
j \nleq i
\mbox{ or }
(\exists a \in A) (\exists y \sqin l(a)) (\forall \phi \sqin p(a))
\ y \Vdash \phi.
\]
\item
Consider again frames for modal
logic with adjoint modalities (see Example~\ref{ex:coalgebras}).
We only consider coalgebras of the form
$c: X\to\UU X$ for this example.
The semantics of $\Delta_{\LLL_\omega}$ works as follows:
\begin{eqnarray*}
x\nVdash\Delta_{\LLL_\omega}\beta
& \equiv &
(c(x) (\mathrel{\ol{\UU}{\nVdash}} \cdot (\nu^{\LLL}_{\Lang})_{\diamond}) \beta)
\\
& \equiv &
((\forall b) (\beta \sqni b \mbox{ implies } (\exists x') (c(x) \sqni x' \mbox{ and } x' \nVdash b))
\\
& \equiv &
\neg (\exists b) (\beta \sqni b \mbox{ and } (\forall x') ( c(x) \sqni x' \mbox{ implies } x' \Vdash b ) )
\end{eqnarray*}
Therefore
\[
\Delta_{\LLL_\omega}\beta
\equiv
\bigvee\Box\beta.
\]
Recall from Example~\ref{ex:nabla} that
\[
\Box b\equiv\nabla_{\LLL_\omega}\{b\}.
\]
We see that in this particular example
$\Delta$ is definable by $\nabla$ as follows:
\[
\Delta_{\LLL_\omega}\beta
\equiv
\bigvee\limits_{b\sqin \beta}\nabla_{\LLL_\omega}\{b\}.
\]
Since the arity functor $\LLL_\omega$ is finitary,
the above expression is a well-formed formula.
\item Consider $\PP^c$ coalgebras for positive modal logic (see Example~\ref{ex:coalgebras}). The semantics of $\Delta_{\PP^c_\omega}$ works as follows
\begin{eqnarray*}
x\nVdash\Delta_{\PP^c_\omega}\beta
& \equiv &
(c(x) (\mathrel{\ol{\PP^c}{\nVdash}} \cdot (\nu^{\PP^c}_{\Lang})_{\diamond}) \beta)
\\
& \equiv &
(\forall x'\in c(x)) (\exists a\in\alpha) x'\nVdash a \ \mbox{ and } \ 
(\forall a\in\alpha) (\exists x'\in c(x)) x'\nVdash a
\\
& \equiv &
\neg ((\exists x'\in c(x)) (\forall a\in\alpha) x'\Vdash a \ \mbox{ or } \ 
(\exists a\in\alpha) (\forall x'\in c(x)) x'\Vdash a)
\end{eqnarray*}
Therefore 
\[
\Delta_{\PP^c_{\omega}}\beta \equiv \Diamond \bigwedge {\uparrow}\beta \vee \bigvee \Box {\downarrow}\beta.
\]
Recall from Example~\ref{ex:nabla} that
\begin{align*}
\Box a
&\mbox{ can be expressed as }
\nabla_{\PP^c_\omega}\{a\}\vee \nabla_{\PP^c_\omega}\emptyset, \\
\Diamond a
& \mbox{ can be expressed as }
\nabla_{\PP^c_\omega}\{a,\top\},
\end{align*}
Thus $\Delta_{\PP^c_\omega}\beta$ can be expressed as a disjunction of $\nabla_{\PP^c_\omega}$-formulas as follows:
\[
\Delta_{\PP^c_\omega}\beta \equiv \nabla_{\PP^c_\omega}\{\bigwedge {\uparrow}\beta,\top\}\vee \bigvee \{ \nabla_{\PP^c_\omega}\{b\},\nabla_{\PP^c_\omega}\emptyset \ 
|\ b\in{\downarrow}\beta
\}.
\]
\end{enumerate}
\end{exa}
 The above example suggests that the modality $\Delta$ might always be expressible by
a certain disjunction of $\nabla$ formulas. This is indeed true, but the disjunction need not in general be finitely generated: for some coalgebra functors the disjunction of the $\nabla$ formulas needed in the expression is inherently infinite and thus $\Delta$ is not definable by
a well-formed formula of the finitary language. The mutual definability of the two modalities becomes important when formulating rules of the sequent calculus. 

To explain the relationship between the modalities $\nabla$ and $\Delta$, we need to unravel the combinatorial principle underlying their mutual translations. 
It will not be as simple as the previous example suggests, but let us have a glimpse at it first. Take an upperset ${\uparrow}{b}$ for some $b\sqin \beta$, and generate a lowerset from it in $\LLL_\omega\UU_\omega\Lang$. Call it $\Phi_b$. It has the following property: $\Phi_b \neg\ol{\UU_\omega}{\not\sqni} \beta$. Moreover, applying the $\LLL_\omega\bigwedge$ map to $\Phi_b$ yields an element in $\LLL_\omega\Lang$, namely, the lowerset of conjunctions of uppersets containing $b$. It is generated by $\bigwedge{\uparrow}{b}$. We may apply $\nabla_{\LLL_\omega}$ to it. The above example suggests, that one of the formulas $\nabla_{\LLL_\omega}\bigwedge{\uparrow}{b}$ will be true whenever $\Delta_{\LLL_\omega}\beta$ is true.

The underlying operations suggested by the above are in particular the lifted relations $\neg\ol{T}{\not\sqin}$ and $\neg\ol{T}{\not\sqni}$.
We describe them first, as they will play a crucial role in the mutual definitions. We provide an inductive definition of these relations for the case of $T$ being a Kripke-polynomial functor.

\begin{exa}
We state explicit formulas for the relation
$
\xymatrix@1{
\neg\ol{T}{\not\sqin}:
T^\partial\LLL_\omega X
\ar[0,1]|-{\object @{/}}
&
T^\partial X
}
$
when $T$ is a Kripke-polynomial functor.
\begin{enumerate}
\item
$T=\const_E$: The relation
$
\xymatrix@1{
\neg\ol{E}{\not\sqin}:
E^\op
\ar[0,1]|-{\object @{/}}
&
E^\op
}
$
is the relation
$
\xymatrix@1{
\not\leq_E:
E^\op
\ar[0,1]|-{\object @{/}}
&
E^\op
}
$.
\item
$T=\Id$: The relation
$
\xymatrix@1{
\neg \ol{\Id}{\not\sqin}:
\LLL_\omega X
\ar[0,1]|-{\object @{/}}
&
X
}
$
is the relation
$
\xymatrix@1{
\sqin:
\LLL_\omega X
\ar[0,1]|-{\object @{/}}
&
X
}
$.
\item
$T=T_1+T_2$: Let $\alpha$ be in $T_i^\partial X$ and $\Phi$ be in $T_j^\partial \LLL_\omega X$ for some $i,j$ in $\{1,2\}$. The relation $\alpha \mathrel{\neg(\ol{T_1 + T_2}{\not\sqin})} \Phi$ holds iff $i \neq j$ or
$\alpha \mathrel{\neg\ol{T_i}{\not\sqin}} \Phi$ with $i = j$.
\item
$T=T_1\times T_2$: The relation $(\alpha_1,\alpha_2) \mathrel{\neg(\ol{T_1 \times T_2}{\not\sqin})} (\Phi_1,\Phi_2)$ holds  iff
$\alpha_1 \mathrel{\neg\ol{T_1}{\not\sqin}} \Phi_1$ or
$\alpha_2 \mathrel{\neg\ol{T_2}{\not\sqin}} \Phi_2$ holds.
\item
$T=T_1^E$: Generalising the case (4), the relation $\alpha \mathrel{\neg\ol{T_1^E}{\not\sqin}} \Phi$ holds iff there is some $e$ from $E$ such that $\alpha(e) \mathrel{\neg\ol{T_1}{\not\sqin}} \Phi(e)$ holds.
\item
$T=\LLL_\omega T_1$: The relation
$v \mathrel{\neg\ol{\LLL_\omega T_1}{\not\sqin}} u$
holds iff the following condition is satisfied:
\[
(\exists \alpha \sqin v)(\forall \Phi \sqin u) \; \alpha \mathrel{\neg\ol{T_1}{\not\sqin}} \Phi.
\]
\end{enumerate}
The formulas for the relation
$
\xymatrix@1{
\neg\ol{T}{\not\sqni}:
T^\partial\LLL_\omega X
\ar[0,1]|-{\object @{/}}
&
T^\partial X
}
$
are easy to obtain by dual reasoning. For example, the relation $u \mathrel{\neg\ol{\LLL_\omega T_1}{\not\sqni}} v$ holds if and only if the following condition is satisfied:
\[
(\exists \Phi \sqin u)(\forall \alpha \sqin v) \; \Phi \mathrel{\neg\ol{T_1}{\not\sqni}} \alpha.
\]
\end{exa}

Next we define collections $R_\alpha$ and $L_\beta$, parametric in $\alpha \in T^\partial_\omega \Lang$ and $\beta \in T^\partial_\omega \Lang$ respectively.  
$R_\alpha$ contains all $\Psi \in T^\partial_\omega \LLL_\omega \base(\alpha)$ of which $\alpha$ is not a lifted non-member. Similarly, $L_\beta$ contains all $\Phi \in T^\partial_\omega \UU_\omega \base(\beta)$ of which $\beta$ is not a lifted non-member.
The collection $R_\alpha$ will be explicitly used to rewrite a formula $\nabla_{T^\partial_\omega}\alpha$ as a conjunction of delta-formulas by the right rule $\nabla$-r in the calculus in the next section (hence the notation). Similarly, $L_\beta$ will be used in rewriting a formula $\Delta_{T^\partial_\omega}\beta$ as a disjunction of nabla-formulas by the left rule $\Delta$-l in the calculus.

\begin{defi}
\label{def:collections_L_R} Let us fix elements $\alpha \in T^\partial_\omega \Lang$ and $\beta \in T^\partial_\omega \Lang$.
We define the collections $R_\alpha$ and $L_\beta$ as follows:
\begin{enumerate}
\item
$R_\alpha$ is the collection of all elements $\Psi \in T^\partial_\omega \LLL_\omega \base(\alpha)$ with
\[
\alpha \mathrel{\neg\ol{T_\omega}{\not\sqin}} \Psi.
\]
\item
$L_\beta$ is the collection of all elements $\Phi \in T^\partial_\omega \UU_\omega \base(\beta)$ with
\[
\Phi \mathrel{\neg\ol{T_\omega}{\not\sqni}} \beta.
\]
\end{enumerate}
\end{defi}

\begin{exa}\label{ex:RL}
To illustrate the above definition we give two simple examples:
\begin{enumerate}
\item Consider $T^\partial_\omega = \LLL_\omega$, thus $T_\omega = \UU_\omega$. Note that
$\alpha\neg\ol{\UU_\omega}\not\sqin\Psi$ holds iff
\[
\exists \psi (\Psi\sqni \psi \mbox{ and } \forall a (\alpha\sqni a \mbox{ implies } a\sqin \psi)).
\]
For $\alpha = \{\bot\}$, the collection $R_{\{\bot\}}$ consists of those elements $\Psi$ of $\LLL_\omega\LLL_\omega\{\bot\}$, where some $\psi\sqin\Psi$ contains $\bot$. The only such lowerset $\Psi$ is $\{\emptyset,\{\bot\}\}$. Thus $R_{\{\bot\}}$ equals $\{\{\emptyset,\{\bot\}\}\}$.

\item Consider $T^\partial = \mathbb{N}^\op\times\Id$ where $\mathbb{N}$ is the poset of natural numbers with their natural order, and $T = \mathbb{N}\times\Id$. For $\alpha = (5,b\wedge c)$, the collection $R_{(5,b\wedge c)}$ consists of elements $(n,\psi)$ of $N^\op\times\LLL_\omega\base(5,b\wedge c)$ where $5\nleq n$ or $b\wedge c\sqin\psi$. Note that $\LLL_\omega\base(5,b\wedge c) = \{\emptyset,\{b\wedge c\}\}$. Therefore $R_{(5,b\wedge c)} = \{(n,\emptyset)\ |\ 5\nleq n\}\cup \{(n,\{b\wedge c\})\ |\ n\in N\}$.

\end{enumerate}
\end{exa}
To compare the current setting with the classical Moss' logic: It was shown in~\cite{kissvenema} that the two modalities $\nabla$ and $\Delta$ are mutually definable already in the positive fragment of the classical Moss' logic. Their mutual definability was used in~\cite{bpv13} in designing a two-sided sequent proof system for the logic.
We will state a similar fact in the poset-setting in the following two propositions. They will be applied in the next section, when we define a sequent proof system.

\begin{prop}
\label{prop:deltadl}
For each coalgebra $c: X\to TX$ and every semantics relation $\Vdash_c$,
the relations
\[
\xymatrixcolsep{4pc}
\xymatrix{
T_\omega^\partial\Lang
\ar[0,1]|-{\object @{/}}^-{(\nu^{T^\partial}_{\Lang})_\diamond}
&
T^\partial\Lang
\ar[0,1]|-{\object @{/}}^-{{\Vdash_c}(-,\Delta-)}
&
X^\op
}
\]
and
\[
\xymatrixcolsep{5pc}
\xymatrix{
T^\partial_\omega\Lang
\ar[0,1]|-{\object @{/}}^-{\neg(\ol{T_\omega}{\not\sqni})}
&
T^\partial_\omega\UU_\omega\Lang
\ar[0,1]|-{\object @{/}}^-{{\Vdash_c}(-,\nabla T^\partial_\omega\bigwedge-)}
&
X^\op
}
\]
are equal.

In other words, for each $x$ in $X$, and each $\beta$
in $T^\partial_\omega\Lang$, we have that $x\Vdash_c\Delta\beta$ holds if and only if
there is a $\Phi_x$ in $T^\partial_\omega\UU_\omega\Lang$ such that
\[
\Phi_x \mathrel{\neg{\ol{T_\omega}{\not\sqni}}} \beta
\mbox{ and }\
x\Vdash_c\nabla(T^\partial_\omega\bigwedge)\Phi_x.
\]
Whenever $L_\beta$ is finite, we have the following semantic equivalence:
\[
\Delta\beta
\equiv \bigvee\limits_{\Phi \in L_\beta}
\nabla(T^\partial_\omega\bigwedge)\Phi
.
\]
\end{prop}
\begin{proof}
The guiding ideas in the proof are quite simple, however, much of the work is due to type-checking reasons. 
We divide the proof into several parts:
\begin{enumerate}
\item
The inequality
$
{\Vdash_c}(-,\Delta-)
\cdot
(\nu^{T^\partial}_{\Lang})_\diamond
\leq
{\Vdash_c}(-,\nabla T^\partial_\omega\bigwedge-)
\cdot
\neg(\ol{T_\omega}{\not\sqni})
$.

Suppose that $x\Vdash\Delta\beta$ holds. We will construct an element $\Phi$ in $(T^\partial_\omega\UU_\omega \Lang)^\op$
such that
\[
\Phi \mathrel{\neg{\ol{T_\omega}{\not\sqni}}} \beta
\
\mbox{ and }
\
x\Vdash_c \nabla(T^\partial_\omega\bigwedge)\Phi.
\]
The main idea for this part of the proof is to simply consider the monotone map assigning to a state $x$ the finitely generated upperset of formulas from $\base(\beta)$ it satisfies:
\[
f^\sharp: X\to(\UU_\omega \Lang)^\op,
\quad
x\mapsto {\uparrow} \{w \in \base(\beta) \ |\ x\Vdash_c w \}.
\]
Most of what follows concerns restricting along $\base(\beta)$.
Using the
$
\monoInline{\base(\beta)}
           {W}
           {\Lang}
$ with $W$ finite, we can factorise $f^\sharp$ as follows:
\[
\xymatrixcolsep{1pc}
\xymatrix{
X
\ar[0,2]^{f}
\ar[1,1]_{f^\sharp}
&
&
(\UU_\omega W)^\op
\ar[1,-1]^{(\UU_\omega \base(\beta))^\op}
\\
&
(\UU_\omega \Lang)^\op
&
}
\]
The mapping $f$ with a finite range is then defined as
\[
f: X\to(\UU_\omega W)^\op,
\quad
x\mapsto \{w \in W \ |\ x\Vdash_c \base(\beta)(w) \}.
\]
Therefore, applying the functor $T$ to $f$ yields an element $(Tf)c(x)$ in $(T^\partial\UU_\omega W)^\op$. 
Observe first that because $W$ (and therefore $\UU_\omega W$) is a finite poset, we have that $T^\partial(\UU_\omega W)^\op = T^\partial_\omega(\UU_\omega W)^\op$ (i.e., the corresponding natural inclusion $\nu^{T^\partial}_{\UU_\omega W}$ is the identity). We may therefore define the element $\Phi_x$ in $(T^\partial_\omega\UU_\omega W)^\op$ to be the unique one with
\[
\nu^{T^\partial}_{\UU_\omega W}(\Phi_x) =(Tf)c(x).
\]
As $(\UU_\omega W)^\op$ is included in $(\UU_\omega \Lang)^\op$ via a mono morphism $(\UU_\omega \base(\beta))^\op$, we may put 
\[
\Phi = (T^\partial_\omega\UU_\omega \base(\beta))^\op(\Phi_x),
\]
and see that it is indeed an element of $(T^\partial_\omega\UU_\omega \Lang)^\op$. Relaxing the typing, we could say that $\Phi,\Phi_x$, and $(Tf)c(x)$ are the same thing. The situation is depicted in the following naturality square:
\[
\xymatrixcolsep{1pc}
\xymatrix{
\Phi_x :T^\partial_\omega\UU_\omega W
\ar@{=}[0,2]^{\nu_{\UU_\omega W}^{T^\partial}}
\ar@{>->}[1,0]_{T^\partial_\omega\UU_\omega\base(\beta)}
&
&
T^\partial\UU_\omega W
\ar@{>->}[1,0]^{T^\partial\UU_\omega\base(\beta)}
&
: (Tf)c(x)
\\
\Phi :T^\partial_\omega\UU_\omega \Lang
\ar@{>->}[0,2]^{\nu_{\UU_\omega \Lang}^{T^\partial}}
&
&
T^\partial\UU_\omega \Lang
&
}
\]
Let 
$
\xymatrix@1{
\not\sqni':
W^\op
\ar[0,1]|-{\object @{/}}
&
(\UU_\omega W)^\op
}
$
be the restriction of the relation 
$
\xymatrix@1{
\not\sqni:
\Lang^\op
\ar[0,1]|-{\object @{/}}
&
(\UU_\omega \Lang)^\op
}
$
along $\base(\beta)$, i.e.\ $\not\sqni'\ =\ \not\sqni(\UU_\omega\base (\beta) - ,\base(\beta) - )$. We obtain that $\Phi$ is related with $\beta$ via the lifted $\not\sqni$ relation, i.e.\ $\Phi \mathrel{{\ol{T_\omega}{\not\sqni}}} \beta$,  if and only if $\Phi_x$ and $\beta$ relate in the lifted $\not\sqni'$ relation restricted along $\base(\beta)$, i.e.\ $\Phi_x\mathrel{{\ol{T_\omega}{\not\sqni'}}} \beta$
This will concern the point (a) below. 

Similarly, let 
$
\xymatrix@1{
\Vdash_c':
W
\ar[0,1]|-{\object @{/}}
&
X^\op
}
$
be the restriction of the relation 
$
\xymatrix@1{
\Vdash_c:
\Lang
\ar[0,1]|-{\object @{/}}
&
X^\op
}
$
along $\base(\beta)$, i.e.\ $\Vdash_c'\ =\ \Vdash_c(-, \base(\beta)-)$.
We obtain that $c(x)$ and $\nu_\Lang^{T^\partial}(\beta)$ are related via the lifted $\nVdash_c$ relation, i.e.\ $c(x) ( \mathrel{{\ol{T}{\nVdash_c}}} )  \nu_\Lang^{T^\partial}(\beta)$, if and only if  $c(x)$ and $\beta$ are related via the lifted restricted $\nVdash_c'$ relation as follows:  $c(x)(\mathrel{{\ol{T}{\nVdash_c'}}}) \beta$.
This will concern the point (a) below.

For the conjunction map (restricted to $\base(\beta)$)
$
\bigwedge: \UU_\omega W \to \Lang,
$
the restricted relation $\Vdash_c(-,\bigwedge -)$ of ``satisfying conjunctions of formulas from $\base(\beta)$'' will be considered in the point (b) below. 
We obtain that $c(x)$ and $\nu_\Lang^{T^\partial}((T^\partial_\omega\bigwedge)\Phi)$ are related via the lifted $\Vdash_c$ relation, i.e.\ $c(x) ( \mathrel{{\ol{T^\partial}{\Vdash_c}}} )  \nu_\Lang^{T^\partial}((T^\partial_\omega\bigwedge)\Phi)$, if and only if  $c(x)$ and $\Phi_x$ are related via the $T^\partial$ lifted restricted $\Vdash_c(-,\bigwedge -)$ relation. These observations simplify a bit the remaining proof.

\smallskip
\begin{enumerate}
\item
To prove that $\Phi_x \mathrel{\neg{\ol{T_\omega}{\not\sqni'}}} \beta$ holds,
it suffices to prove that
\[
c(x) \mathrel{\neg\ol{T}{\nVdash'}} \beta\
\ \mbox{implies}\ \
\Phi_x\mathrel{\neg\ol{T_\omega}{\not\sqni'}}\beta.
\]
We reason by contraposition and we assume that
$\Phi_x\mathrel{\ol{T_\omega}{\not\sqni'}}\beta$
to show that then $c(x)\mathrel{\ol{T}{\nVdash'}} \beta$,
contradicting the original assumption that $x\Vdash\Delta\beta$.

By the definition of $f^\sharp$, observe that for any $w \in \base(\beta)$ we have an implication
\[
x \Vdash_c w \mbox{ implies } fx \sqni' w.
\]
By contraposition,
$
fx \not \sqni' w
$
implies
$
x \not \Vdash_c w
$. Write this implication as the following diagram:
\[
\xymatrix{
W^\op
\ar@{} [0,4]|-{\uparrow}
&
&
&
&
X.
\ar@{<-} `u[llll] `[llll]|-{\object @{/}}_-{{\not\Vdash'_c}(-,\base(\beta)^\op-)} [llll]
\ar@{<-} `d[llll] `[llll]|-{\object @{/}}^-{{\not\sqni'}(f-,-)} [llll]
}
\]
The image under $\ol{T}$ of the above diagram yields
\[
\xymatrix{
(T^\partial W)^\op
\ar@{} [0,4]|-{\uparrow}
&
&
&
&
TX
\ar@{<-} `u[llll] `[llll]|-{\object @{/}}_-{\ol{T}{\not\Vdash'_c}(-,(T^\partial\base(\beta))^\op-)} [llll]
\ar@{<-} `d[llll] `[llll]|-{\object @{/}}^-{\ol{T}{\not\sqni'}(Tf-,-)} [llll]
}
\]
We read the diagram for $\beta$ in $(T^\partial W)^\op$ and $c(x)$ in $TX$ 
(this is correct because $W$ is finite, and $(T^\partial W)^\op = (T^\partial_\omega W)^\op$). The lower part says that $(Tf)(c(x))\mathrel{\ol{T}{\not\sqni'}} \beta$:
By assumption, $\Phi_x \mathrel{\ol{T_\omega}{\not\sqni'}} \beta$ holds, and since
$
\Phi_x = (Tf) c(x)
$,
we deduce that
$(Tf)(c(x))\mathrel{\ol{T}{\not\sqni'}} \beta$
holds. 

For the upper part, we first recall that, by the unit property of base, $\beta \in T_\omega^\partial\base(\beta)$. Together with $\base(\beta)$ being finite, this entails that $\beta \in T^\partial\base(\beta)$.
Hence we have
$c(x)\mathrel{\ol{T}{\nVdash'}} \beta$ by using the above diagram.
This is a contradiction with $x\Vdash_c\Delta\beta$.
\item
We prove that $x\Vdash_c\nabla(T^\partial_\omega\bigwedge)\Phi_x$ holds by proving
$c(x) \mathrel{\ol{T^\partial}{\Vdash(-,\bigwedge -)}}\Phi_x$.

By the definition of $f$ we have that $x \Vdash_c \bigwedge f (x)$ holds. The monotonicity of $\Vdash_c$ together with the definition of conjunction says that for any other $\phi \in \UU_\omega W$ with $f(x) \leq \phi$ we can deduce $x \Vdash_c \bigwedge \phi$. This observation yields a diagram in $\Rel$ of the form
\[
\xymatrix{
\UU_\omega W
\ar@{} [0,4]|-{\uparrow}
&
&
&
&
X^\op
\ar@{<-} `u[llll] `[llll]|-{\object @{/}}_-{{\Vdash_c}(-,\bigwedge-)} [llll]
\ar@{<-} `d[llll] `[llll]|-{\object @{/}}^-{\leq_{\UU_\omega W}( f^\op-,-)} [llll]
}
\]
By applying $\ol{T^\partial}$ to the above diagram
we obtain
\[
\xymatrix{
T^\partial\UU_\omega W
\ar@{} [0,4]|-{\uparrow}
&
&
&
&
(TX)^\op
\ar@{<-} `u[llll] `[llll]|-{\object @{/}}_-{\ol{T^\partial}{\Vdash_c}(-,T^\partial\bigwedge-)} [llll]
\ar@{<-} `d[llll] `[llll]|-{\object @{/}}^-{\leq_{T^\partial\UU_\omega W}(Tf)^\op-,-)} [llll]
}
\]
We instantiate the diagram for $c(x)$ in $(TX)^\op$, and $\Phi_x$ in $T^\partial \UU_\omega W$. 
Since
\[
(Tf)^\op c(x) = (Tf) c(x) = \Phi_x,
\]
we have in particular the inequality
$(Tf)^\op c(x) \leq \Phi_x$ in $T^\partial \UU_\omega W$. By the above diagram, the relation
$c(x) \mathrel{\ol{T^\partial}{\Vdash_c}} (T^\partial\bigwedge)\Phi_x$ follows.

We deduce from this that $c(x) \mathrel{\ol{T^\partial}{\Vdash_c}} \nu^{T^\partial}_\Lang((T^\partial_\omega\bigwedge)\Phi_x)$, 
using the following naturality square for the conjunction map:
\[
\xymatrixcolsep{1pc}
\xymatrix{
\Phi_x :T^\partial_\omega\UU_\omega W
\ar@{=}[0,2]^{\ \nu_{\UU_\omega W}^{T^\partial}}
\ar[1,0]_{T^\partial_\omega\bigwedge}
&
&
T^\partial\UU_\omega W
\ar[1,0]^{T^\partial\bigwedge}
\\
T^\partial_\omega\Lang
\ar@{>->}[0,2]^{\nu_{\Lang}^{T^\partial}}
&
&
T^\partial\Lang
}
\]
\end{enumerate}

\item
The inequality
$
{\Vdash_c}(-,\nabla T^\partial_\omega\bigwedge-)
\cdot
\neg(\ol{T}{\not\sqni})
\leq
{\Vdash_c}(-,\Delta-)
\cdot
(\nu^{T^\partial}_{\Lang})_\diamond
$.

Consider $x$, $\Phi$ in $T^\partial_\omega\UU_\omega\Lang$, and $\beta$ in $T^\partial_\omega\Lang$, such that
\[
\Phi \mathrel{\neg{\ol{T}{\not\sqni}}} \beta
\ \mbox{ and }\
x\Vdash_c\nabla(T^\partial\bigwedge)\Phi.
\]
We need to show that $x\Vdash_c\Delta\beta$ holds, so by the definition of the semantics of $\Delta$ we need
to prove that $c(x)\mathrel{\neg\ol{T}{\nVdash}}\nu_\Lang^{T^\partial}(\beta)$ holds.
Reasoning by contraposition, we show that
\[
c(x)\mathrel{\ol{T}{\nVdash}}\nu_\Lang^{T^\partial}(\beta)\
\ \mbox{entails}\ \
\Phi\mathrel{\ol{T}{\not\sqni}}\beta,
\]
contradicting the assumption.

First let us observe that whenever there is an $x_0 \in X$ such that $x_0 \not \Vdash_c a$ for some $a \in \Lang$, and $x_0 \Vdash_c \bigwedge u$  for an upperset $u \in \UU_\omega \Lang$, then $u \not \sqni w$ follows. This is precisely the information contained in the diagram in $\Rel$ of the form
\[
\xymatrix{
\UU_\omega \Lang
\ar[0,4]|-{\object @{/}}^-{{\not\sqni}^\con}
\ar[1,2]|-{\object @{/}}_-{{\Vdash_c}(-,\bigwedge -)\phantom{MMM}}
&
&
&
&
\Lang
\\
&
&
X^\op
\ar@{}[-1,0]|{\uparrow}
\ar[-1,2]|-{\object @{/}}_-{\phantom{MMM}{\not\Vdash_c}^\con
}
&
&
}
\]
We apply $\ol{T^\partial}$ to the above lax triangle
to obtain
\[
\xymatrix{
T^\partial\UU_\omega \Lang
\ar[0,4]|-{\object @{/}}^-{(\ol{T}{\not\sqni})^\con}
\ar[1,2]|-{\object @{/}}_-{\ol{T^\partial}{\Vdash_c}(-,T^\partial\bigwedge-)\phantom{MMMM}}
&
&
&
&
T^\partial \Lang
\\
&
&
(TX)^\op
\ar@{}[-1,0]|{\uparrow}
\ar[-1,2]|-{\object @{/}}_-{\phantom{MMM}(\ol{T}{\not\Vdash_c})^\con}
&
&
}
\]
We instantiate the diagram for $c(x)$ in $(TX)^\op$, $\Phi$ in  $T^\partial\UU_\omega \Lang$, and $\nu^{T^\partial}_\Lang (\beta)$ in $T^\partial \Lang$.
Using 
the assumptions $
c(x)\mathrel{\ol{T^\partial}{\Vdash}}(T^\partial\bigwedge)\Phi
$
and $c(x)\mathrel{\ol{T}{\nVdash'}}\beta$, we conclude that $\Phi\mathrel{\ol{T}{\not\sqni}}\beta$ holds.
\end{enumerate}
We have proved the desired equality of the two relations.
The last assertion of the proposition trivially follows.
\end{proof}

The previous proposition states that under certain conditions, the modality $\Delta$ is semantically equivalent to a disjunction of $\nabla$ formulas. The following proposition deals with the dual statement: the modality $\nabla$ is semantically equivalent to a conjunction of $\Delta$ formulas.

\begin{prop}
\label{prop:nablaviadelta}
For each coalgebra $c: X\to TX$ and every semantics relation $\Vdash_c$,
the relations
\[
\xymatrixcolsep{4pc}
\xymatrix{
T_\omega^\partial\Lang
\ar[0,1]|-{\object @{/}}^-{(\nu^{T^\partial}_\Lang)_\diamond }
&
T^\partial\Lang
\ar[0,1]|-{\object @{/}}^-{{\Vdash_c}(-,\nabla-)}
&
X^\op
}
\]
and
\[
\xymatrixcolsep{5pc}
\xymatrix{
T_\omega^\partial\Lang
\ar[0,1]|-{\object @{/}}^-{\neg\ol{T_\omega}{\not\sqin}}
&
T_\omega^\partial\LLL_\omega\Lang
\ar[0,1]|-{\object @{/}}^-{{\Vdash_c}(-,\Delta T_\omega^\partial\bigvee-)}
&
X^\op
}
\]
are equal.

In other words, for each $x$ in $X$, and each $\alpha$
in $T_\omega^\partial\Lang$, $x\Vdash_c\nabla\beta$ holds iff
there is a $\Psi$ in $T_\omega^\partial\LLL_\omega\Lang$ such that
\[
\alpha \mathrel{\neg{\ol{T_\omega}{\not\sqin}}} \Psi
\
\mbox{ and }\
x\Vdash_c\Delta(T_\omega^\partial\bigvee)\Psi.
\]
Whenever $R_\alpha$ is finite,
we have the following semantic equivalence:
\[
\nabla\alpha
\equiv
\bigwedge\limits_{\Psi \in R_\alpha}
\Delta(T_\omega^\partial\bigvee)\Psi
.
\]
\end{prop}
\begin{proof}
The proof is completely analogous to that of
Proposition~\ref{prop:deltadl}.
\end{proof}

\section{Proof system and completeness}
\label{sec:proof-system}
In this section we define a two-sided sequent calculus
${\mathbf{G}}_{\nabla\Delta}$ for
the coalgebraic language $\Lang$ defined in
Section~\ref{sec:logic}.  We prove that the calculus
is complete with respect to the semantics of $\Lang$
given therein. The definition of the calculus as well
as the completeness proof is general and parametric
in the coalgebra functor $T$.

The calculus we are going to define can be seen as
closely related to the two-sided sequent
calculus for the finitary Moss' coalgebraic logic of coalgebras
in $\Set$ as presented in~\cite{bpv13}. We point out a closer connection of the two calculi in Subsection~\ref{ssec:positivefragments} at the very end of the paper.

\subsection{The proof system ${\mathbf{G}}^T_{\nabla\Delta}$}
\label{ssec:proof-system}
The calculus ${\mathbf{G}}^T_{\nabla\Delta}$ will manipulate sequents. A \emph{sequent}
is a syntactic object of the form $\phi\Rightarrow\psi$
where $\phi$ is a finitely generated upperset of formulas,
i.e.\ $\phi$ is in $\UU_\omega\Lang$, and $\psi$ is
a finitely generated lowerset of formulas, i.e., $\psi$
is in $\LLL_\omega\Lang$. Sequents can be seen as objects in $(\UU_\omega\Lang)^\op \times \LLL_\omega \Lang$, and therefore they carry an order and form a poset which we denote $\mathcal{S}$ (this order on sequents is a natural one as becomes clear from their semantics). 

Given a coalgebra $c: X\to TX$ equipped with a monotone
valuation $\Vdash_c$ and a state $x_0$, we say that
sequent $\phi\Rightarrow\psi$ is \emph{refuted} in $x_0$
if $x_0\Vdash_c\bigwedge\phi$ while
$x_0\nVdash_c \bigvee\psi$. Otherwise, the sequent
is \emph{valid} in $x_0$. Observe that the validity is monotone w.r.t.\ the order on $\mathcal{S}$.
A sequent that is valid
in all states of all $T$-coalgebras under all valuations
is simply called a \emph{valid sequent}, otherwise it is
called a \emph{refutable sequent}.

A \emph{rule} is a scheme written in the form
\[
\Inference{\mathcal{A}}{S}{}
\]
where $\mathcal{A}$ is a poset of sequents (namely, $\mathcal{A}$ is in $\Sub_{\MM}\mathcal{S}$) forming the assumptions of the rule, and the sequent $S$ is the conclusion of the rule.

Notice that as defined, the number of assumptions of a rule need not be finite.
It is because, for some functors, we will need infinitary modal rules to form
a complete calculus. However, if the functor $T^\partial_\omega$ preserves
finite posets, all the rules will be finitary.

A rule is
said to be \emph{sound} if, whenever all the assumptions
are valid sequents, the conclusion is also a valid sequent.
A rule is called \emph{invertible} if whenever the conclusion
is valid, then all the assumptions are valid as well. A rule has
a \emph{subformula property} if all the assumptions consist
of subformulas of the conclusion only. A rule with the empty set
of assumptions is called an \emph{axiom}.

A \emph{proof} in the calculus ${\mathbf{G}}^T_{\nabla\Delta}$ is a well-founded tree
labelled by sequents in the following way:
\begin{enumerate}
\item Each leaf is labelled
by an instance of an axiom.
\item Each non-leaf node $n$ is labelled by the conclusion of an instance $I$ of some rule. The children nodes of $n$ are labelled by the assumptions of the instance $I$.
\end{enumerate}

We say that
a sequent is \emph{provable} if there is a proof whose root
is labelled by the sequent.

We start the definition of the calculus ${\mathbf{G}}^T_{\nabla\Delta}$
by listing the axioms, rules of weakening and standard
rules for conjunction and disjunction.
If $\phi$ is an upperset of formulas, by writing
$\phi, a$ we mean the upperset generated by $\phi\cup\{a\}$, and by $\phi, \phi'$ we mean the upperset generated by $\phi\cup\phi'$.
Similarly If $\psi$ is a lowerset of formulas, by writing
$\psi, a$ we mean the lowerset generated by $\psi\cup\{a\}$, and by $\psi, \psi'$ we mean the lowerset generated by $\psi\cup\psi'$.

When writing down a sequent, e.g.\ in examples we provide,
we often use a simplified notation and list only the generators of the corresponding lowersets
and uppersets to keep the description finite\footnote{One can do so systematically by mapping sequents to their bases as follows: observing $(\UU\Lang)^\op = \LLL_\omega\Lang^\op$, we can use
$(\base_{\Lang^\op},\base_\Lang ): \Sub_{\MM}(\LLL_\omega\Lang^\op) \times \Sub_{\MM}(\LLL_\omega \Lang) \to \Sub_{\MM}\Lang^\op\times \Sub_{\MM}\Lang$.}.

\begin{defi}
\label{def:calculus_prop}
The propositional part of the calculus ${\mathbf{G}}^T_{\nabla\Delta}$ is given in the following table:

\begin{equation}
\fbox{
$
\begin{array}{ll}
\Inference{}
          {{\uparrow}{a} \Rightarrow {\downarrow}{b}}
          {Ax}
        \text{$a\leq_\Lang b$}
&
\\
&
\\
\Inference{\phi\Rightarrow\psi}
          {\phi\Rightarrow\psi',\psi}
          {w-r}
&
\Inference{\phi\Rightarrow\psi}
          {\phi,\phi'\Rightarrow\psi}
          {w-l}
\\
&
\\
\Inference{\{\phi\Rightarrow a,\psi\ |\ a\in g(\phi')\}}
          {\phi\Rightarrow\bigwedge\phi',\psi}
          { \text{$\bigwedge$-r }}
&
\Inference{\phi,\phi'\Rightarrow\psi}
          {\phi,\bigwedge\phi'\Rightarrow\psi}
          { \text{$\bigwedge$-l} }
\\
&
\\
\Inference{\phi\Rightarrow\psi,\psi'}
          {\phi\Rightarrow\psi,\bigvee\psi'}
          { \text{$\bigvee$-r} }
&
\Inference{\{\phi, a\Rightarrow\psi\ |\ a\in g(\psi') \}}
          {\phi,\bigvee\psi'\Rightarrow\psi}
          { \text{$\bigvee$-l} }
\end{array}
$
}
\end{equation}
\end{defi}
The rules for conjunction and disjunction above are straightforward
adaptations of the standard conjunction and disjunction rules to our
specific setting. In the axioms, the side condition requires that the inequality $a\leq_\Lang b$
holds in the free algebra of formulas.  All the axioms and rules
are sound, and the rules for $\bigwedge$ and $\bigvee$ are
moreover invertible, which follows immediately from the definition of valid sequents.
\begin{rem}
One can see the above tabular of rules as an inductive definition of a monotone relation 
$
\xymatrix@1{
\Rightarrow:
\LLL_\omega\Lang
\ar[0,1]|-{\object @{/}}
&
\UU_\omega\Lang
}
$ as the smallest one such that it in a way subsumes $\leq_\Lang$: $a\leq_\Lang b$ implies ${\uparrow}{a} \Rightarrow {\downarrow}{b}$, and is closed under the conjunction and disjunction rules (the weakening rules ensure the monotonicity).  
\end{rem}
Next we turn to modal rules. The modal part of the calculus will consist of three rule schemes: (1) a rule for introduction of nabla modality to the right-hand side of a sequent, (2) a rule for introduction of delta modality to the left-hand side of a sequent, and (3) a combined nabla-delta rule.

\paragraph{(1) $\nabla$-r and (2) $\Delta$-l.}
Let us first cover the cases
of introducing nabla to the right-hand and delta to the
left-hand side of a sequent. Behind the soundness of following two rules lie the mutual definability of the two modal operators discussed in the previous section.
The two rules,
when read backwards, say how to reduce a nabla
formula on the right-hand side of a sequent to
a conjunction of delta formulas, and dually how
to reduce a delta formula on the right-hand side
of a sequent to a disjunction of nabla formulas.
The laws that guarantee the soundness and invertibility
of these two rules are those of Propositions~\ref{prop:deltadl}
and~\ref{prop:nablaviadelta}.

\begin{defi}[Modal rules for $\nabla$-r and $\Delta$-l]\label{def:modalrules}
$ $
Consider $\alpha \in T_\omega^\partial\Lang$
%
with
$\monoInline{\base(\alpha)}{V}{\Lang}$,
and $\beta \in T_\omega^\partial\Lang$
with
$\monoInline{\base(\beta)}{W}{\Lang}$.
The two modal rules are defined as
\begin{equation}
\fbox{
$
\Inference{\{\phi\Rightarrow\Delta(T_\omega^\partial\bigvee)\Psi,\psi\mid\Psi\mbox{ in }R_\alpha \}}
          {\phi\Rightarrow\nabla\alpha,\psi}
          {\text{$\nabla$-r}}
\quad
\Inference{\{\phi,\nabla(T_\omega^\partial\bigwedge)\Phi\Rightarrow\psi\mid\Phi\mbox{ in }L_\beta\}}
          {\phi,\Delta\beta\Rightarrow\psi}
          {\text{$\Delta$-l}}
$
}
\end{equation}
The sets of assumptions are indexed by collections $R_\alpha$ with each $\Psi$ in $T_\omega^\partial\LLL_\omega\base(\alpha)$, and $L_\beta$ with each $\Phi$ in $T_\omega^\partial\UU_\omega\base(\beta)$, from~Definition~\ref{def:collections_L_R}. The collections can in general be infinite, and therefore the above rules become infinitary. However, in case that $T_\omega^\partial$ preserves finite posets, they will become finitary.
\end{defi}

\begin{exa}\label{ex:therules}
We illustrate the above definition with two simple examples of instances of the rule $\nabla$-r. The instances are based on Example~\ref{ex:RL}.
\begin{enumerate}
\item
Consider $T^\partial = \LLL_\omega$, and $\alpha = \{\bot\}$. By Example~\ref{ex:RL} (1), $R_\alpha = \{\{\emptyset,\{\bot\}\}\}$. Then for the only element $\{\emptyset,\{\bot\}\}$ of $R_\alpha$ we obtain $\LLL_\omega\bigvee\{\emptyset,\{\bot\}\} = \{\bot\}$.
The following is therefore a correct instance of the rule $\nabla$-r:

\[
\Inference{\nabla\emptyset\Rightarrow\Delta\{\bot\}}
          {\nabla\emptyset\Rightarrow\nabla\{\bot\}}
          {$\nabla$-r}
\]

\item
Consider $T^\partial = \mathbb{N}^\op\times\Id$, and $\alpha = (5,b\wedge c)$. By Example~\ref{ex:RL} (2), $R_{(5,b\wedge c)} = \{(n,\emptyset)\ |\ 5>n\}\cup \{(n,\{b\wedge c\})\ |\ n\in \mathbb{N}\}$. Note that $\LLL_\omega\bigvee(n,\emptyset) = (n,\bot)$ while $\LLL_\omega\bigvee(n,\{b\wedge c\}) = (n,b\wedge c)$. The following is therefore a correct instance of the rule $\nabla$-r:

\[
\Inference{\{\nabla(3,b),\nabla(8,c)\Rightarrow\Delta(n,\bot)\mid 5>n\}\cup
           \{\nabla(3,b),\nabla(8,c)\Rightarrow\Delta(n,b\wedge c)\mid n\in \mathbb{N}\}
           }
          {\nabla(3,b),\nabla(8,c)\Rightarrow\nabla(5,b\wedge c)}
          {$\nabla$-r}
\]

\end{enumerate}
\end{exa}

The rule $\nabla$-r is sound and invertible by
Proposition~\ref{prop:nablaviadelta}, and the rule
$\Delta$-l is sound and invertible by Proposition~\ref{prop:deltadl}.
Let us make a syntactic observation on the two rules concerning
the subformula property.

\begin{rem}[A form of the subformula property]
\label{rem:subpropinv}
Given any $\Psi$ of the type
\[
\xymatrix{
\bOne
\mono[0,2]
\mono[1,2]_{\Psi}
&
&
T_\omega^\partial \LLL_\omega V
\mono[1,0]^{T_\omega^\partial\LLL_\omega \base(\alpha)}
\\
&
&
T_\omega^\partial \LLL_\omega \Lang
}
\]
we can state the following, weaker kind of the subformula property:
\[
z \in \base(\Psi) \mbox{ implies } z \in \LLL_\omega \base(\alpha).
\]
Since $\Psi \in T_\omega^\partial \LLL_\omega \base(\alpha)$ holds, we get the inclusion
\[
\base(\Psi) \subseteq \LLL_\omega \base(\alpha)
\]
by Definition~\ref{def:base} of base. The weak subformula property follows from this immediately.
We can use the weak subformula property in particular to study the rule $\nabla$-r: It shows that all the assumptions of the rule are built from subformulas of the conclusion, using one additional modality $\nabla$, and the operation $T^\partial\bigvee$.
Similarly, for any $\Phi \in T^\partial \UU_\omega \base(\beta)$ we get that
\[
z \in \base(\Phi) \mbox{ implies } z \in \UU_\omega \base(\beta).
\]
\end{rem}

\paragraph{(3) $\nabla\Delta$ rule.}
Let us turn now to introducing the last rule scheme. Observe that the invertible rules of the calculus (the $\bigwedge$ and $\bigvee$ rules and rules $\nabla$-r and $\Delta$-l) are strong enough to reduce any valid sequent $\phi\Rightarrow\psi$ to a set of sequents in a \emph{reduced form} by a backwards application of the rules. A sequent is \emph{reduced} if it is of the form
\[
\pi,
\{
\nabla\alpha\ |\ A\sqni\alpha\}\Rightarrow\{\Delta\beta\mid\beta\sqin B
\},
\lambda,
\]
with $\pi$ and $\lambda$ being finitely generated upperset and lowerset of atomic formulas (or their respective generators in the simplified notation), $A$ being
in $\UU_\omega  T_\omega^\partial\Lang$ and $B$ in $\LLL_\omega T_\omega^\partial\Lang$ (or their respective generators in the simplified notation).

The last missing bit of the calculus ${\mathbf{G}}^T_{\nabla\Delta}$ is a rule that decomposes sequents in a reduced form.
The rule $\nabla\Delta$, which we will formulate in Definition~\ref{def:therule} below, introduces the nabla modality to the left-hand side and the modality delta to the right-hand side of
a sequent simultaneously. It is the only rule whose
backwards application reduces the modal depth of a sequent.

The idea behind the rule $\nabla\Delta$ is to express what it means for a sequent in a reduced form to be refuted in a state of a coalgebra. It will describe how validity and refutation of the subformulas of the sequent is ``redistributed'' in the ``successors'' of the state. This semantical idea will be made precise in Example~\ref{ex:srd} below.
To be able to formulate the same idea in a syntactic form of a rule we employ a technical notion of {\em redistribution}.

\begin{defi}
\label{def:rd}
Fix a pair $(A,B)$ in
$(\UU_\omega T_\omega^\partial\Lang\times\LLL_\omega T_\omega^\partial\Lang)$.
A \emph{redistribution} of $(A,B)$ is an element
$\Phi$ of $T_\omega^\partial (\UU_\omega \base(A) \times \LLL_\omega \base(B))$
satisfying the following
two conditions:
\begin{enumerate}
\item
for each $\alpha$ with $A\sqni\alpha$ it holds that
$(T_\omega^\partial p_0)\Phi \mathrel{\ol{T_\omega^\partial}{\sqni}} \alpha$,
\item
for each $\beta$ with $\beta\sqin B$ it holds that
$\beta\mathrel{\ol{T_\omega^\partial}{\sqin}} (T_\omega^\partial p_1)\Phi$,
\end{enumerate}
where
$
\xymatrixcolsep{1pc}
\xymatrix@1{
\UU_\omega\Lang
&
&
\UU_\omega\Lang\times\LLL_\omega\Lang
\ar[0,2]^(.6){p_1}
\ar[0,-2]_(.6){p_0}
&
&
\LLL_\omega\Lang
}
$
are the projection maps.

We denote the collection of all redistributions
of the pair $(A,B)$ by $\rd(A,B)$.
\end{defi}

\begin{rem}
The notion of redistribution was introduced earlier in work on coalgebra automata~\cite{kupk:coal08}.
The idea of ``redistributing the subformulas'' has appeared, for the case of sets, e.g.\ in~\cite{kkv12}
and~\cite{bpv13} under the name of slim redistribution or separated slim redistribution, respectively. For reasons of simplicity we do not adopt the name slim redistribution, and speak simply about redistributions instead.
\end{rem}

We illustrate the notion of redistribution with the following
example which arises semantically. It is a crucial example
because it relates directly to soundness and a weak form of
invertibility of the modal rule we prove later in
Proposition~\ref{prop:sitherule}.

\begin{exa}
\label{ex:srd}
Suppose a sequent
\[
{\uparrow} \{\nabla \alpha \mid A \sqni \alpha \}
\Rightarrow
{\downarrow} \{ \Delta \beta \mid \beta \sqin B \}
\]
is
not valid for some pair $(A,B)$ in
$\UU_\omega T^\partial_\omega\Lang\times\LLL_\omega T^\partial_\omega\Lang$. This means that there is a coalgebra $c: X\to TX$ and a valuation
$
\xymatrix@1{
\Vdash:
\Lang
\ar[0,1]|-{\object @{/}}
&
X^\op
}
$
such that
\[
x_0\Vdash\bigwedge\limits_{A\sqni\alpha}\nabla\alpha
\ \mbox{ and }\ x_0\nVdash\bigvee\limits_{\beta\sqin B}\Delta\beta
\]
for some $x_0$ in $X$.
The bases of $A$ and $B$ in $\Sub_{\MM}\Lang$ will be denoted by
  \[
\monoInline{\base(A)}{V}{\Lang}\ 
\mbox{ and }\ 
\monoInline{\base(B)}{W}{\Lang}
  \]
respectively, with $V$ and $W$ being finite posets.
In this example we will construct a certain redistribution out of this
countermodel.
\begin{enumerate}
\renewcommand{\theenumi}{\Roman{enumi}}
\item
The construction of a redistribution. Again, the main idea of the proof is rather simple, and most of the hassle is due to keep track of types and restricting to bases to keep the reasoning on the language side finitary.

We start by defining two maps
\[
f^\sharp: X^\op\to\UU_\omega \Lang
\mbox{ and }
g^\sharp: X^\op\to\LLL_\omega \Lang
\]
by putting
\begin{align*}
f^\sharp (x) &= {\uparrow} \{a \in \base(A) \mid x \Vdash v\}, \\
g^\sharp (x) &= {\uparrow} \{b\in \base(B) \mid x \nVdash w\}.
\end{align*}
Similarly as before in proof of Proposition~\ref{prop:deltadl}, they factor through the bases of $A$ and $B$ to result in the following two maps:
    \[
f: X^\op\to\UU_\omega V\ 
\mbox{ and }\ 
g: X^\op\to\LLL_\omega W
    \]
by putting
    \[
f(x) = \{v\mid x\Vdash\base(A)(v)\}\ 
\mbox{ and }\ 
g(x) = \{w\mid x\nVdash\base(B)(w)\}.
    \]
Hence the tupling of $f$ and $g$ is of the form
\[
(f,g): X^\op\to\UU_\omega V \times \LLL_\omega W
\]
and by applying $T^\partial$ to it we obtain a map
\[
T^\partial(f,g): (TX)^\op\to T^\partial(\UU_\omega V \times \LLL_\omega W).
\]
Consider the composite
\[
\xymatrix{
X^\op
\ar[0,1]^-{c^\op}
&
(TX)^\op
\ar[0,2]^-{T^\partial(f,g)}
&
&
T^\partial(\UU_\omega V \times \LLL_\omega W)
}
\]
and define $\Phi_{x_0}$ to be its value at $x_0$ in $X^\op$:
    \[
\Phi_{x_0} := T^\partial(f,g)(c^\op(x_0)).
    \]
Because $\UU_\omega V \times \LLL_\omega W$ is a finite poset, $T^\partial$ and $T^\partial_\omega$ agree on it (the corresponding natural inclusion is the identity). Thus we may see  $\Phi_{x_0}$ as an element of $T^\partial_\omega(\UU_\omega V \times \LLL_\omega W)$. 
Therefore, using $\base(A)$ and  $\base(B)$ as maps,
    \[
\Phi = T^\partial_\omega(\UU_\omega \base(A),\LLL_\omega \base(B))(\Phi_{x_0}) = T^\partial(f^\sharp,g^\sharp)(c^\op(x_0))
    \]
is an element of $T^\partial_\omega(\UU_\omega \base(A) \times \LLL_\omega \base(B))$ as required by the definition or a redistribution.  

We claim that $\Phi$ is a redistribution of $(A,B)$.
We verify item~(1) of Definition~\ref{def:rd}, item~(2)
is verified by dual reasoning. Therefore, we want to prove
that
\[
(T_\omega^\partial p_0)\Phi
\mathrel{\ol{T_\omega^\partial}{\sqni}}
\alpha
\]
holds for every $\alpha$ with $A\sqni\alpha$, where
$p_0:\UU_\omega \Lang \times\LLL_\omega \Lang\to\UU_\omega \Lang$
is the product projection. 
By the definition of $\Phi$, we observe that
\[
(T_\omega^\partial p_0) \Phi = T_\omega^\partial\UU_\omega\base(A)(T^\partial f) c^\op(x_0) = (T^\partial f^\sharp) c^\op(x_0).
\]
For the following diagrams to work, we however need to keep things restricted to bases on the language side, and therefore we are going to use $f$ rather then $f^\sharp$.
Consider any $\alpha$ such that $A \sqni \alpha$ holds. From the unit property of the base we know that $A\in \UU_\omega T^\partial_\omega\base(A)$. So, there is a (unique) $\alpha'$ in $T_\omega^\partial V$ with $\alpha =  T^\partial_\omega \base (A)(\alpha')$. To sum up, to prove for
$
\xymatrix@1{
\sqni:
\Lang
\ar[0,1]|-{\object @{/}}
&
\UU_\omega\Lang
}
$
that
\[
(T_\omega^\partial p_0)\Phi
\mathrel{\ol{T_\omega^\partial}{\sqni}}
\alpha,
\]
means to prove that
\[
T_\omega^\partial\UU_\omega\base(A)(T^\partial f) c^\op(x_0)
\mathrel{\ol{T_\omega^\partial}{\sqni}}
T^\partial \base (A)(\alpha'),
\]
and in turn it suffices to prove, for the restricted
$
\xymatrix@1{
\sqni:
V
\ar[0,1]|-{\object @{/}}
&
\UU_\omega V,
}
$
that
\[
(T_\omega^\partial p_0)\Phi_{x_0}
\mathrel{\ol{T_\omega^\partial}{\sqni}}
\alpha'.
\]
This by definition of $\Phi_{x_0}$ boils down to proving that
\begin{equation}
\label{eq:ex-5.6}
(T^\partial f (c^\op(x_0)))
\mathrel{\ol{T_\omega^\partial}{\sqni}}
\alpha'
\end{equation}
holds for any $\alpha$ such that $A\sqni\alpha$ holds.

Given any $x$ in $X$ and $v \in V$ such that $x \Vdash_c \base(A)(v)$ holds, we know that $f(x) \sqni v$ holds by the definition of $f$. The monotonicity of the membership relation $\sqni$ entails that for every $\phi$ in $\UU_\omega V$ with $\phi \leq f(x)$, the relation $\phi \sqni v$ holds as well. Thus the following lax triangle in $\Rel$
\[
\xymatrix{
V
\ar[0,4]|-{\object @{/}}^-{{\sqni}}
\ar[1,2]|-{\object @{/}}_{{\Vdash}(-,\base(A)-)\phantom{MM}}
&
&
&
&
\UU_\omega V
\\
&
&
X^\op
\ar[-1,2]|-{\object @{/}}_{\phantom{M} f_\diamond}
\ar@{}[-1,0]|{\uparrow}
&
&
}
\]
commutes. Its image under $\ol{T^\partial}$ is the lax triangle
\begin{equation}
\label{eq:ex-5.7}
\xymatrix{
T^\partial V
\ar[0,4]|-{\object @{/}}^-{\ol{T^\partial} {\sqni}}
\ar[1,2]|-{\object @{/}}_{\ol{T^\partial} {\Vdash}(-,T^\partial \base(A)-)\phantom{MM}}
&
&
&
&
T^\partial  \UU_\omega V.
\\
&
&
(TX)^\op
\ar[-1,2]|-{\object @{/}}_{\phantom{M} (T^\partial f)_\diamond}
\ar@{}[-1,0]|{\uparrow}
&
&
}
\end{equation}
We prove~\eqref{eq:ex-5.6} using~\eqref{eq:ex-5.7}.
Consider any $\alpha$ such that $A \sqni \alpha$ holds. Recall $\alpha'$ in $T^\partial V$ is the unique element with $\alpha =  T^\partial \base (A)(\alpha')$ (as $V$ is a  finite poset, we are allowed to write $T^\partial$ in place of $T^\partial_\omega$ here).
We instantiate the diagram for $c^\op(x_0)$ in $(T^\partial X)^\op$, and $\alpha'$ in $T^\partial V$. 

We know from assumption that $x_0 \Vdash_c \nabla\alpha$. Let $\monoInline{\base(\alpha)}{Z}{\Lang}$ and $z\in Z$ be the unique element with $T^\partial_\omega\base(\alpha)(z) = \alpha$.
By the semantic definition we obtain that 
\[
c(x_0) \mathrel{\ol{T^\partial} {\Vdash_c}}(-,T^\partial\base(\alpha)-) z.
\]
But observe that $\base(\alpha)\sqsubseteq \base(A)$. 
Therefore this is equivalent to
\[
c(x_0) \mathrel{\ol{T^\partial} {\Vdash_c}}(-,T^\partial\base(A)-) \alpha'.
\]
Applying the diagram we finally obtain that $T^\partial f (c^\op(x_0)) \ol{T^\partial} \sqni \alpha'$ as required.
We consequently use the fact that since $V$ and $\UU_\omega V$ are finite posets, $T^\partial$ and $T^\partial$ coincide on them, and consequently the lifted relations $\ol{T^\partial} \sqni$ and $\ol{T^\partial_\omega} \sqni$ coincide.
\item
\label{item:IIsrd}
The redistribution $\Phi$ constructed above, has the following
additional property: the sequent
\[
p_0(z)\Rightarrow p_1(z)
\]
is refutable, for every $z \in \base(\Phi)$.
\begin{enumerate}
\renewcommand{\theenumi}{\alph{enumi}}

\item
Recall the map
$(f^\sharp,g^\sharp):X^\op\to\UU_\omega \Lang \times\LLL_\omega \Lang$
and consider its $(\EE,\MM)$-factorisation
\[
\xymatrix{
X^\op
\epi[0,1]_-{e}
&
Y
\mono[0,1]_-{m}
&
\UU_\omega \Lang \times \LLL_\omega \Lang
\ar@{<-} `u[ll] `[ll]_-{(f,g)} [ll]
}
\]
We will prove the inclusion
\[
\base(\Phi) \subseteq m.
\]
The diagram
\[
\xymatrix{
\bOne
\mono[0,1]^-{c^\op(x_0)}
\mono[1,2]_-{\Phi}
&
(TX)^\op
\ar[0,1]^-{T^\partial e}
\ar[1,1]^{T^\partial(f^\sharp,g^\sharp)}
&
T^\partial Y
\mono[1,0]^{T^\partial m}
\\
&
&
T^\partial (\UU_\omega \Lang \times \LLL_\omega \Lang)
}
\]
commutes: the left triangle commutes by the definition of $\Phi$, and the right triangle commutes by the $(\EE,\MM)$-factorisation of $(f^\sharp,g^\sharp)$. This proves that $\Phi\in T^\partial m$ holds. Thus $\base(\Phi_{x_0}) \subseteq m$ holds by the definition of base.

\item
We will show that for every $z \in \base(\Phi)$ there is a state $x_z \in X$ such that
\[
p_0(z) \Rightarrow p_1(z)
\]
is not valid in $x_z$.

Since $z \in \base(\Phi)$ holds, we know that $z \in m$ holds as well. From this and from the surjectivity of the mapping $e$ we can thus define $x_z$ to be any element from $X$ such that
\[
z = (m \cdot e)(x_z)
\]
holds. Then we know by the $(\EE,\MM)$-factorisation of $(f^\sharp,g^\sharp)$ that
\[
(m \cdot e)(x_z) = (f^\sharp,g^\sharp)(x_z).
\]
The sequent $p_0(z) \Rightarrow p_1(z)$ can therefore be written as \[
p_0 \cdot (f^\sharp,g^\sharp) (x_z) \Rightarrow p_1 \cdot (f^\sharp,g^\sharp) (x_z),
\]
and this in turn is the sequent
\[
f^\sharp(x_z) \Rightarrow g^\sharp(x_z).
\]
By the definition of $f^\sharp$ and $g^\sharp$, all formulas in $f^\sharp(x_z)$ are valid in $x_z$, while no formula in $g^\sharp(x_z)$ is valid in $x_z$.
\end{enumerate}
\end{enumerate}
\end{exa}
Now we can state the main modal rule of the calculus ${\mathbf{G}}^T_{\nabla\Delta}$, explain how it reads, and show its main properties.
\begin{defi}[The modal $\nabla\Delta$ rule]
\label{def:therule}
We formulate the rule $\nabla\Delta$ as follows:
\begin{equation}
\fbox{
$
\Inference{
           \{
           p_0(z^\Phi)
           \Rightarrow
           p_1(z^\Phi)
           \mid
           \Phi\mbox{ in }\rd(A,B)
           \}
           }
          {
          \{
          \nabla\alpha\mid A\sqni\alpha
          \}
          \Rightarrow
          \{
          \Delta\beta\mid\beta\sqin B
          \}
          }
          {\text{$\nabla\Delta$}}
          \text{$\forall \Phi.\ z^\Phi \in \base(\Phi)$}
$
}
\end{equation}
where the elements in $\rd(A,B)$ are of the form $\Phi \in T^\partial_\omega(\UU_\omega\Lang\times\LLL_\omega\Lang)$, and
\[
\xymatrixcolsep{1pc}
\xymatrix@1{
\UU_\omega\Lang
&
&
\UU_\omega\Lang\times\LLL_\omega\Lang
\ar[0,2]^-{p_1}
\ar[0,-2]_-{p_0}
&
&
\LLL_\omega\Lang
}
\]
are the projection maps of the product $\UU_\omega\Lang\times\LLL_\omega\Lang$.
\end{defi}
Recall from Part~\ref{item:IIsrd} of Example~\ref{ex:srd} that choosing an element $z^\Phi \in \base(\Phi)$, we obtain a finitely generated upperset of
formulas $p_0(z^\Phi)$
and a  finitely generated lowerset of formulas
$p_1(z^\Phi)$, therefore
$p_0(z^\Phi)\Rightarrow p_1(z^\Phi)$
 is a well-formed sequent.
The rule $\nabla\Delta$ says that the conclusion is provable whenever
for each redistribution $\Phi$ of $(A,B)$ there exists some
$z^\Phi \in \base(\Phi)$ such that the sequent
$p_0(z^\Phi)\Rightarrow p_1(z^\Phi)$
is provable. In other words, it says that if the sequent in the conclusion
is refutable, there must be a redistribution $\Phi$ of the pair
$(A,B)$ (in fact the one of Example~\ref{ex:srd}) such that all the sequents
$p_0(z^\Phi)\Rightarrow p_1(z^\Phi)$
are refutable.

\begin{exa}\label{ex:proof}
Continuing on Examples~\ref{ex:RL} and~\ref{ex:therules}, we illustrate the above definition with two instances of the rule:
\begin{enumerate}
\item
Consider $T^\partial_\omega = \LLL_\omega$ and let us compute $\rd(\{\emptyset\},\{\{\bot\}\})$. Any such redistribution $\Phi$ in $\LLL_\omega(\UU_\omega\Lang\times\LLL_\omega\Lang)$ has to satisfy the following:
\[(\LLL_\omega p_0)\Phi\mathrel{\ol{\LLL_\omega}{\sqni}}\emptyset \mbox{  and  } \{\bot\}\mathrel{\ol{\LLL_\omega}{\sqin}}(\LLL_\omega p_1)\Phi.\]
The first condition is met by $(\LLL_\omega p_0)\Phi = \emptyset$ only, while the second condition requires $(\LLL_\omega p_1)\Phi \neq \emptyset$, which together is impossible. The collection of redistributions is therefore empty and the following is a correct instance of the $\nabla\Delta$ rule (in fact an axiom):

\begin{prooftree}
\AXC{$\emptyset$}
\LL{$\nabla\Delta$}
\UIC{$\nabla\emptyset\Rightarrow\Delta\{\bot\}$}
\end{prooftree}
Putting the above observations together with Example~\ref{ex:therules} (1) we obtain a simple proof in the simplified notation:

\begin{prooftree}
\AXC{$\emptyset$}
\LL{$\nabla\Delta$}
\UIC{$\nabla\emptyset\Rightarrow\Delta\{\bot\}$}
\LL{$\nabla$-r}
\UIC{$\nabla\emptyset\Rightarrow\nabla\{\bot\}$}
\end{prooftree}

\item
Consider $T^\partial_\omega = \mathbb{N}^\op\times\Id$ and let us compute $\rd(\{(3,b),(8,c)\},\{(n,\bot)\})$ where $n<5$. Any such redistribution $\Phi$ is an element of $\mathbb{N}^\op \times (\UU_\omega\Lang\times\LLL_\omega\Lang)$, i.e.\ of the form $(m,(\phi,\psi))$ satisfying firstly
\[(m,\phi)(\mathrel{\ol{\mathbb{N}^\op\times\Id}{\sqni}})(3,b) \mbox{ and }
(m,\phi)(\mathrel{\ol{\mathbb{N}^\op\times\Id}{\sqni}})(8,c),
\]
which implies that $m\geq 8$ and $\phi\sqni b$ and $\phi\sqni c$. Secondly,
\[
(n,\bot)(\mathrel{\ol{\mathbb{N}^\op\times\Id}{\sqin}})(m,\psi),
\]
which implies that $n\geq m$. But since $m\geq 8$ and $n<5$, this is clearly impossible. Thus there is no such redistribution and the following is a correct instance of the rule, whenever $n<5$:
\begin{prooftree}
\AXC{$\emptyset$}
\LL{$\nabla\Delta$}
\UIC{$\nabla(3,b),\nabla(8,c)\Rightarrow\Delta(n,\bot)$}
\end{prooftree}
Let us compute $\rd(\{(3,b),(8,c)\},\{(n,b\wedge c)\})$ for any $n$. Any such redistribution $\Phi$ is of the form $(m,(\phi,\psi))$ satisfying, as before, $m\geq 8$ and $\phi\sqni b$ and $\phi\sqni c$. Moreover, it should satisfy
\[
(n,b\wedge c)(\mathrel{\ol{\mathbb{N}^\op\times\Id}{\sqin}})(m,\psi),
\]
which implies that $n\geq m$ and $b\wedge c\sqin \psi$. For $n<8$ the collection of redistributions is again empty. For $n\geq 8$ we know that any redistribution $(m,(\phi,\psi))$ contains in its base the pair $(\phi,\psi)$ where $\phi\Rightarrow\psi$ is a provable sequent (using possibly some weakening inferences). Therefore the following is a correct proof in the simplified notation:
\begin{prooftree}
\AXC{$b\Rightarrow b$}\AXC{$c\Rightarrow c$}
\LL{$\bigwedge$-r}
\BIC{$b,c\Rightarrow b\wedge c$}
\LL{w}
\UIC{$\phi\Rightarrow\psi$}
\LL{$\nabla\Delta$}
\UIC{$\nabla(3,b),\nabla(8,c)\Rightarrow\Delta(n,b\wedge c)$}
\end{prooftree}
Put together with Example~\ref{ex:therules} (2) we obtain the following proof, where the left part of the tree covers cases for $n<8$ and the right part of the tree covers cases for $n\geq 8$ and all the existing redistributions:
\begin{prooftree}
\AXC{$\emptyset$}
\LL{$\nabla\Delta$}
\UIC{$\nabla(3,b),\nabla(8,c)\Rightarrow\Delta(n,\bot)$}
\AXC{$\hdots$}
\AXC{$b\Rightarrow b$}\AXC{$c\Rightarrow c$}
\LL{$\bigwedge$-r}
\BIC{$b,c\Rightarrow b\wedge c$}
\LL{w}
\UIC{$\phi\Rightarrow\psi$}
\LL{$\nabla\Delta$}
\UIC{$\nabla(3,b),\nabla(8,c)\Rightarrow\Delta(n,b\wedge c)$}
\AXC{$\hdots$}
\LL{$\nabla$-r}
\QuaternaryInfC{$\nabla(3,b),\nabla(8,c)\Rightarrow\nabla(5,b\wedge c)$}
\end{prooftree}

\end{enumerate}
\end{exa}
The rule $\nabla\Delta$ satisfies the subformula property, similar in spirit to the properties discussed in Remark~\ref{rem:subpropinv}.

\begin{rem}[Subformula property]
\label{rem:subproptherule}
Since any redistribution $\Phi$ has the form
\[
\xymatrix{
\bOne
\mono[0,2]
\mono[1,2]_{\Phi}
&
&
T_\omega^\partial(\UU_\omega V\times\LLL_\omega W)
\mono[1,0]^{T_\omega^\partial(\UU_\omega\base(A) \times \LLL_\omega\base(B))}
\\
&
&
T_\omega^\partial(\UU_\omega\Lang\times\LLL_\omega\Lang),
}
\]
we get by the definition of base the inclusion
\[
\base(\Phi) \subseteq \UU_\omega\base(A) \times \LLL_\omega\base(B).
\]
Since the image of
$\UU_\omega\base(A)\times\LLL_\omega\base(B)$
are the subformulas of the conclusion of the rule $\nabla \Delta$, this observation tells
us that the rule $\nabla \Delta$ satisfies the subformula property.
\end{rem}
\begin{rem}[One-step nature of the rules]
Coalgebras naturally capture one-step behaviour. 
As introduced by~\cite{Patt03},
both the semantics and syntax of coalgebraic logics can be stratified in layers of transition and modal depth. In particular, proof systems (both Hilbert and Gentzen style) can be presented via one-step axioms with no nesting modalities and rules which concern one layer of modalities only, and their completeness demonstrated by an inductive argument in a one-step manner~\cite{kkv12,KP11}.
We do not adopt the one-step formalism systematically in this paper, however, the construction of the language in Subsection~\ref{ssec:lang} allows for such a treatment, the semantics could be presented in a one-step manner following~\cite{kkv12}, 
and the rules of the sequent calculus allow for a purely one-step reformulation. It is roughly because all the rules except the $\nabla\Delta$ rule operate within the same layer of the language, while the $\nabla\Delta$ rule strips precisely one layer of modalities. 
Namely, in the proof of the following proposition,  
constructing the model refuting a sequent 
\[
\pi,\{\nabla\alpha\mid A\sqni\alpha\}
\Rightarrow
\{\Delta\beta\mid\beta\sqin B \},\lambda
\]
can be seen as a part of presenting a completeness proof by step-by-step method.
\end{rem}
The rule $\nabla \Delta$ is sound, and even invertible in a certain technical sense which we explain  and prove in the following proposition.
\begin{prop}[Soundness and invertibility of the $\nabla\Delta$ rule]
\label{prop:sitherule}
Let $A$, $B$ and their bases be given as in
Definition~\ref{def:therule} and let $\pi$ and $\lambda$
consist of atomic formulas. The following are equivalent:
\begin{enumerate}
\item
The sequent
\[
\pi,\{\nabla\alpha\mid A\sqni\alpha\}
\Rightarrow
\{\Delta\beta\mid\beta\sqin B \},\lambda
\]
is refutable.
\item
The sequent $\pi\Rightarrow\lambda$ is refutable, and there
exists a redistribution $\Phi$ of $(A,B)$
such that for all $z \in \base(\Phi)$ the sequents
\[
p_0(z)\Rightarrow p_1(z)
\]
are refutable.
\end{enumerate}
\end{prop}
\begin{proof}
(1) implies (2).
Assume that the sequent
\[
\pi,\{\nabla\alpha\mid A\sqni\alpha\}
\Rightarrow
\{\Delta\beta\ |\ \beta\sqin B \},\lambda
\]
is refutable. Then there is a coalgebra $c:X\to TX$,
a monotone valuation $\Vdash$ on $c$ and a state $x_0$
in $X$ such that $x_0$ refutes the sequent
$\pi\Rightarrow\lambda$,
and $(x\Vdash\bigwedge\limits_{A\sqni\alpha}\nabla\alpha)$ and
$(x\nVdash\bigvee\limits_{\beta\sqin B}\Delta\beta)$.

Then the redistribution $\Phi$ of $(A,B)$, defined
in Example~\ref{ex:srd}, has the property that all the
sequents
$
p_0(z)\Rightarrow p_1(z)
$
are refutable, see
Part~\eqref{item:IIsrd} of Example~\ref{ex:srd}.

\medskip\noindent
(2) implies (1).
Assume a redistribution $\Phi$ in $\rd(A,B)$ is given.
We moreover assume that for all $z \in \base(\Phi)$, the sequent
$p_0(z)\Rightarrow p_1(z)$
is refuted by some valuation $\Vdash_z$
on a coalgebra $c_z: X_z\to TX_z$ in a state $x_z$.

We define a new coalgebra $c:X\to TX$, a point $x_0$ in $X$,
and a valuation $\Vdash$
such that $x_0$ will validate $\nabla\alpha$ for all
$\alpha$ in $A$, and $x_0$ will refute $\Delta\beta$ for
all $\beta$ in $B$.

\begin{enumerate}
\renewcommand{\theenumi}{\alph{enumi}}
\item
The definition of $c:X\to TX$ and $x_0$ in $X$.
We will denote the base of the redistribution $\Phi$ by
\[
\xymatrixcolsep{1pc}
\xymatrix{
\base(\Phi):Z
\mono[0,2]
&
&
\UU_\omega\Lang\times\LLL_\omega\Lang
}.
\]
Therefore, by the properties of base, there is a
$w_\Phi$ in $T^\partial Z$ such that
\[
T^\partial\base(\Phi)(w_\Phi) = \Phi.
\]
(and we are allowed to write here $T^\partial$ in place of $T^\partial_\omega$ because $Z$ is a finite poset).
Let $Z_d$ denote the discrete underlying poset of $Z$.
We denote by $e:Z_d\to Z$ the obviously monotone mapping
that is identity on the elements.
The monotone map $h:Z_d\to X^\op$ is defined by putting $h(z)=x_z$.

Define
\[
X=\coprod_{z \in Z_d} X_z +\bOne
\]
and denote the unique element of $\bOne$ by $x_0$.
The coalgebra map $c:X\to TX$ is defined on the individual
components as follows:
\[
\xymatrix{
\bOne
\ar[1,0]_{{\mathsf{inj}}}
\ar[1,2]^{\phantom{MM}x_0\mapsto (T^\partial h)(w^\Phi)}
&
&
\\
X
\ar[0,2]^-{c}
&
&
TX
\\
X_z
\ar[-1,0]^{{\mathsf{inj}}_z}
\ar[0,2]_-{c_z}
&
&
TX_z
\ar[-1,0]_{T{\mathsf{inj}}_z}
}
\]
where ${\mathsf{inj}}$ and ${\mathsf{inj}}_z$ are the coproduct
injections, and $w^\Phi$ is the element of $T^\partial Z_d$ with the property that
\[
T^\partial\base(\Phi) \cdot T^\partial e(w^\Phi)=\Phi.
\]
\item
The valuation $\Vdash$ on $c:X\to TX$.

In all states from all $X_z$ the valuation of atoms
remains unchanged, while in $x_0$ we satisfy
all atoms in $\pi$ and refute all atoms in $\lambda$.
This is possible because the sequent $\pi\Rightarrow\lambda$
is refutable by assumption.
\item
We need to prove $x_0\Vdash\bigwedge\nabla\alpha$
for each $\alpha$ in $A$, and
that $x_0\nVdash\Delta\beta$ for each $\beta$ in $B$.
Equivalently, we can prove that
$c(x_0)\mathrel{\ol{T^\partial}{\Vdash}}\alpha$
for each $\alpha$ in $A$, and
$c(x_0)\mathrel{\ol{T}{\nVdash}}\beta$ for each $\beta$ in $B$. We will prove only $c(x_0)\mathrel{\ol{T^\partial}{\Vdash}}\alpha$
for each $\alpha$ in $A$, the second property is verified
by dual reasoning.

We consider $\monoInline{\base(A)}{V}{\Lang}$. 
We first claim that there the following lax diagram in $\Rel$ commutes:
\[
\xymatrixcolsep{3pc}
\xymatrix{
V
\ar[0,4]|-{\object @{/}}^-{\Vdash(-,\base(A)-)}
\ar[1,0]|-{\object @{/}}_{\sqni}
&
&
&
&
X^\op
\ar[1,0]|-{\object @{/}}^{h^\diamond}
\\
\UU_\omega V
\ar[0,1]|-{\object @{/}}_-{(\UU_\omega\base(A))_\diamond}
&
\UU_\omega\Lang
\ar[0,1]|-{\object @{/}}_-{(p_0)^\diamond}
\ar@{} [-1,3]|-{\nearrow}
&
\UU_\omega\Lang\times\LLL_\omega\Lang
\ar[0,1]|-{\object @{/}}_-{(\base(\Phi))^\diamond}
&
Z
\ar[0,1]|-{\object @{/}}_-{e^\diamond}
\ar@{<-} `d[ll] `[ll]|-{\object @{/}}^-{(p_0 \cdot \base(\Phi))^\diamond} [ll]
&
Z_d
}
\]
To see that this is the case, fix a $v\in V$ and a $z$ in $Z_d$
such that $v$ and $z$ are related by the first-down-then-right
passage of the diagram.

Hence there is some $\phi$ in $\UU_\omega\Lang$ with $\phi\sqni \base(A)(v)$,
some $(\phi',\psi')$ with
$\phi'\supseteq\phi$ and therefore $\phi'\sqni \base(A)(v)$,
and $\base(\Phi)(z)\leq_{\UU_\omega\Lang\times\LLL_\omega\Lang}(\phi',\psi')$.
This entails that $p_0(\base(\Phi)(z))\sqni \base(A)(v)$, and $\base(A)(v)$ is valid in $x_z$ by assumption. Therefore the pair $v$ and $z$ are related
by the first-right-then-down passage of the diagram.

We will use the image under $\ol{T^\partial}$
of the above diagram:
\[
\xymatrixcolsep{3.5pc}
\xymatrix{
T^\partial V
\ar[0,4]|-{\object @{/}}^-{\ol{T^\partial}{\Vdash}(-,T^\partial\base(A)-)}
\ar[1,0]|-{\object @{/}}_{\ol{T^\partial}{\sqni}}
&
&
&
&
(TX)^\op
\ar[1,0]|-{\object @{/}}^{(T^\partial h)^\diamond}
\\
T^\partial\UU_\omega V
\ar[0,1]|-{\object @{/}}_-{(T^\partial \UU_\omega\base(A))_\diamond}
&
T^\partial\UU_\omega\Lang
\ar[0,1]|-{\object @{/}}_-{(T^\partial p_0)^\diamond}
\ar@{} [-1,3]|-{\nearrow}
&
T^\partial\UU_\omega\Lang\times T^\partial\LLL_\omega\Lang
\ar[0,1]|-{\object @{/}}_-{(T^\partial\base(\Phi))^\diamond}
&
T^\partial Z
\ar[0,1]|-{\object @{/}}_-{(T^\partial e)^\diamond}
&
T^\partial Z_d
}
\]
Consider a fixed $\alpha$ in $A$. Let $\monoInline{\base(\alpha)}{U}{\Lang}$, and let $u\in U$ be the unique element with $T^\partial\base(\alpha)(u) = \alpha$.
Because $\base(\alpha)\sqsubseteq \base(A)$, there is a unique $\alpha'$ in $T^\partial V$ with $\alpha = T^\partial\base(A)(\alpha')$.

Apply the diagram to $\alpha'$ in $T^\partial V$, $(T^\partial p_0)\Phi$
in $T^\partial\UU_\omega\Lang$, and the unique 
$w_\Phi$ in $T^\partial Z$ such that
\[
T^\partial\base(\Phi)(w_\Phi) = \Phi.
\]
Since $\alpha'$ and $w_\Phi$ are related by the
first-down-then-right passage of the diagram,
there exists $\xi$ in $ (TX)^\op$  with $c(x_0)\leq_{(TX)^\op}\xi$, and such that
\[
\xi\mathrel{\ol{T^\partial}{\Vdash}(-,T^\partial\base(A)-)}\alpha'.
\]
Therefore $c(x_0)\geq_{TX}\xi$
and, consequently,
\[
c(x_0)\mathrel{\ol{T^\partial}{\Vdash}(-,T^\partial\base(A)-)}\alpha'.
\]
But this is equivalent to 
\[
c(x_0)\mathrel{\ol{T^\partial}{\Vdash}(-,T^\partial\base(\alpha)-)}u,
\]
and this 
proves that $x_0$ validates $\nabla\alpha$, as required.

We remark that, in the above diagram, that since $V$ (and also $\UU_\omega V$) is a finite poset, $T^\partial$ and $T^\partial_\omega$ coincide on it. Therefore $\alpha$ and $\alpha'$ are typed properly. For the same reason, also the lifted relations $\ol{T^\partial}{\sqni}$ and $\ol{T^\partial_\omega}{\sqni}$ coincide.\qedhere 
\end{enumerate}
\end{proof}

\subsection{Completeness}
Assume a valid sequent $\phi\Rightarrow\psi$ is given. 
Recall that in a simplified notation, we can identify the finitely generated lowerset $\psi$ (and the finitely generated upperset $\phi$) with the finite (discrete) poset of its generators  $g(\psi)$ (resp.\ $g(\phi)$).
We will show by an inductive argument that the sequent is provable in the calculus ${\mathbf{G}}^T_{\nabla\Delta}$, using the invertibility of the rules of the calculus.
To this end we need to define a measure of syntactic complexity of a sequent in such a way that any backward application of a logical rule of the calculus strictly decreases the defined measure.

\begin{defi}[Measure on sequents]
\label{def:measure}
For each formula $a$ in $\Lang$, we define its complexity
on the left-hand side of a sequent, and on the right-hand side of
the sequent simultaneously as follows,
counting modal formulas as atoms with a slightly bigger complexity.
\[
\begin{array}{rclcrcl}
l(p)
&=&
0
&\phantom{MM}&
r(p)
&=&
0
\\[8pt]
l(\bigwedge\phi)
&=&
\displaystyle\sum_{a\in g(\phi)}l(a)+1
&&
r(\bigwedge\phi)
&=&
\displaystyle\sum_{a\in g(\phi)}r(a)+1
\\[15pt]
l(\bigvee\psi)
&=&
\displaystyle\sum_{a\in g(\psi)}l(a)+1
&&
r(\bigvee\psi)
&=&
\displaystyle\sum_{a\in g(\psi)}r(a)+1
\\[15pt]
l(\Delta\beta)
&=&
3
&&
r(\nabla\alpha)
&=&
3\\[8pt]
l(\nabla\alpha)
&=&
2
&&
r(\Delta\beta)
&=&
2
\end{array}
\]
For a sequent $\phi\Rightarrow\psi$ we define its complexity
as
\[
k(\phi\Rightarrow\psi) = \sum_{a\in g(\phi)}l(a)+\sum_{b\in g(\psi)}r(b),
\]
and its modal depth (recall Definition~\ref{def:subflas}) by
\[
d(\phi\Rightarrow\psi)
=
\max(\{d(a)\mid a\in g(\phi)\}\cup\{d(b)\mid b\in g(\psi)\}).
\]
Finally, we define the measure of the sequent $\phi\Rightarrow\psi$
to be the pair of natural numbers
\[
m(\phi\Rightarrow\psi) = (d(\phi\Rightarrow\psi),k(\phi\Rightarrow\psi)).
\]
We consider the pairs to be ordered \emph{lexicographically}. The measure is defined exactly as in~\cite{bpv13}.
\end{defi}

\begin{prop}\label{prop:measure}
Fix a sequent $\phi\Rightarrow\psi$.
A backward application of any rule of the calculus, except weakening rules\footnote{In a backward application of a weakening inference, the set of generators of the corresponding upperset or lowerset may actually get bigger, even if the actual upperset or lowerset does not.}, to the sequent yields sequents with strictly smaller measure.
\end{prop}

\begin{proof}
In each backward application
of one of the propositional rules, $k(\phi\Rightarrow\psi)$
strictly decreases while modal depth remains unchanged, i.e.,
the measure of any assumption is strictly smaller then the measure
of the conclusion.

We inspect the modal rules. Consider a backward application of
the $\Delta$-l rule with the conclusion
$\phi,\Delta\beta\Rightarrow\psi$. Its complexity is
\[
k(\phi,\Delta\beta\Rightarrow\psi)
=
\sum_{a\in g(\phi)}l(a)+3+\sum_{b\in g(\psi)}r(b).
\]
Complexity of any of the assumptions $\phi,\nabla(T^\partial_\omega\bigwedge)\Phi\Rightarrow\psi$ is strictly smaller:
\[
k(\phi,\nabla(T^\partial_\omega\bigwedge)\Phi\Rightarrow\psi)
=
\sum_{a\in g(\phi)}l(a)+2+\sum_{b\in g(\psi)}r(b).
\]
Since $\Phi \in T^\partial_\omega\UU_\omega\base(\beta)$,
from Remark~\ref{rem:subpropinv} it follows that
the modal depth of the sequent remains unchanged.
The case of the modal rule $\nabla$-r is similar.

Consider a backward application of the $\nabla\Delta$
rule with the conclusion
$
\{\nabla\alpha\mid A\sqni\alpha\}
\Rightarrow
\{\Delta\beta\mid\beta\sqin B \}
$.
From the type of any $\Phi$ in $\rd(A,B)$ and from
Remark~\ref{rem:subproptherule} it follows that
the measure of each of the sequents
$p_0(z^\Phi)\Rightarrow p_1(z^\Phi)$
is strictly smaller because in this case the modal depth
of the assumptions is strictly smaller than that of the conclusion.
\end{proof}

The basic idea is we can apply the rules of the calculus (except the weakening rules) to a sequent backwards, simplifying it in terms of the measure, and using the invertibility of the rules (or a weak form of invertibility in the case of the modal rule). Reaching an irreducible, atomic, sequent, we check for provability. This can be seen as a backward proof-search in a variant of the calculus with the weakening rules built in the axioms.  
Now we can finally prove that the calculus ${\mathbf{G}}^T_{\nabla\Delta}$ is sound and complete.

\begin{thm}[Soundness and Completeness of the calculus ${\mathbf{G}}^T_{\nabla\Delta}$]
Each sequent $\phi\Rightarrow\psi$ is valid
if and only if it is provable.
\end{thm}
\begin{proof}
Soundness of the calculus ${\mathbf{G}}^T_{\nabla\Delta}$ can be proved by
a routine induction on the depth of a proof in the calculus,
using in particular appropriate directions
of Propositions~\ref{prop:deltadl},~\ref{prop:nablaviadelta}
and~\ref{prop:sitherule} establishing soundness of
the modal rules.

The completeness can be proved by induction on the measure.
Assume a valid sequent $\phi\Rightarrow\psi$ is given.
If the measure of the sequent is $(0,0)$, it consists
of atoms and can only be valid if there are some atoms $\phi\sqni p$ and $q\sqin\psi$ with $p\leq_\Lang q$. In that case the sequent
is provable from an axiom by weakening rules.

Suppose that $(d(\phi\Rightarrow\psi),k(\phi\Rightarrow\psi))>(0,0)$,
meaning the sequent contains at least one logical operator.
We distinguish two cases:
\begin{enumerate}
\item
$\phi\Rightarrow\psi$ is not of the reduced form. Then some of
the propositional rules, or the $\Delta$-l or the $\nabla$-r rule
can be applied to it backwards. Any such rule is
by Propositions~\ref{prop:deltadl} and~\ref{prop:nablaviadelta}
invertible and therefore such an application preserves validity.
Moreover, all the assumptions of such  a rule have strictly
smaller measure by Proposition~\ref{prop:measure}. We may therefore apply the induction hypothesis
and conclude that all the assumptions of such a rule application,
being valid, are also provable, and so is, by applying the rule
forward, its conclusion.
\item
$\phi\Rightarrow\psi$ is of the form
$
\pi, \{\nabla\alpha\mid A\sqni\alpha\}
\Rightarrow
\{\Delta\beta\mid \beta\sqin B \}, \lambda
$
with $\pi$ and $\lambda$ consisting of atomic formulas.
Then either $\pi\Rightarrow\lambda$ is valid, and therefore
provable as in the basic case. The sequent itself is then
provable by weakening rules.

Or, by Proposition~\ref{prop:sitherule}, for each redistribution
there is some $z^\Phi \in \base(\Phi)$ such that the sequent
$p_0(z^\Phi)\Rightarrow p_1(z^\Phi)$
is valid. But any such sequent has strictly smaller measure by Proposition~\ref{prop:measure}
and therefore we may apply the induction hypothesis and conclude
that any such sequent, being valid, is therefore also provable.
Then the sequent
$
\pi, \{\nabla\alpha\mid A\sqni\alpha\}
\Rightarrow
\{\Delta\beta\ |\ \beta\sqin B \}, \lambda
$
is provable by the rule $\nabla\Delta$, and weakening rules. \qedhere
\end{enumerate}
\end{proof}

\section{Concluding remarks}
\label{sec:conclusion}
We have shown that a finitary Moss' logic can be meaningfully defined for coalgebras in the category $\Pos$ of posets and monotone maps, and that it is expressive and also complete. All the definitions and proofs are parametric in the coalgebra functor $T$, which is required to be locally monotone, to preserve exact squares, and the syntax functor $T^\partial_\omega$, which is required to preserve intersections of subobjects.

There is still a lot we do not know about endofunctors of $\Pos$. Namely we are interested in the following open problems:
\begin{enumerate}
\item Characterise endofunctors of $\Pos$ that preserve exact squares.
\item Characterise endofunctors of $\Pos$ that preserve order
embeddings.
\item Among the endofunctors of $\Pos$ that preserve order embeddings, characterise those such endofunctors that moreover preserve finite intersections.
\end{enumerate}
There are two interesting questions regarding the logic introduced in this paper that we have not addressed and leave for potential future study:
\begin{enumerate}
\item Is it the case that for a posetification $T_\omega^\prime$ (see~\cite{balan+kurz}) of a $\Set$ endofunctor $T_\omega$, the Moss' logic for $T^\prime$-coalgebras as introduced in this paper is the positive fragment of the Moss' logic for $T$-coalgebras? The corresponding result for the logic of predicate liftings was obtained in~\cite{balan+kurz+velebil}.
\item What is the relation of the Moss' logic for coalgebras for $\Pos$-endofunctors to the logic of all monotone predicate liftings considered in~\cite{kkvpos:12}?
\end{enumerate}
Let us briefly comment on (1). First of all, our definition of the language in Subsection~\ref{ssec:lang} uses $\LLL_\omega$ and $\UU_\omega$ for arities of disjunction and conjunction. We have explained that this makes perfect sense in the poset setting in Remark~\ref{rem:arities} (it is worth noting that even with a discrete poset of atomic propositions, our construction produces an ordered language). But these functors are not poset extensions or liftings of the finitary powerset functor in $\Set$\footnote{The functors  $\LLL,\UU$ can however be obtained as quotients of certain liftings of the powerset functor to the category of preorders and monotone maps~\cite{balan+kurz}.}, which classical Moss' logic uses as arity of both the connectives. While we could in principle use the self-dual finitary convex powerset functor $\PP^c_\omega$, which is the posetification of the finitary powerset functor in $\Set$~\cite{balan+kurz}, for the arity of both conjunction and disjunction (and this choice would actually produce a discrete language, starting from a discrete poset of atomic propositions), this seems to us to make sense precisely for the purpose of studying positive fragments of $\Set$-based coalgebraic logic (which is not the aim of this paper), and not so much for studying the $\Pos$-based coalgebraic logic. Had we done so, it would not be difficult to adjust the rest of the syntactic machinery, including the Definition~\ref{def:collections_L_R} of the collections $R_\alpha$ and $L_\beta$ and Definition~\ref{def:rd} of a redistribution, syntactic shape of sequents, and the formulation of the rules of the calculus accordingly. What is more important, under this modification, all the proofs remain correct. 
We will outline the resulting formalism in the following subsection.

\subsection{Positive fragments}\label{ssec:positivefragments}
We will see how the machinery described in this paper, modulo one alternation concerning arities, can capture positive (i.e.\ $\neg$-free) fragments of finitary Moss' logic in $\Set$, as presented in~\cite{bpv13}. All functors in this section, including the coalgebra functors, are assumed to be finitary.

\begin{asm}\label{asm:fragments-syntax}
Solely for the purpose of this subsection, we assume that the arities of the conjunction and disjunction in $\Pos$ are given by the finitary convex subset functor $\PP^c_\omega$. 
\end{asm}

Let $D: \Set \to \Pos$ be the discrete functor. 
We consider a finitary coalgebra $\Pos$ functor $T'$ to be the canonical extension --- the posetification ---  of a fixed finitary standard $\Set$ functor $T$ preserving weak pull-backs. $T'$ is defined as a completion w.r.t. $\Pos$-enriched colimits~\cite{balan+kurz}, and it preserves exact squares. $T'$ is an extension of $T$, meaning that $T' D \cong D T$. In particular, $T'$ applied to a discrete poset yields a discrete poset. Observe that posetifications are self-dual on discrete posets:
\[
T'^\partial DX = (T'(DX)^\op)^\op \cong (T'DX)^\op \cong (DTX)^\op \cong DTX \cong T'DX.
\]
Therefore we can use $T'$ on the syntax side as the arity of both nabla and delta modalities, as long as the language is discrete. 

\paragraph{Language.} We fix a set of propositional atoms $\At^\Set$. Therefore $D\At^\Set$  can be used as a discrete poset $\At^{\Pos}$ of propositional atoms. 
The definition of the finitary Moss' language in $\Set$~\cite[5.1]{kkv12}, namely its $\neg$-free fragment which we denote $\Lang^{\Set}$, can be seen as an algebra for $\PP_\omega+\PP_\omega+T+ T$,
free on $\At^\Set$ (computed in $\Set$). The four components correspond to arities of conjunction, disjunction, nabla, and the delta modality. We aim at a definition of a discrete language $\Lang^\Pos$ so that $D\Lang^{\Set} \cong  \Lang^{\Pos}$. For conjunction $\bigwedge: \PP_\omega\Lang^\Set \to \Lang^\Set$, we aim at its $\Pos$ counterpart being $D\bigwedge: \PP^c_\omega D\Lang^\Set \to D\Lang^\Set$, similarly for disjunction. For the nabla modality $\nabla_T: T\Lang^\Set \to \Lang^\Set$  we aim at its $\Pos$ counterpart being $\nabla_{T'}= D\nabla_T: T'D\Lang^\Set \to D \Lang^\Set$, similarly for delta.
We therefore define the language $\Lang^{\Pos}$ to be an algebra for $\PP^c_\omega+\PP^c_\omega+T'+ T'$, free on $\At^\Pos$ in $\Pos$. Observe that $\Lang^{\Pos}$ is indeed a discrete poset, because for two convex subsets of a discrete poset, $u\leq^{EM}v$ implies $u = v$, and $T'$ applied to a discrete poset yields a discrete poset.  It follows that indeed $D\Lang^{\Set} \cong  \Lang^{\Pos}$.

Sequents in $\Set$~\cite{bpv13} are pairs of finite subsets of $\Lang^\Set$ written as $\phi\Rightarrow \psi$. For the purpose of this subsection only, sequents in $\Pos$ are pairs of finite convex subsets of $\Lang^\Pos$. Translation of the above sequent will be written as $D\phi\Rightarrow D\psi$. This notation is justified by the following paragraph:

\paragraph{Elements and bases.} Assume that $\alpha$ is an element of $T\Lang^\Set$: $\monoInline{\alpha}{\One}{T\Lang^\Set}$. Applying $D$ to it yields an element of $T'\Lang^\Pos$:  $\monoInline{D\alpha}{D\One}{DT\Lang^\Set\cong T'\Lang^\Pos}$. Similarly for $\phi,\psi$ in $\PP_\omega\Lang^\Set$ we have $D\phi,D\psi$ in $\PP^c_\omega\Lang^\Pos$, and for $\Phi,\Psi$ in $\PP_\omega\PP_\omega\Lang^\Set$ we have $D\Phi,D\Psi$ in $\PP^c_\omega\PP^c_\omega\Lang^\Pos$. Here we again abuse the notation slightly and denote by $\alpha$ or $D\alpha$ both the inclusion map and its image.
Assume the base of $\alpha$ in $T\Lang^\Set$ is given by $\monoInline{\base^T(\alpha)}{Z}{\Lang^\Set}$. Then we obtain the base of $D\alpha$ as $\monoInline{\base^{T'}(D\alpha)=D\base^T(\alpha)}{DZ}{\Lang^\Pos}$. 

\paragraph{Relations and their lifting.}
For a relation $R$ in $\Rel(\Set)$
\[
\xymatrix@1{
R:
Y
\ar[0,1]|-{\object @{/}}
&
X
}
\]
we obtain a monotone relation $DR$ in $\Rel(\Pos)$ by applying the functor $D$ to the span corresponding to $R$ to yield:
\[
\xymatrix@1{
DR:
DY
\ar[0,1]|-{\object @{/}}
&
DX.
}
\]
The lifting of $R$ by $T$ in $\Rel(\Set)$ is given by the composite
\[
\xymatrix@1{
\ol{T}R:
TX
\ar[0,2]|-{\object @{/}}^(.6){(Tp_1)^\diamond}
&
&
TR
\ar[0,2]|-{\object @{/}}^{(Tp_0)_\diamond}
&
&
TY.
}
\]
Applying the functor $D$ to its underlying span we obtain the $T'$ lifting of $DR$
\[
\xymatrix@1{
\ol{T'}DR:
T'DX\cong DTX
\ar[0,2]|-{\object @{/}}^(.6){(T'Dp_1)^\diamond}
&
&
T'DR\cong DTR
\ar[0,2]|-{\object @{/}}^{(T'Dp_0)_\diamond}
&
&
T'DY\cong DTY.
}
\]
From the above diagrams one can conclude the following:
\begin{equation}\label{eq:fragment-lifting}
\ol{T}R (\beta,\alpha) \ \mbox{ iff }\ \ol{T'}DR(D\beta, D\alpha).
\end{equation}

\paragraph{Coalgebras and valuations.} Given a $T$-coalgebra $c: X \to TX$ in $\Set$, we can apply the functor $D$ to it to obtain the $T'$- coalgebra $Dc: DX \to T'DX$ in $\Pos$. Given a valuation of propositional variables on the coalgebra $c$~\cite[Definition 3.5.]{bpv13}, we can see it as a relation
$
\xymatrix@1{\Vdash:\At^\Set
\ar[0,1]|-{\object @{/}}
&
X
}
$, 
we obtain a corresponding valuation on $Dc$ as  
$
\xymatrix@1{D{\Vdash}:\At^\Pos
\ar[0,1]|-{\object @{/}}
&
(DX)^\op
}
$.
Observing carefully how the connectives $\bigwedge,\bigvee$ and $\nabla,\Delta$ are interpreted in~\cite[Definition 3.5.]{bpv13} and in this paper in Subsections~\ref{ssec:semant} and~\ref{ssec:delta}, we observe that the correspondence carries out to link  $
\xymatrix@1{\Vdash:\Lang^\Set
\ar[0,1]|-{\object @{/}}
&
X
}
$ and 
$
\xymatrix@1{\mathrel{D{\Vdash}}:\Lang^\Pos
\ar[0,1]|-{\object @{/}}
&
(DX)^\op
}
$
so that we can prove by induction that for each sequent
\begin{equation}\label{eq:fragments-coalgebras-equality}
c,x \Vdash \phi\Rightarrow\psi \ \mbox{ iff }\ Dc,x \mathrel{D{\Vdash}} D\phi\Rightarrow D\psi.
\end{equation}
(cf.~\cite[Definition 3.18.]{bpv13}, and a definition in Subsection~\ref{ssec:proof-system} of a valid sequent.)
We spell out the case for the nabla modality (the other cases are simpler or similar):
\begin{align*}
    c,x \Vdash \nabla_T\alpha \ &\mbox{ iff }\ c(x)(\ol{T}\Vdash)\alpha \\ &\mbox{ iff }\ Dc(x)(\ol{T'}D{\Vdash})D\alpha \ \mbox{ by (\ref{eq:fragment-lifting})}\\
    &\mbox{ iff }\ Dc,x \mathrel{D{\Vdash}} \nabla_T'(D\alpha).
\end{align*}
Before we go further, we sum up the correspondence between the $\Set$-based setting of the finitary Moss' logic of~\cite{bpv13} and the $\Pos$-based setting of the finitary Moss' logic of this paper (with the arities of conjunction and disjunction adjusted) in the following table:
\medskip
\begin{center}
\begin{tabular}{|l|l|}
\hline
   $T: \Set \to \Set$  & $T': \Pos \to \Pos$ \\
   $\Lang^\Set$ & $\Lang^\Pos$ \\
    $\alpha$ in $T\Lang^\Set$ & $D\alpha$ in  $T'\Lang^\Pos$\\
    $\phi,\psi$ in $\PP_\omega\Lang^\Set$ &  $D\phi,D\psi$ in $\PP^c_\omega\Lang^\Pos$\\
    $\Phi,\Psi$ in $T\PP_\omega\Lang^\Set$ &  $D\Phi,D\Psi$ in $T'\PP^c_\omega\Lang^\Pos$\\
    $\base^T(\alpha)$ & $\base^{T'}(D\alpha)$ \\
    $\bigwedge: \PP_\omega\Lang^\Set \to \Lang^\Set$ & $D\bigwedge: \PP^c_\omega\Lang^\Pos \to \Lang^\Pos$\\
    $\bigvee: \PP_\omega\Lang^\Set \to \Lang^\Set$ & $D\bigvee: \PP^c_\omega\Lang^\Pos \to \Lang^\Pos$\\
    $c: X\to TX$ & $Dc: DX\to T'DX$ \\
    $\xymatrix@1{\Vdash:\Lang^\Set
\ar[0,1]|-{\object @{/}}
&
X
}$ & $\xymatrix@1{D{\Vdash}:\Lang^\Pos
\ar[0,1]|-{\object @{/}}
&
(DX)^\op
}$ \\
    $
\xymatrix@1{
\in:
\PP_\omega \Lang^\Set
\ar[0,1]|-{\object @{/}}
&
\Lang^\Set
}
$ & $
\xymatrix@1{
D{\in}:
\PP^c_\omega \Lang^\Pos
\ar[0,1]|-{\object @{/}}
&
\Lang^\Pos
}
$\\
$\ol{T}R$ & $\ol{T'}DR$\\
    \hline
\end{tabular}
\end{center}
\paragraph{Redistributions.} 
First let us cover the syntactic constructions behind the $\nabla$-r and $\Delta$-l rules. By~\cite[Definition 3.14]{bpv13}, the collections $L_T(\alpha)$ and $R_T(\beta)$ are defined as follows:
\begin{align*}
    L_T(\alpha) &:=   \{ (T\bigwedge)\Phi\ |\ \Phi\in T\base^T(\alpha); \neg (\alpha \mathrel{\ol{T}{\notin}} \Phi) \}    \\
    R_T(\beta) &:=   \{ (T\bigvee)\Psi\ |\ \Psi\in T\base^T(\beta) ; \neg (\beta \mathrel{\ol{T}{\notin}} \Psi) \} 
\end{align*}
Transferring this to $\Pos$, we obtain
\begin{align*}
   D[L_T(\alpha)] &:=   \{ (T'D\bigwedge)D\Phi\ |\ D\Phi\in T'\base^{T'}(D\alpha); \neg (D\alpha \mathrel{\ol{T'}D{\notin}} D\Phi) \}    \\
  D[R_T(\beta)] &:=   \{ (T'D\bigvee)D\Psi\ |\ D\Psi\in T'\base^{T'}(D\beta) ; \neg (D\beta \mathrel{\ol{T'}D{\notin}} D\Psi) \} 
\end{align*}
Comparing this with Definition~\ref{def:collections_L_R}, using the above table, we see that this provides us with the right definition when the arities of conjunction and disjunction are altered to be $\PP^c_\omega$, putting:
\begin{align}\label{eq:fragments-RL}
 D\Phi \in L_{T'}(D\alpha) \ &\mbox{ iff } \ (T\bigwedge)\Phi \in L_T(\alpha)\\
 D\Psi\in R_{T'}(D\beta) \ &\mbox{ iff } \ (T\bigvee)\Psi \in R_T(\beta)
\end{align}
Turning to redistributions, from~\cite[Definition 3.20.]{bpv13} of slim redistribution, and its generalization for two-sided sequents explained on p.~44 therein, we can extract the following: Let $(A,B)$ be an element of $P_\omega T\Lang^\Set \times P_\omega T\Lang^\Set$. A $\Phi$ in $T(\PP_\omega \base^T(A) \times \PP_\omega \base^T(B))$ is a slim redistribution of $(A,B)$ iff 
\begin{align*}
    (\forall \alpha \in A) \  \alpha \mathrel{\ol{T}{\in}} (Tp_0)\Phi \\
    (\forall \beta \in B)\    \beta \mathrel{\ol{T}{\in}} (Tp_1)\Phi,
\end{align*}
where $p_0,p_1$ are the projections of the product $\PP_\omega \base^T(A) \times \PP_\omega \base^T(B)$. Transferring to $\Pos$ we obtain that for $(DA,DB)$ in $P^c_\omega T'\Lang^\Pos \times P^c_\omega T'\Lang^\Pos$ and $\Phi$ a slim redistribution as above, $D\Phi$ in $T'(\PP^c_\omega \base^{T'}(DA) \times \PP^c_\omega \base^{T'}(DB))$ satisfies
\begin{align*}
    (\forall D\alpha \in DA) \  D\alpha \mathrel{\ol{T'}D{\in}} (T'Dp_0)\Phi, \\
    (\forall D\beta \in DB)\    D\beta \mathrel{\ol{T'}D{\in}} (T'Dp_1)\Phi.
\end{align*}
Comparing this with Definition~\ref{def:rd}, using the table above, we see again that this fits with $D\Phi$ being a redistribution of $(DA,DB)$, modulo the arities of conjunction and disjunction being altered to be $\PP^c_\omega$:
\begin{equation}\label{eq:fragments-rd}
 \Phi \in \srd (A,B) \ \mbox{ iff }\ D\Phi \in \rd (DA,DB).
\end{equation}
\paragraph{Rules of the calculus.} Using the transfer machinery described above, we can translate proofs of positive sequents from the two-sided calculus $G2_T$ of~\cite[Definition 5.1.]{bpv13} to the calculus $G^{T'}_{\nabla\Delta}$ of this paper (to be precise, to the calculus we obtain by systematically replacing the arities of conjunctions and disjunctions in $G^{T'}_{\nabla\Delta}$ with $\PP^c_\omega$). We can do so inductively, rule by rule. We will spell out the cases for the modal rules, and leave the propositional part of the calculus to be verified by the reader (propositional part of $G2_T$ can be found in~\cite[Figure 1.]{bpv13}, and that of $G^{T'}_{\nabla\Delta}$ in Definition~\ref{def:calculus_prop}).

The following two rules of $G2_T$ are those of~\cite[Definition 5.1.]{bpv13}, adapted to our current notation:
\begin{equation*}
\Inference{\{\phi\Rightarrow\Delta^T(T\bigvee)\Psi,\psi\mid(T\bigvee)\Psi\mbox{ in }R_T(\alpha) \}}
          {\phi\Rightarrow\nabla^T\alpha,\psi}
          {\text{$\nabla^T$-r}}
\end{equation*}
\begin{equation*}
\Inference{\{\phi,\nabla^T(T\bigwedge)\Phi\Rightarrow\psi\mid(T\bigwedge)\Phi\mbox{ in }L_T(\beta)\}}
          {\phi,\Delta^T\beta\Rightarrow\psi}
          {\text{$\Delta^T$-l}}
\end{equation*}
They translate into the following instances of the two rules of the modified (as for the arities of conjunctions and disjunctions) $G^{T'}_{\nabla\Delta}$ from Definition~\ref{def:modalrules}
\begin{equation*}
\Inference{\{D\phi\Rightarrow\Delta^{T'}(T'D\bigvee)D\Psi,D\psi\mid D\Psi\mbox{ in }R_{T'}(D\alpha) \}}
          {D\phi\Rightarrow\nabla^{T'} D\alpha,D\psi}
          {\text{$\nabla^{T'}$-r}}
\end{equation*}

\begin{equation*}
\Inference{\{D\phi,\nabla^{T'}(T'D\bigwedge)D\Phi\Rightarrow D\psi\mid D\Phi\mbox{ in }L_{T'}(D\beta)\}}
          {D\phi,\Delta^{T'}D\beta\Rightarrow D\psi}
          {\text{$\Delta^{T'}$-l}}
\end{equation*}
The following rule of $G2_T$ is the one of~\cite[Definition 5.1.]{bpv13}, adapted to our current notation (in particular, $p_0(z^\Phi)$ is originally denoted by $A_L^{\Phi}$ and $p_1(z^\Phi)$ is denoted by $B_R^{\Phi}$ and the notation explained on p.~47 (18).)
\[
\Inference{
           \{
           p_0(z^\Phi)
           \Rightarrow
           p_1(z^\Phi)
           \mid
           \Phi\mbox{ in }\srd(A,B)
           \}
           }
          {
          \{
          \nabla_T\alpha\mid \alpha\in A
          \}
          \Rightarrow
          \{
          \Delta_T\beta\mid\beta\in B
          \}
          }
          {\text{$T(\nabla\Delta)$}}
          \text{$\forall \Phi.\ z^\Phi \in \base^T(\Phi)$}
        \]
It translates into the following instance of the $T'(\nabla\Delta)$ rule of Definition~\ref{def:therule}
\[
\Inference{
           \{
           Dp_0(Dz^\Phi)
           \Rightarrow
           Dp_1(Dz^\Phi)
           \mid
           D\Phi\mbox{ in }\rd(DA,DB)
           \}
           }
          {
          \{
          \nabla_{T'}D\alpha\mid D\alpha\in DA
          \}
          \Rightarrow
          \{
          \Delta_{T'}D\beta\mid D\beta\in DB
          \}
          }
          {\text{$T'(\nabla\Delta)$}}
          \text{$\forall \Phi.\ (Dz)^\Phi \in \base^{T'}(D\Phi)$}
        \]
To sum up, using the above, we can state the following theorem
\begin{thm}\label{thm:fragments}
Let $\phi,\psi$ be finite subsets of $\Lang^\Set$ (i.e., the $\neg$-free fragment of the classical Moss' coalgebraic language). 
  \[
\vdash_{G2_T}\phi\Rightarrow\psi\ \mbox{ iff } \vdash_{G^{T'}_{\nabla\Delta}}D\phi\Rightarrow D\psi.
  \]
\end{thm}
\begin{proof}
The left-right direction is proven by induction on the proof in $G2_T$, translating it step-by-step using the discrete functor as described above.

The right-left direction is proven by contraposition, using the completeness of the two calculi. Assume that $\nvdash_{G2_T}\phi\Rightarrow\psi$. Then there is a coalgebra $c: X \to TX$, a valuation of $\Lang^\Set$, and a state $x$ so that 
  \[
c,x \nVdash \phi\Rightarrow\psi.
  \]
Then, by~\ref{eq:fragments-coalgebras-equality}, the coalgebra $Dc: DX \to T'DX$, the translated valuation of $\Lang^\Pos$, and a state $x$ refute the sequent
  \[
Dc,x \nVdash D\phi\Rightarrow D\psi. \qedhere
\]
\end{proof}
\begin{exa}\label{ex:positive-Kripke}
Consider the finitary convex powerset functor $\PP^c_\omega$ as the coalgebra functor. $\PP^c_\omega$ is the posetification of the finitary powerset functor $\PP_\omega$, whose coalgebras in $\Set$ correspond to image-finite Kripke frames. Thus, adapting the arities of conjuctions and disjunctions in the calculus $G^{\PP^c_\omega}_{\nabla\Delta}$ captures the positive fragment of the finitary Moss' logic over image-finite Kripke frames (whose complete proof theory is provided by the calculus $G2_{\PP_\omega}$ of~\cite{bpv13}).
\end{exa}

\bibliographystyle{alphaurl}
\bibliography{bibliography}

\end{document}